\documentclass[a4paper,UKenglish]{article}
\pdfoutput=1
\usepackage[T1]{fontenc}
\usepackage[utf8]{inputenc}
\usepackage[english]{babel}

\usepackage{amsmath, amssymb}
\usepackage{listings}
\usepackage{cmll}
\usepackage{bbold}
\usepackage{stmaryrd}
\usepackage{tikz}
\usetikzlibrary{trees,positioning,arrows,shapes, calc, fit, snakes, decorations.shapes}

\usepackage[left=3cm,right=3cm, bottom=4cm]{geometry}

\usepackage{shuffle}
\DeclareFontFamily{U}{shuffle}{}
\DeclareFontShape{U}{shuffle}{m}{n}{ <-8>shuffle7 <8->shuffle10}{}

\usepackage[hidelinks]{hyperref}

\usepackage{amsthm}
\theoremstyle{plain}
\newtheorem{theorem}{Theorem}
\newtheorem{proposition}[theorem]{Proposition}
\newtheorem{lemma}[theorem]{Lemma}
\newtheorem{corollary}[theorem]{Corollary}
\theoremstyle{definition}
\newtheorem{definition}[theorem]{Definition}
\newtheorem{example}[theorem]{Example}
\theoremstyle{remark}
\newtheorem*{remark}{Remark}

\newcommand{\daimon}{\maltese}
\newcommand{\design}[1]{{\mathfrak{#1}}}
\newcommand{\interpret}[1]{\llbracket#1\rrbracket}
\newcommand{\normalisation}[1]{(\![#1]\!)}
\newcommand{\antiview}[1]{{\llcorner#1\lrcorner}}
\newcommand{\defined}[1]{\textbf{#1}}
\newcommand{\imp}[1]{{\em #1}}
\newcommand{\visit}[1]{V_{\beh{#1}}}
\newcommand{\beh}[1]{{\mathbf{#1}}}
\newcommand{\triv}[1]{\langle#1\rangle}
\newcommand{\trivv}[2]{\langle#1\rangle_{#2}}
\newcommand{\gbar}{~~~\boldsymbol |~~~}
\newcommand{\setst}[2]{\{ #1 ~\boldsymbol |~ #2\}}
\newcommand{\posdes}[3]{#1|\overline{#2}\langle #3 \rangle}
\newcommand{\posdesblue}[3]{#1|\overline{#2}\textcolor{blue}{\langle} #3 \textcolor{blue}{\rangle}}
\newcommand{\posdesdots}[4]{#1|\overline{#2}\langle #3, \dots, #4 \rangle}
\newcommand{\negdes}[3]{\sum_{#1\in \mathcal S} #1(#2).#3}
\newcommand{\negdesdots}[4]{\sum_{#1\in \mathcal S} #1(#2, \dots, #3).#4}
\newcommand{\interseq}[2]{\langle #1 \leftarrow #2\rangle}
\newcommand{\viewseq}[1]{{\mathbb{#1}}}
\newcommand{\cut}[2]{\mathfrak{Cut}_{#1|#2}}
\newcommand{\completed}[1]{\fullview{#1}^c}
\newcommand{\antishuffle}{\text{\rotatebox[origin=c]{180}{$\shuffle$}}}
\newcommand{\vect}[1]{\overrightarrow{#1}}
\newcommand{\symshneg}{\blacktriangle}
\newcommand{\symshpos}{\blacktriangledown}
\newcommand{\sympar}{\wp}
\newcommand{\symtensor}{\bullet}
\newcommand{\symavec}{\pi}
\newcommand{\symplus}{\iota}

\newcommand{\proj}[2]{
\mathinner{\mathchoice
{#1\!\!\upharpoonright\!\!#2}
{#1\!\!\upharpoonright\!\!#2}
{#1\mkern1mu\upharpoonright\mkern1mu #2}
{#1\!\!\upharpoonright\!\!#2}
}
}

\DeclareFontFamily{OT1}{pzc}{}
\DeclareFontShape{OT1}{pzc}{m}{it}{<-> [1.1] pzcmi8t}{} 
\DeclareMathAlphabet{\mathpzc}{OT1}{pzc}{m}{it}
\newcommand{\pathLL}[1]{\mathpzc{#1}}

\newlength{\Viewheight}
\newlength{\ulcornerheight}
\newcommand{\view}[1]{
\settoheight{\Viewheight}{$#1$}
\settoheight{\ulcornerheight}{$\ulcorner$}
\addtolength{\Viewheight}{-\ulcornerheight}
\addtolength{\Viewheight}{2pt}
\raisebox{\Viewheight}{$\ulcorner$}{#1}\raisebox{\Viewheight}{$\urcorner$}}

\newlength{\Fullviewheight}
\newcommand{\fullview}[1]{
\settoheight{\Fullviewheight}{$#1$}
\settoheight{\ulcornerheight}{$\ulcorner$}
\addtolength{\Fullviewheight}{-\ulcornerheight}
\addtolength{\Fullviewheight}{2pt}
\raisebox{\Fullviewheight}{$\ulcorner\mkern-6mu\ulcorner\mkern-2mu$}{#1}\raisebox{\Fullviewheight}{$\mkern-2mu\urcorner\mkern-6mu\urcorner$}}

\newlength{\dualwidth}
\newlength{\dualheight}
\newcommand{\dual}[2][1]{
\settowidth{\dualwidth}{$#2$}
\settoheight{\dualheight}{$#2$}
\makebox[\dualwidth][c]{\mbox{\rule{0cm}{#1\dualheight}$\Widetilde[#1]{#2}$}}
}

\makeatletter
\newlength{\Widetildeheight}
\newlength{\Widetildewidth}
\newcommand{\Widetildestretch}{1}
\newcommand*\Widetilde[2][\Widetildestretch]{
        \begingroup
        \mathchoice{\Widetilde@helper{#1}{#2}{\displaystyle}{\textfont}}
                   {\Widetilde@helper{#1}{#2}{\textstyle}{\textfont}}
                   {\Widetilde@helper{#1}{#2}{\scriptstyle}{\scriptfont}}
                   {\Widetilde@helper{#1}{#2}{\scriptscriptstyle}{\scriptscriptfont}}
        \endgroup
}

\newcommand*\Widetilde@helper[4]{
\settowidth{\Widetildewidth}{$#2$}
\settoheight{\Widetildeheight}{$#2$}
\setlength{\Widetildeheight}{#1\Widetildeheight}
\rlap{\raisebox{\Widetildeheight}{$\resizebox{\Widetildewidth}{\height}{$\sim$}$}}{#2}
}
\makeatother

\title{Inductive and Functional Types in Ludics\footnote{Extended version of the paper accepted for publication in CSL 2017.}}

\author{Alice Pavaux \\ \\ Université Paris 13, Sorbonne Paris Cité, LIPN, CNRS, UMR 7030}

\begin{document}

\maketitle

\begin{abstract}
  Ludics is a logical framework in which types/formulas are modelled by sets of terms with the same computational behaviour. This paper investigates the representation of inductive data types and functional types in ludics. We study their structure following a game semantics approach. Inductive types are interpreted as least fixed points, and we prove an internal completeness result giving an explicit construction for such fixed points. The interactive properties of the ludics interpretation of inductive and functional types are then studied. In particular, we identify which higher-order functions types fail to satisfy type safety, and we give a computational explanation.
\end{abstract}

\section{Introduction}

\subsection{Context and Contributions}

\subparagraph{Context} \hspace{0cm} \imp{Ludics} was introduced by Girard \cite{Girard1} as a variant of \imp{game semantics} with interactive types. Game Semantics has successfully provided fully abstract models for various logical systems and programming languages, among which PCF~\cite{HO}. Although very close to Hyland--Ong (HO) games, ludics reverses the approach: in HO games one defines first the interpretation of a type (an arena) before giving the interpretation for the terms of that type (the strategies), while in ludics the interpretation of terms (the \imp{designs}) is primitive and the types (the \imp{behaviours}) are recovered dynamically as well-behaved sets of terms. This approach to types is similar to what exists in realisability \cite{Krivine} or geometry of interaction \cite{Girard2}.

The motivation for such a framework was to reconstruct logic around the dynamics of proofs. Girard provides a ludics model for (a polarised version of) multiplicative-additive linear logic (MALL); a key role in his interpretation of logical connectives is played by the \imp{internal completeness} results, which allow for a direct description of the behaviours' content. As most behaviours are not the interpretation of MALL formulas, an interesting question, raised from the beginning of ludics, is whether these remaining behaviours can give a logical counterpart to computational phenomena. In particular, data and functions~\cite{Terui, Sironi}, and also fixed points~\cite{BDS} have been studied in the setting of ludics. The present work follows this line of research.

Real life (functional) programs usually deal with data, functions over it, functions over functions, etc. \imp{Data types} allow one to present information in a structured way. Some data types are defined \imp{inductively}, for example:
\begin{lstlisting}[caption={Example of inductive types in OCaml}]
> type nat = Zero | Succ of nat ;;
> type 'a list = Nil | Cons of 'a * 'a list ;;
> type 'a tree = Empty | Node of 'a * ('a tree) list ;;
\end{lstlisting}
Upon this basis we can consider \imp{functional types}, which are either first-order -- from data to data -- or higher-order -- i.e., taking functions as arguments or returning functions as a result. This article aims at interpreting constructively the (potentially inductive) data types and the (potentially higher-order) functional types as behaviours of ludics, so as to study their structural properties. Inductive types are defined as (least) \imp{fixed points}. As pointed out by Baelde, Doumane and Saurin~\cite{BDS}, the fact that ludics puts the most constraints on the formation of terms instead of types, conversely to game semantics, makes it a more natural setting for the interpretation of fixed points than HO games \cite{Clairambault}.

\subparagraph{Contributions}
The main contributions of this article are the following:
\begin{itemize}
\item We prove that internal completeness holds for infinite unions of behaviours satisfying particular conditions (Theorem~\ref{thm_union_beh}), leading to an explicit construction of the least fixed points in ludics (Proposition~\ref{prop_reg_union}).
\item Inductive and functional types are interpreted as behaviours, and we prove that such behaviours are \imp{regular} (Corollary~\ref{coro_data_reg} and Proposition~\ref{prop_func_quasi_pure}). Regularity (that we discuss more in \textsection~\ref{sub_tools}) is a property that could be used to characterise the behaviours corresponding to $\mu$MALL formulas \cite{Baelde,BDS} -- i.e., MALL with fixed points.
\item We show that a functional behaviour fails to satisfy \imp{purity}, a property ensuring the safety of all possible executions (further explained in \textsection~\ref{sub_tools}), if and only if it is higher order and takes functions as argument (Proposition~\ref{prop_main}); this is typically the case of $(\beh A \multimap \beh B) \multimap \beh C$. In \textsection~\ref{ex_discuss} we discuss the computational meaning of this result.
\end{itemize}
The present work is conducted in the term-calculus reformulation of ludics by Terui \cite{Terui} restricted to the linear part -- the idea is that programs call each argument at most once.

\subparagraph{Related Work}
 The starting point for our study of inductive types as fixed points in ludics is the work by Baelde, Doumane and Saurin \cite{BDS}. In their article, they provide a ludics model for $\mu$MALL, a variant of multiplicative-additive linear logic with least and greatest fixed points. The existence of fixed points in ludics is ensured by Knaster-Tarski theorem, but this approach does not provide an explicit way to construct the fixed points; we will consider Kleene fixed point theorem instead. Let us also mention the work of Melliès and Vouillon \cite{MV} which introduces a realisability model for recursive (i.e., inductive and coinductive) polymorphic types.

 The representation of both data and functions in ludics has been studied previously. Terui \cite{Terui} proposes to encode them as designs in order to express computability properties in ludics, but data and functions are not considered at the level of behaviours. Sironi \cite{Sironi} describes the behaviours corresponding to some data types: integers, lists, records, etc. as well as first-order function types; our approach generalises hers by considering generic data types and also higher order functions types.

 \subsection{Background} \label{sub_tools}

\subparagraph{Behaviours and Internal Completeness} A behaviour $\beh B$ is a set of designs which pass the same set of tests $\beh B^\perp$, where tests are also designs. $\beh B^\perp$ is called the \imp{orthogonal} of $\beh B$, and behaviours are closed under bi-orthogonal: $\beh B^{\perp\perp} = \beh B$. New behaviours can be formed upon others using various constructors. In this process, internal completeness, which can be seen as a built-in notion of observational equivalence, ensures that two agents reacting the same way to any test are actually equal. From a technical point of view, this means that it is not necessary to apply a $\perp\perp$-closure for the sets constructed to be behaviours.

\subparagraph{Paths: Ludics as Game Semantics} This paper makes the most of the resemblance between ludics and HO game semantics. The connections between them have been investigated in many pieces of work \cite{BF, Faggian, FQ1} where designs are described as (innocent) strategies, i.e., in terms of the traces of their possible interactions. Following this idea, Fouqueré and Quatrini define \imp{paths} \cite{FQ1}, corresponding to legal plays in HO games, and they characterise a behaviour by its set of \imp{visitable paths}. This is the approach we follow. The definitions of regularity and purity rely on paths, since they are properties of the possible interactions of a behaviour.

\subparagraph{Regularity: Towards a Characterisation of {\boldmath$\mu$}MALL?} Our proof that internal completeness holds for an infinite union of increasingly large behaviours (Theorem~\ref{thm_union_beh}) relies in particular on the additional hypothesis of regularity for these behaviours. Intuitively, a behaviour $\beh B$ is regular if every path in a design of $\beh B$ is realised by interacting with a design of $\beh B^\perp$, and vice versa. This property is not actually ad hoc: it was introduced by Fouqueré and Quatrini~\cite{FQ2} to characterise the denotations of MALL formulas as being precisely the regular behaviours satisfying an additional finiteness condition. In this direction, our intuition is that -- forgetting about finiteness -- regularity captures the behaviours corresponding to formulas of $\mu$MALL. Although such a characterisation is not yet achieved, we provide a first step by showing that the \imp{data patterns}, a subset of positive $\mu$MALL formulas, yield only regular behaviours (Proposition~\ref{prop_interp_reg}).

\subparagraph{Purity: Type Safety} Ludics has a special feature for termination which is not present in game semantics: the \imp{daimon} $\daimon$. On a computational point of view, the daimon is commonly interpreted as an error, an exception raised at run-time causing the program to stop (see for example the notes of Curien~\cite{Curien}). Thinking of Ludics as a programming language, we would like to guarantee \imp{type safety}, that is, ensure that ``well typed programs cannot go wrong''~\cite{Milner}. This is the purpose of purity, a property of behaviours: in a pure behaviour, maximal interaction traces are $\daimon$-free, in other words whenever the interaction stops with $\daimon$ it is actually possible to ``ask for more'' and continue the computation. Introduced by Sironi~\cite{Sironi} (and called \imp{principality} in her work), this property is related to the notions of \imp{winning} designs~\cite{Girard1} and \imp{pure} designs~\cite{Terui}, but at the level of a behaviour. As expected, data types are pure (Corollary~\ref{coro_data_pur}), but not always functional types are; we identify the precise cases where impurity arises (Proposition~\ref{prop_main}), and explain why some types are not safe.

\subsection{Outline}

In Section~\ref{sec-designs} we present ludics and we state internal completeness for the logical connectives constructions.
In Section~\ref{sec-paths} we recall the notion of path, so as to define formally regularity and purity and prove their stability under the connectives.
Section~\ref{sec-induct} studies inductive data types, which we interpret as behaviours; Kleene theorem and internal completeness for infinite union allows us to give an explicit and direct construction for the least fixed point, with no need for bi-orthogonal closure; we deduce that data types are regular and pure.
Finally, in Section~\ref{sec-func}, we study functional types, showing in what case purity fails.

\section{Computational Ludics} \label{sec-designs}

This section introduces the ludics background necessary for the rest of the paper, in the formalism of Terui~\cite{Terui}. The \imp{designs} are the primary objects of ludics, corresponding to (polarised) proofs or programs in a Curry-Howard perspective. Cuts between designs can occur, and their reduction is called \imp{interaction}. The \imp{behaviours}, corresponding to the types or formulas of ludics, are then defined thanks to interaction. Compound behaviours can be formed with \imp{logical connectives} constructions which satisfy \imp{internal completeness}.

\subsection{Designs and Interaction} \label{sub-design}

Suppose given a set of variables $\mathcal V_0$ and a set $\mathcal S$, called \defined{signature}, equipped with an arity function $ar: \mathcal S \to \mathbb N$. Elements $a, b, \dots \in \mathcal S$ are called \defined{names}. A \defined{positive action} is either $\daimon$ (daimon), $\Omega$ (divergence), or $\overline a$ with $a \in \mathcal S$; a \defined{negative action} is $a(x_1, \dots, x_n)$ where $a \in \mathcal S$, $\mathrm{ar}(a)=n$ and $x_1, \dots, x_n \in \mathcal V_0$ distinct. An action is \defined{proper} if it is neither $\daimon$ nor $\Omega$.
\begin{definition}
  Positive and negative \defined{designs}\footnote{In the following, the symbols $\design d, \design e, \dots$ refer to designs of any polarity, while $\design p, \design q, \dots$ and $\design m, \design n, \dots$ are specifically for positive and negative designs respectively.} are coinductively defined by:
    \begin{align*}
      &\design p ~::=~ \daimon \gbar \Omega \gbar \posdesdots{x}{a}{\design n_1}{\design n_{\mathrm{ar}(a)}} \gbar \posdesdots{\design n_0}{a}{\design n_1}{\design n_{\mathrm{ar}(a)}} \\
      &\design n ~::=~ \textstyle\negdesdots{a}{x^a_1}{x^a_{\mathrm{ar}(a)}}{\design p_a}
    \end{align*}
\end{definition}

Positive designs play the same role as \imp{applications} in $\lambda$-calculus, and negative designs the role of \imp{abstractions}, where each name $a \in \mathcal S$ binds $\mathrm{ar}(a)$ variables.

Designs are considered up to $\alpha$-equivalence. We will often write $a(\vect x)$ (resp. $\overline a \langle \vect{\design n} \rangle$) instead of $a(x_1, \dots, x_n)$ (resp. $\overline a \langle \design n_1 \dots \design n_n \rangle$). Negative designs can be written as partial sums, for example $a(x, y).\design p + b().\design q$ instead of $a(x, y).\design p + b().\design q + \sum_{c \neq a, c \neq b} c(\vect{z^c}).\Omega$.

  Given a design $\design d$, the definitions of the \defined{free variables} of $\design d$, written $\mathrm{fv}(\design d)$, and the (capture-free) \defined{substitution} of $x$ by a negative design $\design n$ in $\design d$, written $\design d[\design n/x]$, can easily be inferred. The design $\design d$ is \defined{closed} if it is positive and it has no free variable. A \defined{subdesign} of $\design d$ is a subterm of $\design d$. A \defined{cut} in $\design d$ is a subdesign of $\design d$ of the form $\posdes{\design n_0}{a}{\vect{\design n}}$, and a design is \defined{cut-free} if it has no cut.

In the following, we distinguish a particular variable $x_0$, that cannot be bound. A positive design $\design p$ is \defined{atomic} if $\mathrm{fv}(\design p) \subseteq \{x_0\}$; a negative design $\design n$ is \defined{atomic} if $\mathrm{fv}(\design n) = \emptyset$.

A design is \defined{linear} if for every subdesign of the form $\posdes{x}{a}{\vect{\design n}}$ (resp. $\posdes{\design n_0}{a}{\vect{\design n}}$), the sets $\{x\}$, $\mathrm{fv}(\design n_1)$, \dots, $\mathrm{fv}(\design n_{\mathrm{ar}(a)})$ (resp. the sets $\mathrm{fv}(\design n_0)$, $\mathrm{fv}(\design n_1)$, \dots, $\mathrm{fv}(\design n_{\mathrm{ar}(a)})$) are pairwise disjoint. This article focuses on linearity, so in the following when writing ``design'' we mean ``linear design''.

\begin{definition}
  The \defined{interaction} corresponds to reduction steps applied on cuts:
  \[\textstyle\sum_{a \in \mathcal S} a(x^a_1, \dots, x^a_{\mathrm{ar}(a)}).\design p_a~|~ \overline b \langle \design n_1, \dots, \design n_k\rangle \hspace{.6cm} \leadsto \hspace{.6cm} \design p_b[\design n_1/x^b_1, \dots, \design n_k/x^b_k] \]
\end{definition}
We will later describe an interaction as a sequence of actions, a path (Definition~\ref{def-path}).

Let $\design p$ be a design, and let $\leadsto^*$ denote the reflexive transitive closure of $\leadsto$; if there exists a design $\design q$ which is neither a cut nor $\Omega$ and such that $\design p \leadsto^* \design q$, we write $\design p \Downarrow \design q$; otherwise we write $\design p \Uparrow$. The normal form of a design, defined below, exists and is unique \cite{Terui}.

\begin{definition} The \defined{normal form} of a design $\design d$, noted $\normalisation{\design d}$, is defined by:
\begin{align*}
  & \normalisation{\design p} = \daimon \mbox{\hspace{.3cm}~if~} \design p \Downarrow \daimon && \normalisation{\design p} = \posdesdots{x}{a}{\normalisation{\design n_1}}{\normalisation{\design n_n}} \mbox{\hspace{.3cm}~if~} \design p \Downarrow \posdesdots{x}{a}{\design n_1}{\design n_n} \\
  & \normalisation{\design p} = \Omega \mbox{\hspace{.3cm}~if~} \design p \Uparrow && \normalisation{\textstyle \negdes{a}{\vect{x^a}}{\design p_a}} = \textstyle \negdes{a}{\vect{x^a}}{\normalisation{\design p_a}}
\end{align*}
\end{definition}
Note that the normal form of a closed design is either $\daimon$ (convergence) or $\Omega$ (divergence). Orthogonality expresses the convergence of the interaction between two atomic designs, and behaviours are sets of designs closed by bi-orthogonal.

\begin{definition}
Two atomic designs $\design p$ and $\design n$ are \defined{orthogonal}, noted $\design p \perp \design n$, if $\normalisation{\design p[\design n/x_0]} = \daimon$.
\end{definition}
Given an atomic design $\design d$, define $\design d^\perp = \setst{\design e}{\design d \perp \design e}$; if $E$ is a set of atomic designs of same polarity, define $E^\perp = \setst{\design d}{\forall \design e \in E, \design d \perp \design e}$.

\begin{definition}    
  A set $\beh B$ of atomic designs of same polarity is a \defined{behaviour}\footnote{Symbols $\beh A, \beh B, \dots$ will designate behaviours of any polarity, while $\beh M, \beh N \dots$ and $\beh P, \beh Q, \dots$ will be for negative and positive behaviours respectively.} if $\beh B^{\perp\perp} = \beh B$. A behaviour is either positive or negative depending on the polarity of its designs.
\end{definition}
Behaviours could alternatively be defined as the orthogonal of a set $E$ of atomic designs of same polarity -- $E$ corresponds to a set of \imp{tests} or \imp{trials}. Indeed, $E^\perp$ is always a behaviour, and every behaviour $\beh B$ is of this form by taking $E = \beh B^\perp$.

The \imp{incarnation} of a behaviour $\beh B$ contains the cut-free designs of $\beh B$ whose actions are all visited during an interaction with a design in $\beh B^\perp$. Those correspond to the cut-free designs that are minimal for the \defined{stable ordering} $\sqsubseteq$, where $\design d' \sqsubseteq \design d$ if $\design d$ can be obtained from $\design d'$ by substituting positive subdesigns for some occurrences of $\Omega$.

\begin{definition} \label{def-incarn} Let $\beh B$ be a behaviour and $\design d \in \beh B$ cut-free.
  \begin{itemize}
  \item The \defined{incarnation} of $\design d$ in $\beh B$, written $|\design d|_{\beh B}$, is the smallest (for $\sqsubseteq$) cut-free design $\design d'$ such that $\design d' \sqsubseteq \design d$ and $\design d' \in \beh B$. If $|\design d|_{\beh B} = \design d$ we say that $\design d$ is \defined{incarnated} in $\beh B$.
  \item The \defined{incarnation} $|\beh B|$ of $\beh B$ is the set of the (cut-free) incarnated designs of $\beh B$.
  \end{itemize}
\end{definition}

\subsection{Logical Connectives}  \label{sub_connectives}

Behaviour constructors -- the \imp{logical connectives} -- can be applied so as to form compound behaviours. These connectives, coming from (polarised) linear logic, are used for interpreting formulas as behaviours, and will also indeed play the role of type constructors for the types of data and functions. In this subsection, after defining the connectives we consider, we state the \imp{internal completeness} theorem for these connectives.

Let us introduce some notations. In the rest of this article, suppose the signature $\mathcal S$ contains distinct unary names $\symshneg, \symavec_1, \symavec_2$ and a binary name $\sympar$, and write $\symshpos = \overline \symshneg, \symplus_1 = \overline \symavec_1, \symplus_2 = \overline \symavec_2$ and $\symtensor = \overline \sympar$. Given a behaviour $\beh B$ and $x$ fresh, define $\beh B^x = \setst{\design d[x/x_0]}{\design d \in \beh B}$; such a substitution operates a ``delocation'' with no repercussion on the behaviour's inherent properties. Given a $k$-ary name $a \in \mathcal S$, we write $\overline{a}\langle \beh N_1, \dots, \beh N_k \rangle$ or even $\overline{a}\langle \vect{\beh N} \rangle$ for $\setst{\posdes{x_0}{a}{\vect{\design n}}}{\design n_i \in \beh N_i}$, and write $a(\vect x).\beh P$ for $\setst{a(\vect x).\design p}{\design p \in \beh P}$. For a negative design $\design n = \negdes{a}{\vect{x^a}}{\design p_a}$ and a name $a \in \mathcal S$, we denote by $\proj{\design n}{a}$ the design $a(\vect{x^a}).\design p_a$ (that is $a(\vect{x^a}).\design p_a + \sum_{b \neq a} b(\vect{x^b}).\Omega$).

\begin{definition}[Logical connectives]
  \begin{align*}
  & \shpos \beh N = \symshpos \langle \beh N \rangle^{\perp\perp} && \mbox{ (\defined{positive shift}) } \\
  & \shneg \beh P = (\symshneg(x).\beh P^x)^{\perp\perp} \mbox{, with $x$ fresh} && \mbox{ (\defined{negative shift}) }\\
  & \beh M \oplus \beh N = (\symplus_1 \langle \beh M \rangle \cup \symplus_2 \langle \beh N \rangle)^{\perp\perp} && \mbox{ (\defined{plus}) }\\
  & \beh M \otimes \beh N = \symtensor \langle \beh M, \beh N \rangle^{\perp\perp} && \mbox{ (\defined{tensor}) }\\
  & \beh N \multimap \beh P = (\beh N \otimes \beh P^\perp)^\perp && \mbox{ (\defined{linear map}) }
  \end{align*}
\end{definition}

Our connectives $\shpos$, $\shneg$, $\oplus$ and $\otimes$ match exactly those defined by Terui \cite{Terui}, who also proves the following internal completeness theorem stating that connectives apply on behaviours in a constructive way -- there is no need to close by bi-orthogonal. For each connective, we present two versions of internal completeness: one concerned with the full behaviour, the other with the behaviour's incarnation.

\begin{theorem}[Internal completeness for connectives]\label{thm_intcomp_all}
  \begin{align*}
    &\shpos \beh N = \symshpos \langle \beh N \rangle \cup \{\daimon\} && |\shpos \beh N| = \symshpos \langle |\beh N| \rangle \cup \{\daimon\} \\
    &\shneg \beh P = \setst{\design n}{\proj{\design n}{\symshneg} \in \symshneg(x).\beh P^x} && |\shneg \beh P| = \symshneg(x).|\beh P^x| \\
    &\beh M \oplus \beh N = \symplus_1 \langle \beh M \rangle \cup \symplus_2 \langle \beh N \rangle \cup \{\daimon\} && |\beh M \oplus \beh N| = \symplus_1 \langle |\beh M| \rangle \cup \symplus_2 \langle |\beh N| \rangle \cup \{\daimon\} \\
    &\beh M \otimes \beh N = \symtensor \langle \beh M, \beh N \rangle \cup \{\daimon\} && |\beh M \otimes \beh N| = \symtensor \langle |\beh M|, |\beh N| \rangle \cup \{\daimon\}
  \end{align*}
\end{theorem}

\section{Paths and Interactive Properties of Behaviours} \label{sec-paths}

\imp{Paths} are sequences of actions recording the trace of a possible interaction. For a behaviour $\beh B$, we can consider the set of its \imp{visitable paths} by gathering all the paths corresponding to an interaction between a design of $\beh B$ and a design of $\beh B^\perp$. This notion is needed for defining \imp{regularity} and \imp{purity} and proving that those two properties of behaviours are stable under (some) connectives constructions.

\subsection{Paths} \label{sub-paths}

This subsection adapts the definitions of path and visitable path from \cite{FQ1} to the setting of computational ludics. In order to do so, we need first to recover \imp{location} in actions so as to consider sequences of actions.

Location is a primitive idea in Girard's ludics \cite{Girard1} in which the places of a design are identified with \imp{loci} or \imp{addresses}, but this concept is not visible in Terui's presentation of designs-as-terms. We overcome this by introducing actions with more information on location, which we call \imp{located actions}, and which are necessary to:
\begin{itemize}
\item represent cut-free designs as trees -- actually, forests -- in a satisfactory way,
\item define views and paths.
\end{itemize}

\begin{definition} \label{loc-ac}
  A \defined{located action}\footnote{Located actions will often be denoted by symbol $\kappa$, sometimes with its polarity: $\kappa^+$ or $\kappa^-$.} $\kappa$ is one of: $\daimon \gbar \posdesdots{x}{a}{x_1}{x_{\mathrm{ar}(a)}} \gbar a_x(x_1, \dots, x_{\mathrm{ar}(a)})$ \\
where in the last two cases (\defined{positive proper} and \defined{negative proper} respectively), $a \in \mathcal S$ is the \defined{name} of $\kappa$, the variables $x, x_1, \dots, x_{\mathrm{ar}(a)}$ are distinct, $x$ is the \defined{address} of $\kappa$ and $x_1, \dots, x_{\mathrm{ar}(a)}$ are the \defined{variables bound by} $\kappa$.
\end{definition}
In the following, ``action'' will always refer to a located action. Similarly to notations for designs, $\posdes{x}{a}{\vect x}$ stands for $\posdesdots{x}{a}{x_1}{x_n}$ and $a_x(\vect x)$ for $a_x(x_1, \dots, x_n)$.

\begin{example} \label{ex-tree}
  We show how cut-free designs can be represented as trees of located actions in this example. Let $a^2, b^2, c^1, d^0 \in \mathcal S$, where exponents stand for arities. The following design is represented by the tree of Fig.~\ref{fig-ex-path}.
  \[\design d = a(x_1,x_2).\textcolor{red}(\posdesblue{x_2}{b}{a(x_3, x_4).\daimon + c(y_1).\textcolor{orange}{(}\posdes{y_1}{d}{}\textcolor{orange}{)}\textcolor{blue}{,} c(y_2).\textcolor{green}{(}\posdes{x_1}{d}{}\textcolor{green}{)}}\textcolor{red})\]
\end{example}
  \begin{figure}
  \centering
    \begin{tikzpicture}[grow=up, level distance=1cm, sibling distance = 2.5cm]
      \draw[-]
      node (a) {$a_{x_0}(x_1, x_2)$}
      child{
        node (b) {$x_2|\overline b \langle z_1,z_2 \rangle$}
        child {
          node (c) {$c_{z_2}(y_2)$}
          child {
            node (d) {$x_1|\overline d\langle\rangle$}
          }
        }
        child	{
          node (c') {$c_{z_1}(y_1)$}
          child	{
            node (d') {$y_1|\overline d \langle\rangle$}
          }
        }
        child	{
          node (a') {$a_{z_1}(x_3,x_4)$}
          child	{
            node (dai) {$\daimon$}
          }
        }
      };
        \draw[->, violet, thick, rounded corners] ($(a)+(-.9,-.2)$) -- ($(b)+(-.9,0)$) -- ($(a')+(-.9,-.2)$) -- ($(dai)+(-.9,.3)$) ;
        \node[violet] at ($(a)+(-1.7,0)$) {a view};
        \draw[-, red, thick, rounded corners] ($(a)+(.9,-.2)$) -- ($(d')+(.9,.3)$) -- ($(d')+(1,.1)$);
        \draw[-,red, dashed, thick] ($(d')+(1,.1)$) -- ($(c)+(-.8,-.1)$);
        \draw[->, red, thick, rounded corners] ($(c)+(-.8,-.1)$) -- ($(c)+(-.7,-.2)$) -- ($(d)+(-.7,.4)$) ;
        \node[red] at ($(a)+(1.6,0)$) {a path};
    \end{tikzpicture}
    \caption{Representation of design $\design d$ from Example~\ref{ex-tree}, with a path and a view of $\design d$.}
    \label{fig-ex-path}
  \end{figure}
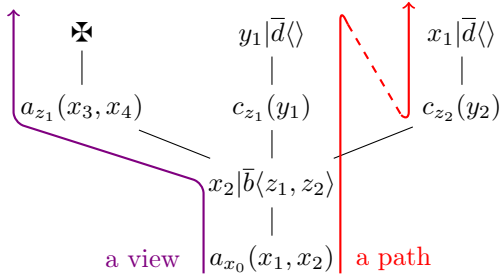
Such a representation is in general a forest: a negative design $\negdes{a}{\vect{x^a}}{\design p_a}$ gives as many trees as there is $a \in \mathcal S$ such that $\design p_a \neq \Omega$. The distinguished variable $x_0$ is given as address to every negative root of a tree, and fresh variables are picked as addresses for negative actions bound by positive ones. This way, negative actions from the same subdesign, i.e., part of the same sum, are given the same address. A tree is indeed to be read bottom-up: a proper action $\kappa$ is \defined{justified} if its address is bound by an action of opposite polarity appearing below $\kappa$ in the tree; otherwise $\kappa$ is called \defined{initial}. Except the root of a tree, which is always initial, every negative action is justified by the only positive action immediately below it. If $\kappa$ and $\kappa'$ are proper, $\kappa$ is \defined{hereditarily justified} by $\kappa'$ if there exist actions $\kappa_1, \dots, \kappa_n$ such that $\kappa = \kappa_1$, $\kappa' = \kappa_n$ and for all $i$ such that $1 \le i < n$, $\kappa_i$ is justified by $\kappa_{i+1}$.

Before giving the definitions of \imp{view} and \imp{path}, let us give an intuition. On Fig.~\ref{fig-ex-path} are represented a view and a path of design $\design d$. Views are branches in the tree representing a cut-free design (reading bottom-up), while paths are particular ``promenades'' starting from the root of the tree; not all such promenades are paths, though. Views correspond to \imp{chronicles} in original ludics \cite{Girard1}.

For every positive proper action $\kappa^+ = \posdes{x}{a}{\vect{y}}$ define $\overline{\kappa^+} = a_x(\vect{y})$, and similarly if $\kappa^- = a_x(\vect{y})$ define $\overline{\kappa^-} = \posdes{x}{a}{\vect{y}}$. Given a finite sequence of proper actions $\pathLL s = \kappa_1 \dots \kappa_n$, define $\overline{\pathLL s} = \overline{\kappa_1} \dots \overline{\kappa_n}$. Suppose now that if $\pathLL s$ contains an occurrence of $\daimon$, it is necessarily in last position; the \defined{dual} of $\pathLL s$, written $\dual{\pathLL s}$, is the sequence defined by:
  \begin{itemize}
  \item $\dual{\pathLL s} = \overline{\pathLL s}\daimon$ if $\pathLL s$ does not end with $\daimon$,
  \item $\dual{\pathLL s} = \overline{\pathLL s'}$ if $\pathLL s = \pathLL s' \daimon$.
  \end{itemize}
Note that $\dual{\dual{\pathLL s}} = \pathLL s$. The notions of \defined{justified}, \defined{hereditarily justified} and \defined{initial} actions also apply in sequences of actions.
\begin{definition} \label{aj-seq} An \defined{alternated justified sequence} (or \defined{aj-sequence}) $\pathLL s$ is a finite sequence of actions such that:
  \begin{itemize}
  \item (Alternation) Polarities of actions alternate.
  \item (Daimon) If $\daimon$ appears, it is the last action of $\pathLL s$.
  \item (Linearity) Each variable is the address of at most one action in $\pathLL s$.
  \end{itemize}
\end{definition}
The (unique) justification of a justified action $\kappa$ in an aj-sequence is noted $\mathrm{just}(\kappa)$, when there is no ambiguity on the sequence we consider.

\begin{definition}
A \defined{view} $\viewseq v$ is an aj-sequence such that each negative action which is not the first action of $\viewseq v$ is justified by the immediate previous action. Given a cut-free design $\design d$, $\viewseq v$ is a \defined{view of} $\design d$ if it is a branch in the representation of $\design d$ as a tree (modulo $\alpha$-equivalence).
\end{definition}
The way to extract \defined{the view} of an aj-sequence is given inductively by:
  \begin{itemize}
    \item $\view{\epsilon} = \epsilon$, where $\epsilon$ is the empty sequence,
    \item $\view{\pathLL s\kappa^+} = \view{\pathLL s}\kappa^+$,
    \item $\view{\pathLL s\kappa^-} = \view{\pathLL s_0}\kappa^-$ where $\pathLL s_0$ is the prefix of $\pathLL s$ ending on $\mathrm{just}(\kappa^-)$, or $\pathLL s_0 = \epsilon$ if $\kappa^-$ initial.
  \end{itemize}
The \defined{anti-view} of an aj-sequence, noted $\antiview{\pathLL s}$, is defined symmetrically by reversing the role played by polarities; equivalently $\antiview{\pathLL s} = \dual{\view{\dual{\pathLL s}}}$.

\begin{definition} \label{def-path}
  A \defined{path} $\pathLL s$ is a positive-ended aj-sequence satisfying:
  \begin{itemize}
  \item (P-visibility) For all prefix $\pathLL s' \kappa^+$ of $\pathLL s$, $\mathrm{just}(\kappa^+) \in \view{\pathLL s'}$
  \item (O-visibility) For all prefix $\pathLL s' \kappa^-$ of $\pathLL s$, $\mathrm{just}(\kappa^-) \in \antiview{\pathLL s'}$
  \end{itemize}
  Given a cut-free design $\design d$, a path $\pathLL s$ is a \defined{path of} $\design d$ if for all prefix $\pathLL s'$ of $\pathLL s$, $\view{\pathLL s'}$ is a view of $\design d$.
\end{definition}
Remark that the dual of a path is a path.

Paths are aimed at describing an interaction between designs. If $\design d$ and $\design e$ are cut-free atomic designs such that $\design d \perp \design e$, there exists a unique path $\pathLL s$ of $\design d$ such that $\dual{\pathLL s}$ is a path of $\design e$. We write this path $\interseq{\design d}{\design e}$, and the good intuition is that it corresponds to the sequence of actions followed by the interaction between $\design d$ and $\design e$ on the side of $\design d$. An alternative way defining orthogonality is then given by the following proposition.
\begin{proposition} \label{perp-path0}
  $\design d \perp \design e$ if and only if there exists a path $\pathLL s$ of $\design d$ such that $\dual{\pathLL s}$ is a path of $\design e$.
\end{proposition}

At the level a behaviour $\beh B$, the set of visitable paths describes all the possible interactions between a design of $\beh B$ and a design of $\beh B^\perp$.

\begin{definition}
A path $\pathLL s$ is \defined{visitable} in a behaviour $\beh B$ if there exist cut-free designs $\design d \in \beh B$ and $\design e \in \beh B^\perp$ such that $\pathLL s = \interseq{\design d}{\design e}$. The set of visitable paths of $\beh B$ is written $\visit B$.
\end{definition}
Note that for every behaviour $\beh B$, $\dual{\visit B} = \visit{B^\perp}$.

\subsection{Regularity, Purity and Connectives} \label{sub-regpur}

The meaning of regularity and purity has been discussed in the introduction. After giving the formal definitions, we prove that regularity is stable under all the connectives constructions. We also show that purity may fail with $\multimap$, and only a weaker form called \imp{quasi-purity} is always preserved.

\begin{definition}\label{reg}
  $\beh B$ is \defined{regular} if the following conditions are satisfied:
  \begin{itemize}
  \item for all $\design d \in |B|$ and all path $\pathLL s$ of $\design d$, $\pathLL s \in \visit{B}$,
  \item for all $\design d \in |B^\perp|$ and all path $\pathLL s$ of $\design d$, $\pathLL s \in \visit{B^\perp}$,
  \item The sets $\visit{B}$ and $\visit{B^\perp}$ are stable under shuffle.
  \end{itemize}
  where the operation of \defined{shuffle} ($\shuffle$) on paths corresponds to an interleaving of actions respecting alternation of polarities, and is defined below. 
\end{definition}
Let $\proj{\pathLL s}{\pathLL s'}$ refer to the subsequence of $\pathLL s$ containing only the actions that occur in $\pathLL s'$. Let $\pathLL s$ and $\pathLL t$ be paths of same polarity, let $S$ and $T$ be sets of paths of same polarity. We define:
\begin{itemize}
\item $\pathLL s \shuffle \pathLL t = \setst{\pathLL u \mbox{ path formed with actions from } \pathLL s \mbox{ and } \pathLL t}{\proj{\pathLL u}{\pathLL s} = \pathLL s \mbox{ and } \proj{\pathLL u}{\pathLL t} = \pathLL t}$ if $\pathLL s, \pathLL t$ negative,
\item $\pathLL s \shuffle \pathLL t = \setst{\kappa^+ \pathLL u \mbox{ path}}{\pathLL u \in \pathLL s' \shuffle \pathLL t'}$ if $\pathLL s = \kappa^+ \pathLL s'$ and $\pathLL t = \kappa^+ \pathLL t'$ positive with same first action,
\item $S \shuffle T = \setst{\pathLL u \mbox{ path}}{\exists \pathLL s \in S, \exists \pathLL t \in T \mbox{ such that } \pathLL s \shuffle \pathLL t \mbox{ is defined and } \pathLL u \in \pathLL s \shuffle \pathLL t}$,
\end{itemize}

In fact, a behaviour $\beh B$ is regular if every path formed with actions of the incarnation of $\beh B$, even mixed up, is a visitable path of $\beh B$, and similarly for $\beh B^\perp$. Remark that regularity is a property of both a behaviour and its orthogonal since the definition is symmetrical: $\beh B$ is regular if and only if $\beh B^\perp$ is regular.

\begin{definition}
A behaviour $\beh B$ is \defined{pure} if every $\daimon$-ended path $\pathLL s \daimon \in \visit B$ is \defined{extensible}, i.e., there exists a proper positive action $\kappa^+$ such that $\pathLL s \kappa^+ \in \visit{B}$.
\end{definition}
Purity ensures that when an interaction encounters $\daimon$, this does not correspond to a real error but rather to a partial computation, as it is possible to continue this interaction. Note that daimons are necessarily present in all behaviours since the converse property is always true: if $\pathLL s \kappa^+ \in \visit{B}$ then $\pathLL s \daimon \in \visit{B}$.

\begin{proposition} \label{prop_reg_stable}
Regularity is stable under $\shpos$, $\shneg$, $\oplus$, $\otimes$ and $\multimap$.
\end{proposition}

\begin{proposition} \label{prop_pure_stable}
Purity is stable under $\shpos$, $\shneg$, $\oplus$ and $\otimes$.
\end{proposition}

Unfortunately, when $\beh N$ and $\beh P$ are pure, $\beh N \multimap \beh P$ is not necessarily pure, even under regularity assumption. However, a weaker form of purity holds for $\beh N \multimap \beh P$.

\begin{definition}
  A behaviour $\beh B$ is \defined{quasi-pure} if all the $\daimon$-ended \imp{well-bracketed} paths in $V_{\beh B}$ are extensible.
\end{definition}
We recall that a path $\pathLL s$ is \defined{well-bracketed} if, for every justified action $\kappa$ in $\pathLL s$, when we write $\pathLL s = \pathLL s_0 \kappa' \pathLL s_1 \kappa \pathLL s_2$ where $\kappa'$ justifies $\kappa$, all the actions in $\pathLL s_1$ are hereditarily justified by $\kappa'$.

\begin{proposition} \label{prop_arrow_princ}
  If $\beh N$ and $\beh P$ are quasi-pure and regular then $\beh N \multimap \beh P$ is quasi-pure.
\end{proposition}

\section{Inductive Data Types} \label{sec-induct}

Some important contributions are presented in this section. We interpret inductive data types as positive behaviours, and we prove an internal completeness result allowing us to make explicit the structure of fixed points. Regularity and purity of data follows.

Abusively, we denote the positive behaviour $\{\daimon\}$ by $\daimon$ all along this section.

\subsection{Inductive Data Types as Kleene Fixed Points} \label{data_types_kleene}

We define the \imp{data patterns} via a type language and interpret them as behaviours, in particular $\mu$ is interpreted as a least fixed point. \imp{Data behaviours} are the interpretation of \imp{steady} data patterns.

Suppose given a countably infinite set $\mathcal V$ of second-order variables: $X, Y, \dots \in \mathcal V$. Let $\mathcal S' = \mathcal S \setminus \{\symshneg, \symavec_1, \symavec_2, \sympar\}$ and define the set of \defined{constants} $\mathrm{Const} = \setst{\beh C_a}{a \in \mathcal S'}$ which contains a behaviour $\beh C_a = \{\posdes{x_0}{a}{\vect{\Omega^-}}\}^{\perp\perp}$ (where $\Omega^- := \negdes{a}{\vect{x^a}}{\Omega}$) for each $a \in \mathcal S'$, i.e., such that $a$ is not the name of a connective. Remark that $V_{\beh C_a} = \{\daimon \hspace{.1cm} , \hspace{.1cm} \posdes{x_0}{a}{\vect x}\}$, thus $\beh C_a$ is regular and pure.

\begin{definition}
The set $\mathcal P$ of \defined{data patterns} is generated by the inductive grammar:
\[ A, B ::=  X \in \mathcal V \gbar a \in \mathcal S' \gbar A \oplus^+ B \gbar A \otimes^+ B \gbar \mu X. A\]
\end{definition}
The set of free variables of a data pattern $A \in \mathcal P$ is denoted by $\mathrm{FV}(A)$.

\begin{example} \label{ex-patterns}
Let $b, n, l, t \in \mathcal S'$ and $X \in \mathcal V$. The data types given as example in the introduction can be written in the language of data patterns as follows:
\begin{align*}
  \mathrm{\mathbb Bool} & = b \oplus^+ b \hspace{1.8cm} \mathrm{\mathbb Nat} = \mu X.( n \oplus^+ X) \hspace{1.8cm} \mathrm{\mathbb List}_{A} = \mu X.(l \oplus^+ (A \otimes^+ X)) \\
  & \mathrm{\mathbb Tree}_A = \mu X.(t \oplus^+ (A \otimes^+ \mathrm{\mathbb List}_X)) = \mu X.(t \oplus^+ (A \otimes^+ \mathbb \mu Y.(l \oplus^+ (X \otimes^+ Y))))
\end{align*}
\end{example}

Let $\mathcal B^+$ be the set of positive behaviours. Given a data pattern $A \in \mathcal P$ and an environment $\sigma$, i.e., a function that maps free variables to positive behaviours, the interpretation of $A$ in the environment $\sigma$, written $\interpret{A}^{\sigma}$, is the positive behaviour defined by:
\begin{align*}
  & \interpret{X}^{\sigma} = \sigma(X) & \interpret{A \oplus^+ B}^{\sigma} & = (\shneg \interpret{A}^{\sigma}) \oplus (\shneg \interpret{B}^{\sigma}) \\
  & \interpret{a}^{\sigma} = \beh C_a & \interpret{A \otimes^+ B}^{\sigma} & = (\shneg \interpret{A}^{\sigma}) \otimes (\shneg \interpret{B}^{\sigma}) \\
  & \interpret{\mu X. A}^{\sigma} = \mathrm{lfp}(\phi^A_{\sigma})
\end{align*}
where $\mathrm{lfp}$ stands for the least fixed point, and the function $\phi^A_{\sigma}: \mathcal B^+ \to \mathcal B^+, \beh P \mapsto \interpret{A}^{\sigma, X\mapsto \beh P}$ is well defined and has a least fixed point by Knaster-Tarski fixed point theorem, as shown by Baelde, Doumane and Saurin \cite{BDS}. Abusively we may write $\oplus^+$ and $\otimes^+$, instead of $(\shneg \cdot) \oplus (\shneg \cdot)$ and $(\shneg \cdot) \otimes (\shneg \cdot)$ respectively, for behaviours. We call an environment $\sigma$ regular (resp. pure) if its image contains only regular (resp. pure) behaviours. The notation $\sigma, X\mapsto \beh P$ stands for the environment $\sigma$ where the image of $X$ has been changed to $\beh P$.

In order to understand the structure of fixed point behaviours that interpret the data patterns of the form $\mu X.A$, we need a constructive approach, thus Kleene fixed point theorem is best suited than Knaster-Tarski. We now prove that we can apply this theorem.

Recall the following definitions and theorem. A partial order is a \defined{complete partial order} (CPO) if each directed subset has a supremum, and there exists a smallest element, written $\bot$. A function $f : E \to F$ between two CPOs is \defined{Scott-continuous} (or simply continuous) if for every directed subset $D \subseteq E$ we have $\bigvee_{x \in D} f(x) = f(\bigvee_{x \in D}x)$.

\begin{theorem}[Kleene fixed point theorem]
  Let $L$ be a CPO and let $f:L \to L$ be Scott-continuous. The function $f$ has a least fixed point, defined by
  \[\mathrm{lfp}(f) = \bigvee_{n \in \mathbb N}f^n(\bot)\]
\end{theorem}
The set $\mathcal B^+$ ordered by $\subseteq$ is a CPO, with least element $\daimon$; indeed, given a subset $\mathbb P \subseteq \mathcal B^+$, it is directed and we have $\bigvee \mathbb P = (\bigcup \mathbb P)^{\perp\perp}$. Hence next proposition proves that we can apply the theorem.

\begin{proposition} \label{prop_scott_conti}  Given a data pattern $A \in \mathcal P$, a variable $X \in \mathcal V$ and an environment $\sigma : \mathrm{FV}(A) \setminus \{X\} \to \mathcal B^+$, the function $\phi^{A}_{\sigma}$ is Scott-continuous.
\end{proposition}

\begin{corollary} \label{coro_kleene_sup} For every $A \in \mathcal P$, $X \in \mathcal V$ and $\sigma : \mathrm{FV}(A) \setminus \{X\} \to \mathcal B^+$,
  \[\interpret{\mu X.A}^{\sigma} = \bigvee_{n \in \mathbb N} (\phi^{A}_{\sigma})^n(\daimon) = (\bigcup_{n \in \mathbb N} (\phi^{A}_{\sigma})^n(\daimon))^{\perp\perp}\]
\end{corollary}
This result gives an explicit formulation for least fixed points. However, the $\perp\perp$-closure might add new designs which were not in the union, making it difficult to know the exact content of such a behaviour. The point of next subsection will be to give an internal completeness result proving that the closure is actually not necessary.

Let us finish this subsection by defining a restricted set of data patterns so as to exclude the degenerate ones. Consider for example ${\mathbb List_{A}}' = \mu X. (A \otimes^+ X)$, a variant of $\mathbb List_{A}$ (see Example~\ref{ex-patterns}) which misses the base case. It is degenerate in the sense that the base element, here the empty list, is interpreted as the design $\daimon$. This is problematic: an interaction going through a whole list will end with an error, making it impossible to explore a pair of lists for example. The pattern $\mathbb Nat' = \mu X.X$ is even worse since $\interpret{\mathbb Nat'} = \daimon$. The point of steady data patterns is to ensure the existence of a basis; this will be formalised in Lemma~\ref{lem_cb_basis}.

\begin{definition}
  The set of \defined{steady} data patterns is the smallest subset $\mathcal P^s \subseteq \mathcal P$ such that:
  \begin{itemize}
  \item $\mathcal S' \subseteq \mathcal P^s$
  \item If $A \in \mathcal P^s$ and $B$ is such that $\interpret{B}^\sigma$ is pure if $\sigma$ is pure, then $A \oplus^+ B\in \mathcal P^s$ and $B \oplus^+ A \in \mathcal P^s$
  \item If $A \in \mathcal P^s$ and $B\in \mathcal P^s$ then $A \otimes^+ B \in \mathcal P^s$
  \item If $A \in \mathcal P^s$ then $\mu X.A \in \mathcal P^s$
  \end{itemize}
\end{definition}

The condition on $B$ in the case of $\oplus^+$ admits data patterns which are not steady, possibly with free variables, but ensuring the preservation of purity, i.e., type safety; the basis will come from side $A$. We will prove (\textsection~\ref{reg_pure_data}) that behaviours interpreting steady data patterns are pure, thus in particular a data pattern of the form $\mu X.A$ is steady if the free variables of $A$ all appear on the same side of a $\oplus^+$ and under the scope of no other $\mu$ (since purity is stable under $\shpos, \shneg, \oplus, \otimes$). We claim that steady data patterns can represent every type of finite data.
\begin{definition}
  A \defined{data behaviour} is the interpretation of a closed steady data pattern.
\end{definition}

\subsection{Internal Completeness for Infinite Union}

Our main result is an internal completeness theorem, stating that an infinite union of \imp{simple} regular behaviours with increasingly large incarnations is a behaviour: $\perp\perp$-closure is useless.

\begin{definition}
  \begin{itemize}
  \item A \defined{slice} is a design in which all negative subdesigns are either $\Omega^-$ or of the form $a(\vect x).\design p_a$, i.e., at most unary branching. $\design c$ is a \defined{slice of} $\design d$ if $\design c$ is a slice and $\design c \sqsubseteq \design d$. A slice $\design c$ of $\design d$ is \defined{maximal} if for any slice $\design c'$ of $\design d$ such that $\design c \sqsubseteq \design c'$, we have $\design c = \design c'$.
  \item A behaviour $\beh B$ is \defined{simple} if for every design $\design d \in |\beh B|$:
    \begin{enumerate}
    \item $\design d$ has a finite number of maximal slices, and
    \item every positive action of $\design d$ is justified by the immediate previous negative action.
    \end{enumerate}
  \end{itemize}
\end{definition}

Condition (2) of simplicity ensures that, given $\design d \in |\beh B|$ and a slice $\design c \sqsubseteq \design d$, one can find a path of $\design c$ containing all the positive proper actions of $\design c$ until a given depth; thus by condition (1), there exists $k \in \mathbb N$ depending only on $\design d$ such that $k$ paths can do the same in $\design d$.

Now suppose $(\beh A_n)_{n \in \mathbb N}$ is an infinite sequence of simple regular behaviours such that for all $n \in \mathbb N$, $|\beh A_n| \subseteq |\beh A_{n+1}|$ (in particular we have $\beh A_n \subseteq \beh A_{n+1}$).

\begin{theorem} \label{thm_union_beh}
  The set $\bigcup_{n \in \mathbb N} \beh A_n$ is a behaviour.
\end{theorem}
A union of behaviours is not a behaviour in general. In particular, counterexamples are easily found if releasing either the inclusion of incarnations or the simplicity condition. Moreover, our proof for this theorem relies strongly on regularity. Under the same hypotheses we can prove $V_{\bigcup_{n \in \mathbb N} \beh A_n} = \bigcup_{n \in \mathbb N} V_{\beh A_n}$ and $|\bigcup_{n \in \mathbb N} \beh A_n| = \bigcup_{n \in \mathbb N} |\beh A_n|$, hence the following corollary.

\begin{corollary} \label{coro_reg_pur}
  \begin{itemize}
  \item $\bigcup_{n \in \mathbb N} \beh A_n$ is simple and regular;
  \item if moreover all the $\beh A_n$ are pure then $\bigcup_{n \in \mathbb N} \beh A_n$ is pure.
  \end{itemize}
\end{corollary}

\subsection{Regularity and Purity of Data} \label{reg_pure_data}

The goal of this subsection is to show that the interpretation of data patterns of the form $\mu X.A$ can be expressed as an infinite union of behaviours $(\beh A_n)_{n \in \mathbb N}$ satisfying the hypotheses of Theorem~\ref{thm_union_beh}, in order to deduce regularity and purity. We will call an environment $\sigma$ simple if its image contains only simple behaviours.

\begin{lemma} \label{lem_incarn_hier}
For all $A \in \mathcal P$, $X \in \mathcal V$, $\sigma : \mathrm{FV}(A) \setminus \{X\} \to \mathcal B^+$ simple and regular\footnote{The hypothesis ``simple and regular'' has been added, compared to the CSL version of this article, for correction.}, and $n \in \mathbb N$ we have
\[|(\phi^{A}_{\sigma})^n(\daimon)| \subseteq |(\phi^{A}_{\sigma})^{n+1}(\daimon)|\]
\end{lemma}

\begin{proposition}  \label{prop_interp_reg}
For all $A \in \mathcal P$ and simple regular environment $\sigma$, $\interpret{A}^{\sigma}$ is simple regular.
\end{proposition}

\begin{proof}
  By induction on data patterns. If $A = X$ or $A = a$ the conclusion is immediate. If $A = A_1 \oplus^+ A_2$ or $A = A_1 \otimes^+ A_2$ then regularity comes from Proposition~\ref{prop_reg_stable}, and simplicity is easy since the structure of the designs in $\interpret{A}^{\sigma}$ is given by internal completeness for the logical connectives (Theorem~\ref{thm_intcomp_all}). So suppose $A = \mu X.A_0$. By induction hypothesis, for every simple regular behaviour $\beh P \in \mathcal B^+$ we have $\phi^{A_0}_{\sigma}(\beh P) = \interpret{A_0}^{\sigma, X \mapsto \beh P}$ simple regular. From this, it is straightforward to show by induction that for every $n \in \mathbb N$, $(\phi^{A_0}_{\sigma})^n(\daimon)$ is simple regular. Moreover, for every $n \in \mathbb N$ we have $|(\phi^{A_0}_{\sigma})^n(\daimon)| \subseteq |(\phi^{A_0}_{\sigma})^{n+1}(\daimon)|$ by Lemma~\ref{lem_incarn_hier}, thus by Corollary~\ref{coro_kleene_sup} and Theorem~\ref{thm_union_beh}, $\interpret{\mu X.A_0}^{\sigma}  = \bigvee_{n \in \mathbb N} (\phi^{A}_{\sigma})^n(\daimon) = (\bigcup_{n \in \mathbb N}(\phi^{A_0}_{\sigma})^n(\daimon))^{\perp\perp} = \bigcup_{n \in \mathbb N}(\phi^{A_0}_{\sigma})^n(\daimon)$.
Consequently, by Corollary~\ref{coro_reg_pur}, $\interpret{\mu X.A_0}^{\sigma}$ is simple regular.
\end{proof}
Remark that we have proved at the same time, using Theorem~\ref{thm_union_beh}, that behaviours interpreting data patterns $\mu X.A$ admit an explicit construction:
\begin{proposition} \label{prop_reg_union}
  If $A \in \mathcal P$, $X \in \mathcal V$, and $\sigma : \mathrm{FV}(A)\setminus X \to \mathcal B^+$ is simple regular,
  \[\interpret{\mu X.A}^{\sigma} = \bigcup_{n \in \mathbb N} (\phi^{A}_{\sigma})^n(\daimon)\]
\end{proposition}

\begin{corollary} \label{coro_data_reg}
  Data behaviours are regular.
\end{corollary}

We now move on to proving purity. The proof that the interpretation of a steady data pattern $A$ is pure relies on the existence of a basis for $A$ (Lemma~\ref{lem_cb_basis}). Let us first widen (to $\daimon$-free paths) and express in a different way (for $\daimon$-ended paths) the notion of extensible visitable path.

\begin{definition} Let $\beh B$ be a behaviour.
  \begin{itemize}
  \item A $\daimon$-free path $\pathLL s \in \visit B$ is \defined{extensible} if there exists $\pathLL t \in \visit B$ of which $\pathLL s$ is a strict prefix.
  \item A $\daimon$-ended path $\pathLL s\daimon \in \visit B$ is \defined{extensible} if there exists a positive action $\kappa^+$ and $\pathLL t \in \visit B$ of which $\pathLL s\kappa^+$ is a prefix.
  \end{itemize}
\end{definition}
Write $V_{\beh B}^{max}$ for the set of maximal, i.e., non extensible, visitable paths of $\beh B$.
  
\begin{lemma} \label{lem_cb_basis}
  Every steady data pattern $A \in \mathcal P^s$ has a basis, i.e., a simple regular behaviour $\beh B$ such that for all simple regular environment $\sigma$ we have
  \begin{itemize}
  \item $\beh B \subseteq \interpret{A}^{\sigma}$,
  \item for every path $\pathLL s \in V_{\beh B}$, there exists $\pathLL t \in V_{\beh B}^{max}$ $\daimon$-free extending $\pathLL s$ (in particular $\beh B$ pure),
  \item $V_{\beh B}^{max} \subseteq V_{\interpret{A}^{\sigma}}^{max}$.
    \end{itemize}
\end{lemma}

\begin{proof}[Proof (idea)]
  If $A = a$, a basis is $\beh C_a$. If $A = A_1 \oplus^+ A_2$, and $A_i$ is steady with basis $\beh B_i$, then $\otimes_i \shneg \beh B_i:=\symplus_i\langle\shneg \beh B_i\rangle$ is a basis for $A$. If $A = A_1 \otimes^+ A_2$, a basis is $\beh B_1 \otimes^+ \beh B_2$ where $\beh B_1$ and $\beh B_2$ are basis of $A_1$ and $A_2$ respectively. If $A = \mu X.A_0$, its basis is the same as $A_0$.
\end{proof}

\begin{proposition} \label{prop_basis}
  If $A \in \mathcal P^s$ of basis $\beh B$, $X \in \mathcal V$, and $\sigma : \mathrm{FV}(A)\setminus X \to \mathcal B^+$ simple regular,
  \[\interpret{\mu X.A}^{\sigma} = \bigcup_{n \in \mathbb N} (\phi^{A}_{\sigma})^n(\beh B)\]
\end{proposition}

\begin{proof}
  Since $\beh B$ is a basis for $A$ we have $\daimon \subseteq \beh B \subseteq \interpret{A}^{\sigma, X \to \daimon} = \phi^{A}_{\sigma}(\daimon)$. The Scott-continuity of the function $\phi^{A}_{\sigma}$ implies that it is increasing, thus $(\phi^{A}_{\sigma})^n(\daimon) \subseteq (\phi^{A}_{\sigma})^n(\beh B) \subseteq (\phi^{A}_{\sigma})^{n+1}(\daimon)$ for all $n \in \mathbb N$.
  Hence $\interpret{A}^{\sigma} = \bigcup_{n \in \mathbb N}(\phi^{A}_{\sigma})^n(\daimon) = \bigcup_{n \in \mathbb N} (\phi^{A}_{\sigma})^n(\beh B)$.
\end{proof}

\begin{proposition}
   For all $A \in \mathcal P^s$ and simple regular pure environment $\sigma$, $\interpret{A}^\sigma$ is pure.
\end{proposition}

\begin{proof}
  By induction on $A$. The base cases are immediate and the connective cases are solved using Proposition~\ref{prop_pure_stable}. Suppose now $A = \mu X.A_0$, where $A_0$ is steady with basis $\beh B_0$. We have $\interpret{A}^{\sigma} = \bigcup_{n \in \mathbb N} (\phi^{A_0}_{\sigma})^n(\beh B_0)$ by Proposition~\ref{prop_basis}, let us prove it satisfies the hypotheses needed to apply Corollary~\ref{coro_reg_pur}(2). By induction hypothesis and Proposition~\ref{prop_interp_reg}, for every simple, regular and pure behaviour $\beh P \in \mathcal B^+$ we have $\phi^{A_0}_{\sigma}(\beh P) = \interpret{A_0}^{\sigma, X \mapsto \beh P}$ simple, regular and pure, hence it is easy to show by induction that for every $n \in \mathbb N$, $(\phi^{A_0}_{\sigma})^n(\beh B_0)$ is as well. Moreover, for every $n \in \mathbb N$ we prove that $|(\phi^{A_0}_{\sigma})^n(\beh B_0)| \subseteq |(\phi^{A_0}_{\sigma})^{n+1}(\beh B_0)|$ similarly to Lemma~\ref{lem_incarn_hier}, replacing $\daimon$ by the basis $\beh B_0$. Finally, by Corollary~\ref{coro_reg_pur}, $\interpret{A}^{\sigma}$ is pure.
\end{proof}

\begin{corollary} \label{coro_data_pur}
  Data behaviours are pure.
\end{corollary}

\begin{remark}
Although here the focus is on the interpretation of data patterns, we should say a word about the interpretation of (polarised) $\mu$MALL formulas, which are a bit more general. These formulas are generated by:
  \begin{align*}
    P, Q & \hspace{.3cm} ::= \hspace{.3cm} X_P \gbar X_N^\perp \gbar 1 \gbar 0 \gbar M \oplus N \gbar M \otimes N \gbar \shpos N \gbar \mu X.P \\
    M, N & \hspace{.3cm} ::= \hspace{.3cm} P^\perp
  \end{align*}
  where the usual involutive negation hides the negative connectives and constants, through the dualities $1/\bot$, $0/\top$, $\oplus/\with$, $\otimes/\parr$, $\shpos/\shneg$, $\mu/\nu$ . The interpretation as ludics behaviours, given in~\cite{BDS}, is as follows: $1$ is interpreted as a constant behaviour $\beh C_a$, $0$ is the daimon $\daimon$, the positive connectives match their ludics counterparts, $\mu$ is interpreted as the least fixed point of a function $\phi^A_\sigma$ similarly to data patterns, and the negation corresponds to the orthogonal. Since in ludics constants and $\daimon$ are regular, and since regularity is preserved by the connectives (Proposition~\ref{prop_reg_stable}) and by orthogonality, the only thing we need in order to prove that all the behaviours interpreting $\mu$MALL formulas are regular is a generalisation of regularity stability under fixed points (for now we only have it in our particular case: Corollary~\ref{coro_reg_pur} together with Proposition~\ref{prop_reg_union}).
  
Note however that interpretations of $\mu$MALL formulas are not all pure. Indeed, as we will see in next section, orthogonality (introduced through the connective $\multimap$) does not preserve purity in general.
\end{remark}

\section{Functional Types} \label{sec-func}

In this section we define \imp{functional behaviours} which combine data behaviours with the connective $\multimap$. A behaviour of the form $\beh N \multimap \beh P$ is the set of designs such that, when interacting with a design of type $\beh N$, outputs a design of type $\beh P$; this is exactly the meaning of its definition $\beh N \multimap \beh P :=  (\beh N \otimes \beh P^\perp)^\perp$.
  We prove that some particular higher-order functional types -- where functions are taken as arguments, typically $(A \multimap B) \multimap C$ -- are exactly those who fail at being pure, and we interpret this result from a computational point of view.

\subsection{Where Impurity Arises} \label{sub_func_impur}

We have proved that data behaviours are regular and pure. However, if we introduce functional behaviours with the connective $\multimap$, purity does not hold in general. Proposition~\ref{prop_func_quasi_pure} indicates that a weaker property, quasi-purity, holds for functional types, and Proposition~\ref{prop_main} identifies exactly the cases where purity fails.

Let us write $\mathcal D$ for the set of data behaviours.

\begin{definition}
A \defined{functional behaviour} is a behaviour inductively generated by the grammar below, where $\beh P \multimap^+ \beh Q$ stands for $\shpos((\shneg \beh P) \multimap \beh Q)$.
 \[ \beh P, \beh Q ::= \beh P_0 \in \mathcal D \gbar \beh P \oplus^+ \beh Q \gbar \beh P \otimes^+ \beh Q \gbar \beh P \multimap^+ \beh Q \]
\end{definition}
 
From Propositions~\ref{prop_reg_stable}, \ref{prop_pure_stable} and \ref{prop_arrow_princ} we easily deduce the following result.
\begin{proposition} \label{prop_func_quasi_pure}
  Functional behaviours are regular and quasi-pure.
\end{proposition}

For next proposition, consider \defined{contexts} defined inductively as follows (where $\beh P$ is a functional behaviour):
\[\mathcal C ::= [~] \gbar \mathcal C \oplus^+ \beh P \gbar \beh P \oplus^+ \mathcal C \gbar \mathcal C \otimes^+ \beh P \gbar \beh P \otimes^+ \mathcal C \gbar \beh P \multimap^+ \mathcal C\]

\begin{proposition} \label{prop_main}
  A functional behaviour $\beh P$ is impure if and only if there exist contexts $\mathcal C_1, \mathcal C_2$ and functional behaviours $\beh Q_1, \beh Q_2, \beh R$ with $\beh R \notin \mathrm{Const}$ such that
  \[\beh P = \mathcal C_1[~\mathcal C_2[\beh Q_1 \multimap^+ \beh Q_2] \multimap^+ \beh R~]\]
\end{proposition}

\subsection{Example and Discussion} \label{ex_discuss}

Proposition~\ref{prop_main} states that a functional behaviour which takes functions as argument is not pure: some of its visitable paths end with a daimon $\daimon$, and there is no possibility to extend them. In terms of proof-search, playing the daimon is like giving up; on a computational point of view, the daimon appearing at the end of an interaction expresses the sudden interruption of the computation. In order to understand why such an interruption can occur in the specific case of higher-order functions, consider the following example which illustrates the proposition.

\begin{example} \label{ex-final}  
  Let $\beh Q_1, \beh Q_2, \beh 1$ be functional behaviours, with $\beh 1 \in \mathrm{Const}$. Define $\beh{Bool} = \beh 1 \oplus^+ \beh 1$ and consider the behaviour $\beh P = (\beh Q_1 \multimap^+ \beh Q_2) \multimap^+ \beh{Bool}$: this is a type of functions which take a function as argument and output a boolean. Let $\alpha_1, \alpha_2, \beta$ be respectively the first positive action of the designs of $\beh Q_1, \beh Q_2, \beh 1$. It is possible to exhibit a design $\design p \in \beh P$ and a design $\design n \in \beh P^\perp$ such that the visitable path $\pathLL s = \interseq{\design p}{\design n}$ is $\daimon$-ended and maximal in $\visit P$, in other words $\pathLL s$ is a witness of the impurity of $\beh P$. The path $\pathLL s$ contains the actions $\alpha_1$ and $\overline{\alpha_2}$ in such a way that it cannot be extended with $\beta$ without breaking the P-visibility condition, and there is no other available action in designs of $\beh P$ to extend it. Reproducing the designs $\design p$ and $\design n$ and the path $\pathLL s$ here would be of little interest since those objects are too large to be easily readable ($\pathLL s$ visits the entire design $\design p$, which contains 11 actions). We however give an intuition in the style of game semantics: Fig.~\ref{ho-play} represents $\pathLL s$ as a legal play in a strategy of type $\beh P = (\beh Q_1 \multimap^+ \beh Q_2) \multimap^+ \beh{Bool}$ (note that only one ``side'' $\oplus_1 \shneg \beh 1$ of $\beh{Bool}$ is represented, corresponding for example to \texttt{True}, because we cannot play in both sides). This analogy is informal, it should stand as an intuition rather than as a precise correspondence with ludics; for instance, and contrary to the way it is presented in game semantics, the questions are asked on the connectives, while the answers are given in the sub-types of $\beh P$. On the right are given the actions in $\pathLL s$ corresponding to the moves played. The important thing to remark is the following: if a move $b$ corresponding to action $\beta$ were played instead of $\daimon$ at the end of this play, it would break the P-visibility of the strategy, since this move would be justified by move $q_{\shneg}$.

\begin{figure}
  \centering
  \includegraphics[scale=1]{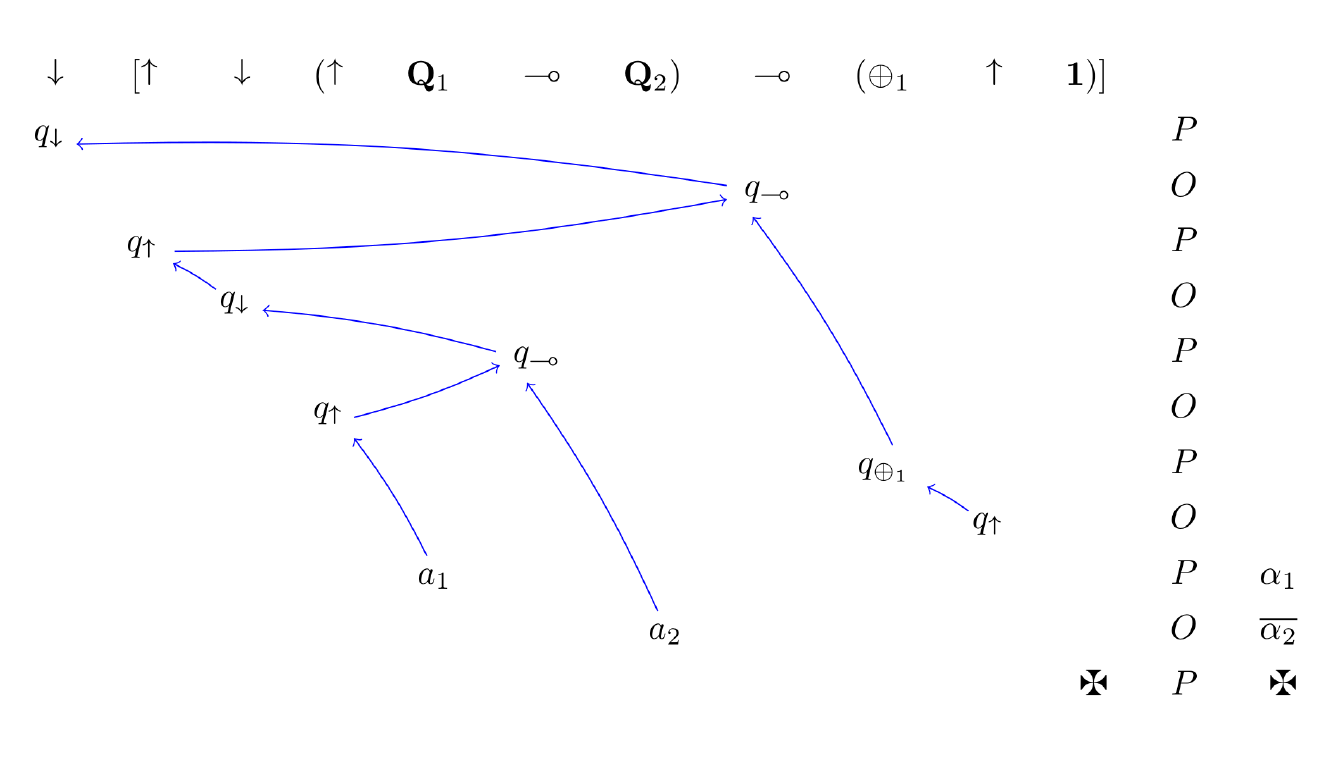}
   \caption{Representation of path $\pathLL s$ from Example~\ref{ex-final} in the style of a legal play}
   \label{ho-play}
\end{figure}
  
The computational interpretation of the $\daimon$-ended interaction between $\design p$ and $\design n$ is the following: a program $p$ of type $\beh P$ launches a child process $p'$ to compute the argument of type $\beh Q_1 \to \beh Q_2$, but $p$ starts to give a result in $\beh{Bool}$ before the execution of $p'$ terminates, leading to a situation where $p$ cannot compute the whole data in $\beh{Bool}$. The interaction outputs $\daimon$, i.e., the answer given in $\beh{Bool}$ by $p$ is incomplete.

Moreover by Proposition~\ref{prop_func_quasi_pure} functional behaviours are quasi-pure, therefore the maximal $\daimon$-ended visitable paths are necessarily not well-bracketed. This is indeed the case of $\pathLL s$: remark for example that the move $q_{\oplus_1}$ appears between $a_1$ and its justification $q_{\shneg}$ in the sequence, but $q_{\oplus_1}$ is not hereditarily justified by $q_{\shneg}$.
In HO games, well-bracketedness is a well studied notion, and relaxing it introduces control operators in program. If we extend such an argument to ludics, this would mean that the appearance of $\daimon$ in the execution of higher-order functions can only happen in the case of programs with control operators such as \imp{jumps}, i.e. programs which are not purely functional.
\end{example}

\section{Conclusion}

This article is a contribution to the exploration of the behaviours of linear ludics in a computational perspective. Our focus is on the behaviours representing data types and functional types. Inductive data types are interpreted using the logical connectives constructions and a least fixed point operation. Adopting a constructive approach, we provide an internal completeness result for fixed points, which unveils the structure of data behaviours. This leads us to proving that such behaviours are regular -- the key notion for the characterisation of MALL in ludics --  and pure -- that is, type safe.
But behaviours interpreting types of functions taking functions as argument are impure; for well-bracketed interactions, corresponding to the evaluation of purely functional programs, safety is however guaranteed.

\subparagraph{Further Work} Two directions for future research arise naturally:
\begin{itemize}
\item Extending our study to greatest fixed points $\nu X.A$, i.e., coinduction, is the next objective. Knaster--Tarski ensures that such greatest fixed point behaviours exist \cite{BDS}, but Kleene fixed point theorem does not apply here, hence we cannot find an explicit form for coinductive behaviours the same way we did for the inductive ones. However it is intuitively clear that, compared to least fixed points, greatest ones add the infinite ``limit'' designs in (the incarnation of) behaviours. For example, if $\mathbb Nat_{\omega} = \nu X. (1 \oplus X)$ then we should have $|\interpret{\mathbb Nat_{\omega}}| = |\interpret{\mathbb Nat}| \cup \{\design d_\omega\}$ where $\design d_\omega = \mathrm{succ}(\design d_\omega) = x_0|\symplus_2\langle\shneg(x).{\design d_\omega}^{x}\rangle$.
\item Another direction would be to get a complete characterisation of $\mu$MALL in ludics, by proving that a behaviour is regular -- and possibly satisfying a supplementary condition -- if and only if it is the denotation of a $\mu$MALL formula.
\end{itemize}

\subparagraph{Acknowledgement} I thank Claudia Faggian, Christophe Fouqueré, Thomas Seiller and the anonymous referees for their wise and helpful comments.

\bibliographystyle{plainurl}
\bibliography{references}

\appendix

In the appendix, we adopt Barendregt's variable convention; that is, among objects in a given context, we will always assume that:
\begin{enumerate}
\item no variable appears both free and bound, and
\item bound variables have all distinct names.
\end{enumerate}
This affects designs, multi-designs, representations of designs as trees, and paths.

\section{Proof of Proposition~\ref{perp-path0}} \label{multi}

The purpose of this section is to lift the framework to \imp{multi-designs}, in order to prove properties of the path recording the interaction between multi-designs (thus in particular, between designs). We show:
\begin{itemize}
\item the existence and uniqueness of the interaction path between two orthogonal multi-designs (Proposition~\ref{inter-unique}),
\item the equivalence between the existence of such a path and the orthogonality of two multi-designs (Proposition~\ref{perp-path}, a generalisation of Proposition~\ref{perp-path0}),
\item an associativity theorem for paths (Proposition~\ref{asso-path}).
\end{itemize}
These results are needed for next section. Their proofs require a lot of supplementary formalism, so the reader intuitively convinced may jump directly to next section.

\subsection{Multi-Designs}

The notion of \imp{multi-design} introduced below generalises the one of \imp{anti-design} given by Terui \cite{Terui}, thus in particular it generalises designs. Interaction between two \imp{compatible} multi-designs $\design D$ and $\design E$ corresponds to eliminating the cuts in another multi-design $\cut{\design D}{\design E}$. Several well-known notions of Ludics can be extended to this setting.

\begin{definition} \label{multi-des} \ 
  \begin{itemize}
  \item A \defined{negative multi-design} is a set $\{(x_1, \design n_1), \dots, (x_n, \design n_n)\}$ where $x_1 , \dots , x_n$ are distinct variables and $\design n_1, \dots, \design n_n$ are negative designs, such that for all ${1 \le i \le n}$, $\mathrm{fv}(\design n_i) \cap \{x_1 , \dots , x_n\} = \emptyset$, and for all $j \neq i$, $\mathrm{fv}(\design n_i) \cap \mathrm{fv}(\design n_j) = \emptyset$.
  \item A \defined{positive multi-design} is a set $\{\design p, (x_1, \design n_1), \dots, (x_n, \design n_n)\}$ where $\{(x_1, \design n_1), \dots, (x_n, \design n_n)\}$ is a negative multi-design and $\design p$ is a positive design such that $\mathrm{fv}(\design p) \cap \{x_1, \dots, x_n\} = \emptyset$, and for all $1 \le i \le n$, $\mathrm{fv}(\design p) \cap \mathrm{fv}(\design n_i) = \emptyset$.
  \end{itemize}
\end{definition}

We will use $\design D, \design E, \dots$ to denote multi-designs of any polarity, $\design M, \design N, \dots$ for negative ones and $\design P, \design Q, \dots$ for positive ones. A pair $(x, \design n)$ in a multi-design will be denoted by $\design n/x$ or $(\design n/x)$; hence a negative multi-design will be written $\{\design n_1/x_1, \dots, \design n_n/x_n\}$ (or even $\{\vect{\design n/x}\}$), a positive one $\{\design p, \design n_1/x_1, \dots, \design n_n/x_n\}$, and we will write $(\design n/x) \in \design D$ instead of $(x, \design n) \in \design D$.
This notation makes the parallel with substitution: if $\design N = \{\design n_1/x_1, \dots, \design n_n/x_n\}$ and $\design d$ is a design, then we will allow to write $\design d[\design N]$ for the substitution $\design d[\design n_1/x_1, \dots, \design n_n/x_n]$. By abuse, we might also write $\design n \in \design D$ when the variable associated to $\design n$ in the multi-design $\design D$ does not matter; thus when writing ``let $\design d \in \design D$'', the design $\design d$ can be either positive or negative associated with a variable in $\design D$.

A design can be viewed as a multi-design: a positive design $\design p$ corresponds to the positive multi-design $\{\design p\}$, and a negative design $\design n$ to the negative multi-design $\{\design n/x_0\}$, where $x_0$ is the same distinguished variable we introduced for atomic designs. Notations $\design p$ and $\design n$ will be used instead of $\{\design p\}$ and $\{\design n/x_0\}$ respectively.

Note that if $\design D$ and $\design E$ are multi-designs, $\design D \cup \design E$ is not always a multi-design.

\begin{definition} \label{multi-normal}
  Let $\design D$ be a multi-design. Its \defined{normal form} is the cut-free multi-design defined by
  \[\normalisation{\design D} = \setst{(\normalisation{\design n}/x)}{(\design n/x) \in \design D} \cup \setst{\normalisation{\design p}}{\design p \in \design D}\]
\end{definition}

\begin{definition}
  Let $\design D$ be a multi-design.
  \begin{itemize}
  \item The \defined{free variables} of $\design D$ are $\mathrm{fv}(\design D) = \bigcup_{\design d \in \design D}\mathrm{fv}(\design d)$
  \item The \defined{negative places} of $\design D$ are $\mathrm{np}(\design D) = \setst{x}{\exists \design n \ (\design n/x) \in \design D}$.
  \end{itemize}
\end{definition}

In Definition~\ref{multi-des}, the condition ``for all ${1 \le i \le n}$, $\mathrm{fv}(\design n_i) \cap \{x_1 , \dots , x_n\} = \emptyset$'' (adding the similar condition for $\design p$ in the positive case) can thus be rephrased as ``$\mathrm{fv}(\design D) \cap \mathrm{np}(\design D) = \emptyset$''.
When two multi-designs $\design D$ and $\design E$ interact, this condition will ensure that a substitution specified in $\design D$ or in $\design E$ creates a cut between a design from $\design D$ and a design from $\design E$, and never between two designs on the same side. This is exactly the form of interaction we want in the following: an interaction with two distinct sides. But in order to talk about interaction between two multi-designs, we must first determine when two multi-designs are \imp{compatible}, i.e., when we can define substitution between them in a unique way, without ambiguity, which is not the case in general.

\begin{definition}
  Let $\design D$ and $\design E$ be multi-designs.
  \begin{itemize}
  \item $\design D$ and $\design E$ are \defined{compatible} if they satisfy the following conditions:
    \begin{itemize}
    \item $\mathrm{fv}(\design D) \cap \mathrm{fv}(\design E) = \mathrm{np}(\design D) \cap \mathrm{np}(\design E) = \emptyset$
    \item either they are both negative and there exists $x \in \mathrm{np}(\design D) \cup \mathrm{np}(\design E)$ such that $x \notin \mathrm{fv}(\design D) \cup \mathrm{fv}(\design E)$, or they are of opposite polarities
    \end{itemize}
  \item $\design D$ and $\design E$ are \defined{closed-compatible} if they are of opposite polarities, compatible, and satisfying $\mathrm{fv}(\design D) = \mathrm{np}(\design E)$ and $\mathrm{fv}(\design E) = \mathrm{np}(\design D)$.
  \end{itemize}
\end{definition}

Intuitively, \imp{compatible} means that we are able to define the multi-design of the interaction between $\design D$ and $\design E$, and \imp{closed-compatible} means that this multi-design is a closed design.

\begin{definition} \label{cut}
  Let $\design D$ and $\design E$ be compatible multi-designs. $\design{Cut}_{\design D|\design E}$ is a multi-design defined by induction on the number of designs in $\design E$ by:
    \begin{align}
      \label{1} \design{Cut}_{\design D|\emptyset} & = \design D \\
      \label{2} \design{Cut}_{\design D|\design E' \cup \{\design p\}} & = \design{Cut}_{(\design D \setminus S) \cup \{\design p[S]\}|\design E'} \\
      \label{3} \design{Cut}_{\design D|\design E' \cup \{\design n/x\}} & = \design{Cut}_{(\design D \setminus S) \cup \{\design n[S]/x\}|\design E'} & \mbox{ if } x \notin \mathrm{fv}(\design D) \\
      \label{4} \design{Cut}_{\design D|\design E' \cup \{\design n/x\}} & = \design{Cut}_{(\design D \setminus S)[\design n[S]/x]|\design E'} & \mbox{ if } x \in \mathrm{fv}(\design D)
    \end{align}
    \begin{tabular}{ll}
    where $S$ & $= \setst{(\design m/y) \in \design D}{y \in \mathrm{fv}(\design p)}$ in (\ref{2}) \\
    & $= \setst{(\design m/y) \in \design D}{y \in \mathrm{fv}(\design n)}$ in (\ref{3}) and (\ref{4}).
    \end{tabular}
\end{definition}

The successive pairs of compatible (resp. closed-compatible) multi-designs stay compatible (resp. closed-compatible) after one step of the definition, thus this is well defined. Moreover, if $\design D$ and $\design E$ are closed-compatible then, according to the base case, $\design{Cut}_{\design D|\design E}$ will be a closed design.
In particular, if $\design p$ and $\design n$ are atomic designs then $\design{Cut}_{\design p|\design n} = \design p[\design n/x_0]$.

In order to prove an \imp{associativity} theorem for multi-designs, recall first the original theorem on designs:
\begin{theorem}[Associativity] \label{assoc}
  Let $\design d$ be a design and $\design n_1, \dots, \design n_k$ be negative designs.
  \[\normalisation{\design d[\design n_1/y_1, \dots, \design n_k/y_k]} = \normalisation{\normalisation{\design d}[\normalisation{\design n_1}/y_1, \dots, \normalisation{\design n_k}/y_k]}.\]
\end{theorem}
This result was first established by Girard \cite{Girard1}. The theorem, in the form given above, was proved by Basaldella and Terui \cite{BT}. Associativity naturally extends to multi-designs as follows:

\begin{theorem}[Multi-associativity] \label{assoc2}
Let $\design D$ be a multi-design and $\design n_1, \dots, \design n_k$ be negative designs.
\[\normalisation{\design D[\design n_1/y_1, \dots, \design n_k/y_k]} = \normalisation{\normalisation{\design D}[\normalisation{\design n_1}/y_1, \dots, \normalisation{\design n_k}/y_k]}.\]
\end{theorem}

\begin{proof}
Immediate from the definition of the normal form of a multi-design (Definition~\ref{multi-normal}) and simple associativity (Theorem~\ref{assoc}).
\end{proof}

\begin{corollary} \label{multi-assoc}
\[\normalisation{\cut{\design D}{\design E}} = \normalisation{\cut{\normalisation{\design D}}{\normalisation{\design E}}}\]
\end{corollary}

\begin{proof}
By induction on $\design E$:
\begin{itemize}
\item If $\design E = \emptyset$ then $\normalisation{\design{Cut}_{\design D|\emptyset}} = \normalisation{\design D} = \normalisation{\cut{\normalisation{\design D}}{\emptyset}} = \normalisation{\cut{\normalisation{\design D}}{\normalisation{\emptyset}}}$.
\item If $\design E = \design E' \cup \{\design p\}$, let $S = \{\design m_1/y_1, \dots, \design m_k/y_k\} = \setst{(\design m/y) \in \design D}{y \in \mathrm{fv}(\design p)}$.
  By definitions of the normal form of multi-designs (Definition~\ref{multi-normal}) and of $\cut{.}{.}$ (Definition~\ref{cut}), and using associativity (Theorem~\ref{assoc2}), we have:
  \begin{align*}
    \normalisation{\design{Cut}_{\design D|\design E}} & = \normalisation{\design{Cut}_{(\design D \setminus S) \cup \{\design p[S]\}|\design E'}} & \mbox{by Def.~\ref{cut}} \\
    & = \normalisation{\design{Cut}_{\normalisation{(\design D \setminus S) \cup \{\design p[S]\}}|\normalisation{\design E'}}} & \mbox{ by induction hypothesis}\\
    & = \normalisation{\design{Cut}_{\normalisation{\design D \setminus S} \cup \{\normalisation{\normalisation{\design p}[\normalisation{\design m_1}/y_1, \dots, \normalisation{\design m_k}/y_k]}\}|\normalisation{\design E'}}} & \mbox{by Def.~\ref{multi-normal} and Thm.~\ref{assoc2}} \\ 
    & = \normalisation{\design{Cut}_{\normalisation{\normalisation{\design D \setminus S} \cup \{\normalisation{\design p}[\normalisation{\design m_1}/y_1, \dots, \normalisation{\design m_k}/y_k]\}}|\normalisation{\design E'}}} & \mbox{by Def.~\ref{multi-normal}} \\
    & = \normalisation{\design{Cut}_{\normalisation{\design D \setminus S} \cup \{\normalisation{\design p}[\normalisation{\design m_1}/y_1, \dots, \normalisation{\design m_k}/y_k]\}|\normalisation{\design E'}}} & \mbox{ by induction hypothesis}\\
    & = \normalisation{\cut{\normalisation{\design D}}{\normalisation{\design E}}} & \mbox{by Def.~\ref{multi-normal} and \ref{cut}}
  \end{align*}
\item If $\design E = \design E' \cup \{\design n/x\}$ with $x \notin \mathrm{fv}(\design D)$, similar as above with $S = \setst{(\design m/y) \in \design D}{y \in \mathrm{fv}(\design n)}$.
\item If $\design E = \design E' \cup \{\design n/x\}$ with  $x \in \mathrm{fv}(\design D)$, let $S = \{\design m_1/y_1, \dots, \design m_k/y_k\} = \setst{(\design m/y) \in \design D}{y \in \mathrm{fv}(\design n)}$.
  We have:
    \begin{align*}
    \normalisation{\design{Cut}_{\design D|\design E}} & = \normalisation{\design{Cut}_{(\design D \setminus S)[\design n[S]/x]|\design E'}} & \mbox{by Def.~\ref{cut}} \\
    & = \normalisation{\design{Cut}_{\normalisation{(\design D \setminus S)[\design n[S]/x]}|\normalisation{\design E'}}} & \mbox{ by induction hypothesis}\\
    & = \normalisation{\design{Cut}_{\normalisation{\normalisation{\design D \setminus S}[\normalisation{\normalisation{\design n}[\normalisation{\design m_1}/y_1, \dots, \normalisation{\design m_k}/y_k]}/x]}|\normalisation{\design E'}}} & \mbox{by using Thm.~\ref{assoc2} twice} \\ 
    & = \normalisation{\design{Cut}_{\normalisation{\normalisation{\design D \setminus S}[\normalisation{\design n}[\normalisation{\design m_1}/y_1, \dots, \normalisation{\design m_k}/y_k]/x]}|\normalisation{\design E'}}} & \mbox{by Thm.~\ref{assoc2}} \\
    & =  \normalisation{\design{Cut}_{\{\normalisation{\design D \setminus S}[\normalisation{\design n}[\normalisation{\design m_1}/y_1, \dots, \normalisation{\design m_k}/y_k]/x]|\normalisation{\design E'}}} & \mbox{ by induction hypothesis}\\
    & = \normalisation{\cut{\normalisation{\design D}}{\normalisation{\design E}}} & \mbox{by Def.~\ref{multi-normal} and \ref{cut}}
  \end{align*}
\end{itemize}
\end{proof}

\begin{lemma} \label{cut-commut}
$\design{Cut}_{\design D|\design E} = \design{Cut}_{\design E|\design D}$.
\end{lemma}

\begin{proof} By induction on the number $n$ of variables in $(\mathrm{fv}(\design D) \cap \mathrm{np}(\design E)) \cup (\mathrm{fv}(\design E) \cap \mathrm{np}(\design D))$.
  \begin{itemize}
  \item If $n = 0$ then $\design{Cut}_{\design D|\design E} = \design{Cut}_{\design E|\design D} = \design D \cup \design E$.
  \item Let $n > 0$ and suppose the property is satisfied for all $k < n$. Without loss of generality suppose there exists $x \in (\mathrm{fv}(\design D) \cap \mathrm{np}(\design E))$. Thus $\design E$ is of the form $\design E = \design E' \cup \{\design n/x\}$. Let $S = \setst{(\design m/y) \in \design D}{y \in \mathrm{fv}(\design n)}$. 
    \begin{itemize}
    \item If $S = \emptyset$, let $\design d \in \design D$ be the design such that $x \in \mathrm{fv}(\design d)$, and let us write $\design D = \design D' \cup \{\design d\}$. If $\design d$ is positive then:
      \begin{align*}
        \cut{\design D}{\design E} & = \cut{\design D' \cup \{\design d[\design n/x]\}}{\design E'} & \mbox{ by one step~\ref{4} of Def.~\ref{cut}}  \\
        & = \cut{\design E'}{\design D' \cup \{\design d[\design n/x]\}} & \mbox{ by induction hypothesis} \\
        & = \cut{(\design E'\setminus T') \cup \{\design d[\design n/x, T']\}}{\design D'} & \mbox{ by one step~\ref{2} of Def.~\ref{cut}}
      \end{align*}
        where $T' = \setst{(\design m/y) \in \design E'}{y \in \mathrm{fv}(\design d[\design n/x])}$. Let $T = \setst{(\design m/y) \in \design E}{y \in \mathrm{fv}(\design d)}$, we have $T = T' \cup \{\design n/x\}$, indeed: $\mathrm{fv}(\design d[\design n/x]) = (\mathrm{fv}(\design d) \setminus \{x\}) \cup \mathrm{fv}(\design n)$, where $\mathrm{fv}(\design n) \cap \mathrm{np}(\design E) = \emptyset$ by definition of a multi-design, thus also $\mathrm{fv}(\design n) \cap \mathrm{np}(\design E') = \emptyset$. Therefore:
        \[\cut{(\design E'\setminus T') \cup \{\design d[\design n/x, T']\}}{\design D'} = \cut{(\design E\setminus T) \cup \{\design d[T]\}}{\design D'} = \cut{\design E}{\design D}\]
by one step~\ref{2} of Def.~\ref{cut} backwards, hence the result. The reasoning is similar if $\design d$ is negative and $\design D = \design D' \cup \{\design d/y\}$, we just have to distinguish between the cases $y \in \mathrm{fv}(\design E')$ and $y \notin \mathrm{fv}(\design E')$.
    \item Otherwise, let $S' = \setst{(\design m/y) \in \design E}{y \in \mathrm{fv}(S)}$ and $S'' = \setst{(\design m/y) \in \design D}{y \in \mathrm{fv}(S')}$; note that $S' \subseteq \design E'$ and $S'' \subseteq (\design D \setminus S)$. We have: 
      \begin{align*}
        \cut{\design E}{\design D} & = \cut{(\design E'\setminus S') \cup \{\design n[S[S']]\}}{\design D \setminus S} & \mbox{ by several steps~\ref{4} of Def.~\ref{cut}}  \\
        & = \cut{\design D \setminus S}{(\design E'\setminus S') \cup \{\design n[S[S']]\}} & \mbox{ by induction hypothesis, since } S \neq \emptyset \\
        & = \cut{(\design D \setminus (S \cup S''))[\design n[S[S'[S'']]]/x]}{\design E'\setminus S'} & \mbox{ by one step~\ref{4} of Def.~\ref{cut}} \\
        & = \cut{\design D}{\design E} & \mbox{ by steps~\ref{4} of Def.~\ref{cut} backwards}
      \end{align*}
      The last equality is obtained by moving successively, from left to right, all the designs from $S'$, and finally the design $\design n$.
    \end{itemize}
  \end{itemize}
\end{proof}

\begin{lemma} \label{cutcut}
Let $\design D_1$, $\design D_2$ and $\design E$ be multi-designs such that $\design D_1 \cup \design D_2$ is a multi-design, and $\design E$ is compatible with $\design D_1 \cup \design D_2$. We have:
\[ \cut{\design D_1 \cup \design D_2}{\design E} = \cut{\design D_1}{\cut{\design E}{\design D_2}} \]
\end{lemma}

\begin{proof}
By induction on $\design D_2$:
\begin{itemize}
\item If $\design D_2 = \emptyset$ then $\cut{\design E}{\design D_2} = \design E$ hence the result.
\item If $\design D_2 = \design D_2' \cup \{\design p\}$ then $\cut{\design E}{\design D_2} = \cut{(\design E \setminus S) \cup \{\design p[S]\}}{\design D_2'}$ where $S = \setst{(\design m/y) \in \design E}{y \in \mathrm{fv}(\design p)}$. Thus by induction hypothesis:
\begin{align*}
\cut{\design D_1}{\cut{\design E}{\design D_2}} & = \cut{\design D_1 \cup \design D_2'}{(\design E \setminus S) \cup \{\design p[S]\}} & \\
& = \cut{((\design D_1 \cup \design D_2')\setminus S') \cup \{\design p[S[S']]\}}{\design E \setminus S} & \mbox{ by one step~\ref{2} of Def.~\ref{cut}} \\
& = \cut{((\design D_1 \cup \design D_2) \setminus S')[S[S']]}{\design E \setminus S} &
\end{align*}
where $S' = \setst{(\design m/y) \in (\design D_1 \cup \design D_2)}{y \in \mathrm{fv}(S)}$.
 Finally, by several steps~\ref{4} backwards of Definition~\ref{cut}, this is equal to $\cut{\design D_1 \cup \design D_2}{\design E}$.
\item If $\design D_2 = \design D_2' \cup \{\design n/x\}$ and $x \notin \mathrm{fv}(\design E)$, then similar to previous case.
\item If $\design D_2 = \design D_2' \cup \{\design n/x\}$ and $x \in \mathrm{fv}(\design E)$, then $\cut{\design E}{\design D_2} = \cut{(\design E \setminus S)[\design n[S]/x]}{\design D_2'}$ where $S = \setst{(\design m/y) \in \design E}{y \in \mathrm{fv}(\design n)}$. Thus by induction hypothesis:
\begin{align*}
\cut{\design D_1}{\cut{\design E}{\design D_2}} & = \cut{\design D_1 \cup \design D_2'}{(\design E \setminus S)[\design n[S]/x]} & \\
& = \cut{(\design E \setminus S)[\design n[S]/x]}{\design D_1 \cup \design D_2'} & \mbox{ by Lemma~\ref{cut-commut}} \\
& = \cut{\design E}{\design D_1 \cup \design D_2} & \mbox{ by one step~\ref{4} backwards of Def.~\ref{cut}} \\
& = \cut{\design D_1 \cup \design D_2}{\design E} & \mbox{ by Lemma~\ref{cut-commut}}
\end{align*}
\end{itemize}
\end{proof}

We now extend the notion of orthogonality to multi-designs.

\begin{definition} \label{ortho}
  Let $\design D$ and $\design E$ be closed-compatible multi-designs. $\design D$ and $\design E$ are \defined{orthogonal}, written $\design D \perp \design E$, if $\design{Cut}_{\design D|\design E} \Downarrow \daimon$.
\end{definition}

\subsection{Paths and Multi-Designs}

Recall that we write $\epsilon$ for the empty sequence.

\begin{definition}
  Let $\design D$ be a multi-design.
  \begin{itemize}
  \item A \defined{view of} $\design D$ is a view of a design in $\design D$.
  \item A \defined{path of} $\design D$ is a path $\pathLL s$ of same polarity as $\design D$ such that for all prefix $\pathLL s'$ of $\pathLL s$, $\view{\pathLL s'}$ is a view of $\design D$.
  \end{itemize}
\end{definition}

We are now interested in a particular form of closed interactions, where we can identify two sides of the multi-design: designs are separated into two groups such that there are no cuts between designs of the same group. This corresponds exactly to the interaction between two closed-compatible multi-designs.

\begin{definition} \label{inter-path-multi}
  Let $\design D$ and $\design E$ be closed-compatible multi-designs such that $\design D \perp \design E$. The \defined{interaction path} of $\design D$ with $\design E$ is the unique path $\pathLL s$ of $\design D$ such that $\dual{\pathLL s}$ is a path of $\design E$.
\end{definition}

But nothing ensures the existence and uniqueness of such a path: this will be proved in the rest of this subsection. We will moreover show that, if $\design D \perp \design E$, this path corresponds to the \imp{interaction sequence} defined below. For the purpose of giving an inductive definition of the interaction sequence, we define it not only for a pair of closed-compatible multi-designs but for a larger class of pairs of multi-designs.

\begin{definition} \label{inter-seq}
  Let $\design D$ and $\design E$ be multi-designs of opposite polarities, compatible, and satisfying $\mathrm{fv}(\design D) \subseteq \mathrm{np}(\design E)$ and $\mathrm{fv}(\design E) \subseteq \mathrm{np}(\design D)$. The \defined{interaction sequence} of $\design D$ with $\design E$, written $\interseq{\design D}{\design E}$, is the sequence of actions followed by interaction on the side of $\design D$. More precisely, if we write $\design p$ for the only positive design of $\design D \cup \design E$, the interaction sequence is defined recursively by:
  \begin{itemize}
  \item If $\design p = \daimon$ then:
    \begin{itemize}
    \item $\interseq{\design D}{\design E} = \daimon$ if $\daimon \in \design D$
    \item $\interseq{\design D}{\design E} = \epsilon$ if $\daimon \in \design E$.
    \end{itemize}
  \item If $\design p = \Omega$ then $\interseq{\design D}{\design E} = \epsilon$.
  \item If $\design p = \posdes{x}{a}{\vect{\design m}}$ then there exists $\design n$ such that $(\design n/x) \in \design E$ if $\design p \in \design D$, $(\design n/x) \in \design D$ otherwise. Let us write $\design n = \negdes{b}{\vect{y^b}}{\design p_b}$. We have $\interseq{\design D}{\design E} = \kappa \interseq{\design D'}{\design E'}$ where:
    \begin{itemize}
    \item if $\design p \in \design D$, then $\kappa = \posdes{x}{a}{\vect{y^a}}$, $\design D' = (\design D \setminus \{\design p\}) \cup \{\vect{\design m/y^a}\}$ and $\design E' = {(\design E \setminus \{\design n/x\}) \cup \{\design p_a\}}$.
    \item otherwise $\kappa = a_x(\vect{y^a})$, $\design D' = (\design D \setminus \{\design n/x\}) \cup \{\design p_a\}$ and $\design E' = (\design E \setminus \{\design p\}) \cup \{\vect{\design m/y^a}\}$.
    \end{itemize}
  \end{itemize}
\end{definition}
Note that this applies in particular to two closed-compatible multi-designs. Remark also that this definition follows exactly the interaction between $\design D$ and $\design E$: indeed, in the inductive case of the definition, the multi-designs $\design D'$ and $\design E'$ are obtained from $\design D$ and $\design E$ similarly to the following lemma. In particular the interaction sequence is finite whenever the interaction between $\design D$ and $\design E$ is finite.

\begin{lemma} \label{lem-norm}
  Let $\design D$ and $\design E$ be closed-compatible multi-designs of opposite polarities. Suppose the only positive design $\design p \in \design D$ is of the form $\design p = \posdes{x}{a}{\vect{\design n}}$, and suppose moreover there exists $\design n_0$ such that $(\design n_0/x) \in \design E$, say $\design n_0 = \negdes{b}{\vect{x^b}}{\design p_b}$. Then:
  \[\cut{\design D}{\design E} 
		\leadsto 
	\cut{\design D'}{\design E'} \setminus \setst{(\design m/x_i^a)}{x_i^a \notin \mathrm{fv}(\design p_a)} \]
  where $\design D' = (\design D \setminus \{\design p\}) \cup \{\vect{\design n/x^a}\}$ and $\design E' = (\design E \setminus \{\design n_0/x\}) \cup \{\design p_a\}$.
\end{lemma}

\begin{proof}
  First notice that as $\design D$ and $\design E$ are closed-compatible, $\cut{\design D}{\design E}$ is a design, and since this design has cuts we can indeed apply one step of reduction to it. Let $S' = \setst{(\design m/x_i^a)}{x_i^a \notin \mathrm{fv}(\design p_a)})$. We have to prove $\cut{\design D}{\design E} \leadsto \cut{\design D'}{\design E'} \setminus S'$. The proof is done by induction on the number of designs in $\design E$.
  \begin{itemize}
  \item If $\design E = \{\design n_0/x\}$, then $\design E' = \{\design p_a\}$.
    In this case let $S= \setst{(\design m/y) \in \design D}{y \in \mathrm{fv}(\design n_0)}$, and remark that, as $\design E$ and $\design D$ are closed-compatible, $S = \design D \setminus \{\design p\}$.
    Thus:
    \begin{align*}
      \cut{\design D}{\design E} & = \cut{(\design D \setminus S)[\design n_0[S]/x]}{\emptyset} & \mbox{ by one step~\ref{4} of Def.~\ref{cut}} \\
      & = \design p[\design n_0[S]/x] \\
      & \leadsto \design p_a[S][\vect{\design n/x^a}] \\
      & = \design p_a[\design D'] \\
      & = \{\design p_a[\design D' \setminus S'_0]\} \cup S'_0 \setminus S' & \mbox{where } S'_0 = \proj{S'}{\design D'} \\
      & = \cut{S'_0 \cup \{\design p_a[\design D' \setminus S'_0]\}}{\emptyset} \setminus S'\\
      & = \cut{\design D'}{\design p_a} \setminus S' & \mbox{ by one step~\ref{2} backwards of Def.~\ref{cut}} \\
      & = \cut{\design D'}{\design E'} \setminus S'
    \end{align*}

  \item Otherwise there exists $(\design n_1/z) \in \design E$ such that $x \neq z$. Suppose $z \notin \mathrm{fv}(\design D)$ (resp. $z \in \mathrm{fv}(\design D)$).
    Define:
    \begin{itemize}
    \item $S = \setst{(\design m/y) \in \design D}{y \in \mathrm{fv}(\design n_1)}$, and remark that $S = \setst{(\design m/y) \in \design D'}{y \in \mathrm{fv}(\design n_1)}$.
    \item $\design D'' = (\design D' \setminus S)\cup \{(\design n_1[S]/z)\}$ (resp. $\design D'' = (\design D' \setminus S)[\design n_1[S]/z]$)
    \item $\design E'' = \design E' \setminus \{(\design n_1/z)\}$.
  \end{itemize}
    We have:
    \begin{align*}
      \cut{\design D}{\design E} & = \cut{(\design D \setminus S)\cup \{(\design n_1[S]/z)\}}{\design E \setminus \{\design n_1/z\}} & \mbox{by one step~\ref{3} of Def.~\ref{cut}} \\
      \mbox{( \hspace{.2cm} resp. } & = \cut{(\design D \setminus S)[(\design n_1[S]/z)]}{\design E \setminus \{\design n_1/z\}} & \mbox{by one step~\ref{4} of Def.~\ref{cut} \hspace{.2cm})} \\
      & \leadsto \cut{\design D''}{\design E''} \setminus S' & \mbox{by induction hypothesis} & \\
      & = \cut{\design D'}{\design E'} \setminus S' & \mbox{by step~\ref{3} (resp. \ref{4}) of Def.~\ref{cut} backwards}
    \end{align*}
  \end{itemize}
\end{proof}

\begin{lemma} \label{inter-dual}
If $\daimon \in \normalisation{\cut{\design D}{\design E}}$ (in particular if $\design D \perp \design E$) then $\interseq{\design D}{\design E} = \dual{\interseq{\design E}{\design D}}$. Otherwise $\interseq{\design D}{\design E} = \overline{\interseq{\design E}{\design D}}$.
\end{lemma}

\begin{proof}
It is clear from the definition of the interaction sequence that the proper actions in $\interseq{\design D}{\design E}$ are the opposite of those in $\interseq{\design E}{\design D}$. Concerning the daimon: since the interaction sequence follows the interaction between $\design D$ and $\design E$, $\daimon$ appears at the end of one of the sequences $\interseq{\design D}{\design E}$ or $\interseq{\design E}{\design D}$ if and only if $\daimon \in \normalisation{\cut{\design D}{\design E}}$, and in this case $\interseq{\design D}{\design E} = \dual{\interseq{\design E}{\design D}}$.
\end{proof}

\begin{proposition} \label{pref-inter-path}
  Every positive-ended prefix of $\interseq{\design D}{\design E}$ is a path of $\design D$. In particular, if $\interseq{\design D}{\design E}$ is finite and positive-ended then it is a path of $\design D$.
\end{proposition}

\begin{proof}
  First remark that every (finite) prefix of $\interseq{\design D}{\design E}$ is an aj-sequence. Indeed, since $\design D$ and $\design E$ are well shaped multi-designs the definition of interaction sequence ensures that an action cannot appear before its justification, and all the conditions of the definition of an aj-sequence are satisfied: \imp{Alternation} and \imp{Daimon} are immediate from the definition of interaction sequence, while \imp{Linearity} is indeed satisfied as variables are disjoint in $\design D$ and $\design E$ (by Barendregt's convention).

  By definition, for every prefix $\pathLL s$ of $\interseq{\design D}{\design E}$, $\view{\pathLL s}$ is a view. We show that it is a view of $\design D$ by induction on the length of $\pathLL s$:
  \begin{itemize}
  \item If $\pathLL s = \epsilon$ then $\view{\epsilon} = \epsilon$ is indeed a view of $\design D$.
  \item If $\pathLL s = \daimon$ then $\interseq{\design D}{\design E} = \daimon$. From the definition of interaction sequence, we know that in this case $\daimon \in \design D$, hence $\view{\daimon} = \daimon$ is a view of $\design D$.
  \item If $\pathLL s = \kappa\pathLL s'$ where $\kappa$ is proper, then $\interseq{\design D}{\design E} = \kappa\interseq{\design D'}{\design E'}$ where $\design D'$ and $\design E'$ are as in Definition~\ref{inter-seq}, and $\pathLL s'$ is a prefix of $\interseq{\design D'}{\design E'}$. By induction hypothesis, $\view{\pathLL s'}$ is a view of $\design D'$. Two possibilities:
    \begin{itemize}
    \item Either $\kappa = \kappa^+$ is positive.
From the definition of the interaction sequence, this means $\design p := \posdes{x}{a}{\vect{\design m}} \in \design D$, $\kappa^+ = \posdes{x}{a}{\vect{y^a}}$ and $\design D' = (\design D \setminus \{\design p\}) \cup \{\vect{\design m/y^a}\}$. We have $\view{\pathLL s} = \view{\kappa^+ \pathLL s'}$ and either $\view{\kappa^+ \pathLL s'} = \kappa^+ \view{\pathLL s'}$ if the first negative action of $\view{\pathLL s'}$ is justified by $\kappa^+$ (i.e., $\exists i$ such that $\view{\pathLL s'}$ is a view of $\design m_i/y^a_i$), or $ \view{\kappa^+ \pathLL s'} = \view{\pathLL s'}$ otherwise (i.e., $\view{\pathLL s'}$ is a view of $\design D \setminus \{\design p\}$). In the second case, there is nothing more to show; in the first one, by definition of the views of a design, $\kappa^+ \view{\pathLL s'}$ is a view of $\design p = \posdes{x}{a}{\vect{\design m}}$.
    \item Or $\kappa = \kappa^-$ is negative. Hence there exists a design $\design n = \negdes{b}{\vect{y^b}}{\design p_b}$ such that $(\design n/x) \in \design D$, $\kappa^- = a_x(\vect{y^a})$, and $\design D' = (\design D  \setminus \{\design n/x\}) \cup \{\design p_a\}$. We have $\view{\pathLL s} = \view{\kappa^- \pathLL s'}$ and either $\view{\kappa^- \pathLL s'} = \kappa^- \view{\pathLL s'}$ if the first action of $\view{\pathLL s'}$ is positive (i.e., $\view{\pathLL s'}$ is a view of $\design p_a$), or $ \view{\kappa^- \pathLL s'} = \view{\pathLL s'}$ otherwise (i.e., $\view{\pathLL s'}$ is a view of $\design D' \setminus \{\design p_a\} \subseteq \design D$). In the second case, there is nothing to do; in the first one, note that $\kappa^- \view{\pathLL s'}$ is a view of $(\design n/x)$, hence the result.
    \end{itemize}
  \end{itemize}
We have proved that $\view{\pathLL s}$ is a view of $\design D$. This implies in particular that $\interseq{\design D}{\design E}$ satisfies P-visibility, indeed: given a prefix $\pathLL s \kappa^+$ of $\interseq{\design D}{\design E}$, the action $\kappa^+$ is either initial or it is justified in $\pathLL s$ by the same action that justifies it in $\design D$; since $\view{\pathLL s}$ is a view of $\design D$, the justification of $\kappa^+$ is in it, thus P-visibility is satisfied. Similarly, we can prove that $\view{\pathLL t}$ is a view of $\design E$ whenever $\pathLL t$ is a prefix of $\interseq{\design E}{\design D}$, therefore $\interseq{\design E}{\design D}$ also satisfies P-visibility; by Lemma~\ref{inter-dual} either $\interseq{\design E}{\design D} = \dual{\interseq{\design D}{\design E}}$ or $\interseq{\design E}{\design D} = \overline{\interseq{\design D}{\design E}}$, thus this implies that $\interseq{\design D}{\design E}$ satisfies O-visibility. Hence every positive-ended prefix of $\interseq{\design D}{\design E}$ is a path, and since the views of its prefixes are views of $\design D$, it is a path of $\design D$.
\end{proof}

\begin{remark}\label{differ-neg}
  If $\pathLL s\kappa_1^+$ and $\pathLL s\kappa_2^+$ are views (resp. paths) of a multi-design $\design D$ then $\kappa_1^+ = \kappa_2^+$. Indeed, if $\pathLL s\kappa_1^+$ and $\pathLL s\kappa_2^+$ are views of $\design D$, the result is immediate by definition of the views of a design; if they are paths of $\design D$, just remark that $\view{\pathLL s\kappa_1^+} = \view{\pathLL s}\kappa_1^+$ and $\view{\pathLL s\kappa_2^+} = \view{\pathLL s}\kappa_2^+$ are views of $\design D$, hence the conclusion.
\end{remark}

\begin{proposition} \label{prefix-norm}
  Suppose $\design D \perp \design E$, $\pathLL s$ is a path of $\design D$ and $\overline{\pathLL s}$ is a path of $\design E$. The path $\pathLL s$ is a prefix of $\interseq{\design D}{\design E}$.
\end{proposition}

\begin{proof}
  Suppose $\pathLL s$ is not a prefix of $\interseq{\design D}{\design E}$. Let $\pathLL t$ be the longest common prefix of $\pathLL s$ and $\interseq{\design D}{\design E}$ (possibly $\epsilon$). Without loss of generality, we can assume there exist actions of same polarity $\kappa_1$ and $\kappa_2$ such that $\kappa_1 \neq \kappa_2$, $\pathLL t \kappa_1$ is a prefix of $\pathLL s$ and $\pathLL t \kappa_2$ is a prefix of $\interseq{\design D}{\design E}$: indeed, if there are no such actions, it is because $\interseq{\design D}{\design E}$ is a strict prefix of $\pathLL s$; in this case, it suffices to consider $\interseq{\design E}{\design D}$ and $\overline{\pathLL s}$ instead.
  \begin{itemize}
  \item If $\kappa_1$ and $\kappa_2$ are positive, then $\pathLL t \kappa_1$ and $\pathLL t \kappa_2$ are paths of $\design D$, and by previous remark $\kappa_1 = \kappa_2$: contradiction.
  \item If $\kappa_1$ and $\kappa_2$ are negative, a contradiction arises similarly from the fact that $\overline{\pathLL t \kappa_1}$ and $\overline{\pathLL t \kappa_2}$ are paths of $\design E$ where $\overline{\kappa_1}$ and $\overline{\kappa_2}$ are positive.
  \end{itemize}
Hence the result.
\end{proof}

The following result ensures that the interaction path is well defined.

\begin{proposition} \label{inter-unique}
  If $\design D \perp \design E$, there exists a unique path $\pathLL s$ of $\design D$ such that $\dual{\pathLL s}$ is a path of $\design E$.
\end{proposition}

\begin{proof}
Lemma~\ref{inter-dual} and Proposition~\ref{pref-inter-path} show that such a path exists, namely $\interseq{\design D}{\design E}$. Unicity follows from Proposition~\ref{prefix-norm}.
\end{proof}
Conversely, we prove that the existence of such a path implies the orthogonality of multi-designs (Proposition~\ref{perp-path}).

\begin{proposition} \label{path-perp}
  Let $\design P$ and $\design N$ be closed-compatible multi-designs such that $\Omega \notin \design P$ and such that their interaction is finite.
 Suppose that for every path $\pathLL s\kappa^+$ of $\design P$ such that $\kappa^+$ is proper and $\overline{\pathLL s}$ is a path of $\design N$, $\overline{\pathLL s\kappa^+}$ is a path of $\design N$, and suppose also that the same condition is satisfied when reversing $\design P$ and $\design N$. Then $\design P \perp \design N$.
\end{proposition}

\begin{proof}
  By induction on the number $n$ of steps of the interaction before divergence/convergence:
  \begin{itemize}
  \item If $n = 0$, then we must have $\design P = \daimon$, since $\Omega \notin \design P$. Hence the result.
  \item If $n > 0$ then $\design p \in \design P$ is of the form $\design p = \posdes{x}{a}{\vect{\design n}}$ and there exists $\design n_0 = \negdes{b}{\vect{x^b}}{\design p_b}$ such that $(\design n_0/x) \in \design N$. Let $\kappa^+ = \posdes{x}{a}{\vect{x^a}}$ and remark that $\kappa^+$ is a path of $\design p$. By hypothesis, $\overline{\kappa^+} = a_x(\vect{x^a})$ is a path of $\design N$, thus a path of $\design n_0$, and this implies $\design p_a \neq \Omega$. By Lemma~\ref{lem-norm}, we have $\cut{\design P}{\design N} \leadsto \cut{\design P'}{\design N'} \setminus \setst{(\design m/x^a_i)}{x_i \notin \mathrm{fv}(\design p_a)}$ where $\design P' = (\design P \setminus \{\design p\}) \cup \{\vect{\design n/x^a}\}$ and $\design N' = (\design N \setminus \{\design n_0/x\}) \cup \{\design p_a\}$. This corresponds to a cut-net between two closed-compatible multi-designs $\design P'' \subseteq \design P'$ (negative) and $\design N'' \subseteq \design N'$ (positive), where:
    \begin{itemize}
    \item $\Omega \notin \design N''$ because $\design p_a \neq \Omega$;
    \item their interaction is finite and takes $n-1$ steps;
    \item the condition on paths stated in the proposition is satisfied for $\design P''$ and $\design N''$, because it is for $\design P$ and $\design N$: indeed, the paths of $\design P''$ (resp. $\design N''$) are the paths $\pathLL t$ such that $\kappa^+ \pathLL t$ is a path of $\design P$ (resp. $\overline{\kappa^+} \pathLL t$ is a path of $\design N$), unless such a path $\pathLL t$ contains a negative initial action whose address is not the address of a positive action on the other side, but this restriction is harmless with respect to our condition.
    \end{itemize}
    We apply the induction hypothesis to get $\design P'' \perp \design N''$. Finally $\design P \perp \design N$. 
  \end{itemize}
\end{proof}

\begin{proposition} \label{perp-path}
  Let  $\design D$ and $\design E$ be closed-compatible multi-designs. $\design D \perp \design E$ if and only if there exists a path $\pathLL s$ of $\design D$ such that $\dual{\pathLL s}$ is a path of $\design E$.
\end{proposition}

\begin{proof}
  \noindent $(\Rightarrow)$ If $\design D \perp \design E$ then result follows from Proposition~\ref{inter-unique}.
  
  \noindent $(\Leftarrow)$ We will prove that the hypothesis of Proposition~\ref{path-perp} is satisfied. Let us show that every path of $\design D$ (resp. of $\design E$) of the form $\pathLL t\kappa^+$ where $\kappa^+$ is proper and $\overline{\pathLL t}$ is a path of $\design E$ (resp. of $\design D$) is a prefix of $\pathLL s$ (resp. of $\overline{\pathLL s}$). By induction on the length of $\pathLL t$, knowing that it is either empty or negative-ended:
  \begin{itemize}
  \item If $\pathLL t$ is empty, $\kappa^+$ is necessarily the first action of the positive design in $\design D$ (resp. in $\design E$), hence the first action of $\pathLL s$ (resp. of $\overline{\pathLL s}$).

  \item If $\pathLL t = \pathLL t_0\kappa^-$, then $\overline{\pathLL t_0\kappa^-}$ is a path of $\design E$ (resp. of $\design D$) and $\pathLL t_0$ is a path of $\design D$ (resp. of $\design E$). By induction hypothesis, $\overline{\pathLL t} = \overline{\pathLL t_0\kappa^-}$ is a prefix of $\overline{\pathLL s}$ (resp. of $\pathLL s$), thus $\pathLL t$ is a prefix of $\pathLL s$ (resp. of $\overline{\pathLL s}$). The path $\pathLL s$ is of the form $\pathLL s = \pathLL t\kappa^{\prime+}\pathLL s'$. But since $\pathLL s$ and $\pathLL t\kappa^+$ are both paths of $\design D$ (resp. $\design E$), they cannot differ on a positive action, hence $\kappa^+ = \kappa^{\prime+}$. Thus $\pathLL t\kappa^+$ is a prefix of $\pathLL s$. 
  \end{itemize}
\end{proof}

\subsection{Associativity for Interaction Paths}

If $\pathLL s$ is a path of a multi-design $\design D$, and $\design E \subseteq \design D$, then we write $\proj{\pathLL s}{\design E}$ for the longest subsequence of $\pathLL s$ that is a path of $\design E$. Note that this is well defined.

\begin{proposition}[Associativity for paths] \label{asso-path}
    Let $\design D$, $\design E$ and $\design F$ be cut-free multi-designs such that $\design E \cup \design F$ is a multi-design, and suppose $\design D \perp (\design E \cup \design F)$. We have:
\[\interseq{\design E}{\normalisation{\cut{\design F}{\design D}}} = \proj{\interseq{\design E \cup \design F}{\design D}}{\design E}\]
\end{proposition}

\begin{proof}
  We will prove the result for a larger class of multi-designs. Instead of the assumption $\design D \perp (\design E \cup \design F)$, suppose that $\design D$ and $\design E \cup \design F$ are:
  \begin{itemize}
  \item of opposite polarities
  \item compatible
  \item satisfying $\mathrm{fv}(\design D) \subseteq \mathrm{np}(\design E \cup \design F)$ and $\mathrm{fv}(\design E \cup \design F) \subseteq \mathrm{np}(\design D)$
  \item and such that $\daimon \in \normalisation{\cut{\design E \cup \design F}{\design D}}$ (in particular their interaction is finite).
  \end{itemize}
  First remark that $\design F$ and $\design D$ are compatible, hence it is possible to define $\cut{\design F}{\design D}$. Then since $\daimon \in \normalisation{\cut{\design E \cup \design F}{\design D}}$, we have $\daimon \in \normalisation{\cut{\design E}{\normalisation{\cut{\design F}{\design D}}}}$, indeed:
  \begin{align*}
    \normalisation{\cut{\design E \cup \design F}{\design D}} & = \normalisation{\cut{\design E}{\cut{\design F}{\design D}}} & \mbox{ by Lemmas~\ref{cutcut} and \ref{cut-commut}} \\
    & = \normalisation{\cut{\design E}{\normalisation{\cut{\design F}{\design D}}}} & \mbox{ by Corollary~\ref{multi-assoc}}
  \end{align*}
This also shows that $\design E$ and $\normalisation{\cut{\design F}{\design D}}$ are compatible. As they are of opposite polarities and they satisfy the condition on variables, $\langle \design E \leftarrow \normalisation{\cut{\design F}{\design D}} \rangle$ is defined.
  
Let $\pathLL s = \interseq{\design E \cup \design F}{\design D}$, and let us show the result (i.e., $\proj{\pathLL s}{\design E} = \interseq{\design E}{\normalisation{\cut{\design F}{\design D}}}$) by induction on the length of $\pathLL s$, which is finite because the interaction between $\design D$ and $\design E \cup \design F$ is finite.
\begin{itemize}
\item If $\pathLL s = \epsilon$ then necessarily $\daimon \in \design D$ thus also $\daimon \in \normalisation{\cut{\design F}{\design D}}$. Hence $\proj{\pathLL s}{\design E} = \epsilon = \langle \design E \leftarrow \normalisation{\cut{\design F}{\design D}} \rangle$.
\item If $\pathLL s = \daimon$ then $\daimon \in \design E \cup \design F$. If $\daimon \in \design E$ then $\langle  \design E \leftarrow \normalisation{\cut{\design F}{\design D}} \rangle = \daimon = \proj{\pathLL s}{\design E}$. Otherwise $\daimon \in \design F$, thus $\daimon \in \normalisation{\cut{\design F}{\design D}}$ and $\interseq{\design E}{\normalisation{\cut{\design F}{\design D}}} = \epsilon = \proj{\pathLL s}{\design E}$.
\item If $\pathLL s = \kappa^+ \pathLL s'$ where $\kappa^+ = \posdes{x}{a}{\vect{x^a}}$ is a proper positive action, then $\design E \cup \design F$ is a positive multi-design such that its only positive design is of the form $\design p = \posdes{x}{a}{\vect{\design m}}$. Thus $\design D$ is negative, and there exists $\design n$ such that $(\design n/x) \in \design D$ of the form $\design n = \negdes{b}{\vect{x^b}}{\design p_b}$, where $\design p_a \neq \Omega$ because the interaction converges. Let $\design D' = (\design D \setminus \{\design n/x\}) \cup \{\design p_a\}$.
  \begin{itemize}
  \item Either $\design p \in \design F$ [\textbf{reduction step}]. \\ In this case, we have $\proj{\pathLL s}{\design E} = \proj{\pathLL s'}{\design E}$, so let us show that $\proj{\pathLL s'}{\design E} = \langle \design E \leftarrow \normalisation{\cut{\design F}{\design D}} \rangle$. By definition of the interaction sequence, we have $\pathLL s' = \interseq{\design E \cup \design F'}{\design D'}$ where $\design F' = (\design F \setminus \{\design p\}) \cup \{\vect{(\design m/x^a)}\}$. Thus by induction hypothesis $\proj{\pathLL s'}{\design E} = \interseq{\design E}{\normalisation{\cut{\design F'}{\design D'}}}$. But by Lemma~\ref{lem-norm}, $\interseq{\design E}{\normalisation{\cut{\design F'}{\design D'}}} = \interseq{\design E}{\normalisation{\cut{\design F}{\design D}}}$ because the negatives among $\vect{(\design m/x^a)}$ in $\normalisation{\cut{\design F'}{\design D'}}$ will not interfere in the interaction with $\design E$, since the variables $\vect{x^a}$ do not appear in $\design E$. Hence the result.
  \item Or $\design p \in \design E$ [\textbf{commutation step}]. \\ In this case, we have $\proj{\pathLL s}{\design E} = \kappa^+(\proj{\pathLL s'}{\design E})$, and by definition of the interaction sequence $\pathLL s' = \interseq{\design E' \cup \design F}{\design D'}$ where $\design E' = (\design E \setminus \{\design p\}) \cup \{\vect{(\design m/x^a)}\}$. Thus by induction hypothesis $\proj{\pathLL s'}{\design E} = \proj{\pathLL s'}{\design E'} = \interseq{\design E'}{\normalisation{\cut{\design F}{\design D'}}}$. 
    But we have
    \begin{align*}
      \interseq{\design E}{\normalisation{\cut{\design F}{\design D}}} 
      &= \interseq{\design E}{\normalisation{\cut{\design F}{\design D' \cup \{(\design n/x)\} \setminus \{\design p_a\}}}} \\
      &= \interseq{\design E}{\normalisation{\cut{\design F}{\design D'}}\cup \{(\design n'/x)\} \setminus \{\design p'_a\}} \\
      &= \kappa^+\interseq{\design E'}{\normalisation{\cut{\design F}{\design D'}}}
    \end{align*}
    where $\design n'$ is the only negative design of $\normalisation{\cut{\design F}{\design D}}$ on variable $x$, and $\design p_a'$ the only positive design of $\normalisation{\cut{\design F}{\design D'}}$. 
    Hence $\interseq{\design E}{\normalisation{\cut{\design F}{\design D}}} = \kappa^+ (\proj{\pathLL s'}{\design E}) = \proj{\pathLL s}{\design E}$.
  \end{itemize}
\item If $\pathLL s = \kappa^- \pathLL s'$ where $\kappa^- = a_x(\vect{x^a})$, then $\design D$ is positive with only positive design of the form $\design p = \posdes{x}{a}{\vect{\design m}}$, and there exists a negative design $\design n$ such that $(\design n/x) \in \design E \cup \design F$, with $\design n$ of the form $\design n = \negdes{b}{\vect{x^b}}{\design p_b}$ where $\design p_a \neq \Omega$. By definition of the interaction sequence, we have $\pathLL s' = \interseq{((\design E \cup \design F) \setminus \{\design n/x\}) \cup \{\design p_a\}}{\design D'}$ where $\design D' = (\design D \setminus \{\design p\}) \cup \{\vect{(\design m/x^a)}\}$.
  \begin{itemize}\item Either $\design n \in \design F$ [\textbf{reduction step}]. \\ In this case, we have $\proj{\pathLL s}{\design E} = \proj{\pathLL s'}{\design E}$, so let us show that $\proj{\pathLL s'}{\design E} = \langle \design E \leftarrow \normalisation{\cut{\design F}{\design D}} \rangle$. By induction hypothesis $\proj{\pathLL s'}{\design E} = \interseq{\design E}{\normalisation{\cut{\design F'}{\design D'}}}$ where $\design F' = (\design F \setminus \{\design n/x\}) \cup \{\design p_a\}$, and by Lemma~\ref{lem-norm} we deduce $\proj{\pathLL s'}{\design E} = \interseq{\design E}{\normalisation{\cut{\design F}{\design D}}}$, hence the result.
  \item Or $\design n \in \design E$ [\textbf{commutation step}]. \\ In this case, we have $\proj{\pathLL s}{\design E} = \kappa^-(\proj{\pathLL s'}{\design E})$. By induction hypothesis $\proj{\pathLL s'}{\design E} = \proj{\pathLL s'}{\design E'} = \interseq{\design E'}{\normalisation{\cut{\design F}{\design D'}}}$ where $\design E' = (\design E \setminus \{\design n/x\}) \cup \{\design p_a\}$.
    But we have
    \begin{align*}
      \interseq{\design E}{\normalisation{\cut{\design F}{\design D}}} 
	&= \interseq	{\design E}
					{\normalisation
						{\cut	{\design F}
								{\design D' \cup \{\design p\} \setminus \{\vect{(\design m/x^a)}\}}
						}
					} \\
	&= \interseq	{\design E}
					{\normalisation
						{\cut	{\design F}
								{\design D'}
						}  
					  \cup \{\design p'\} \setminus \{\vect{(\design m'/x^a)}\}
					} \\
	&= \kappa^-\interseq	{\design E'}
							{ \normalisation {\cut {\design F}{\design D'}}}
\end{align*}
where $\design p'$ is the only positive design of $\normalisation{\cut{\design F}{\design D}}$, and for each $i \le ar(a)$, $\design m'_i$ is the only negative design of $\normalisation{\cut{\design F}{\design D'}}$ on variable $x^a_i$.
Therefore $\interseq{\design E}{\normalisation{\cut{\design F}{\design D}}} = \kappa^- (\proj{\pathLL s'}{\design E}) = \proj{\pathLL s}{\design E}$, which concludes the proof.
  \end{itemize}
\end{itemize}
\end{proof}

\section{Proofs of Subsection~\ref{sub-regpur}} \label{sec-paths-bis}

We now come back to (non ``multi-'') designs, and we prove:
\begin{itemize}
\item the form of visitable paths for each connective (\textsection~\ref{sec-inc-visit}), which is needed for next point;
\item that (some) connectives preserve regularity (Propositions~\ref{reg-sh}, \ref{reg-tensor}, \ref{reg-arrow}, corresponding to Proposition~\ref{prop_reg_stable}), purity (Proposition~\ref{prop_pure_stable}) and quasi-purity (Proposition~\ref{prop_arrow_princ}).
\end{itemize}

\subsection{Preliminaries}

\subsubsection{Observational Ordering and Monotonicity} \label{ord-mono}

We consider the \defined{observational ordering} $\preceq$ over designs:
$\design d' \preceq \design d$ if $\design d$ can be obtained from $\design d'$ by substituting:
\begin{itemize}
\item positive subdesigns for some occurrences of $\Omega$.
\item $\daimon$ for some positive subdesigns.
\end{itemize}
Remark in particular that for all positive designs $\design p$ and $\design p'$, we have $\Omega \preceq \design p \preceq \daimon$, and if $\design p \sqsubseteq \design p'$ then $\design p \preceq \design p'$. We can now state the \imp{monotonicity} theorem, an important result of ludics. A proof of the theorem formulated in this form is found in \cite{Terui}.
\begin{theorem}[Monotonicity] \label{thm_mono}
  \begin{itemize}
  \item If $\design d \preceq \design e$ and $\design m \preceq \design n$, then $\design d[\design m/x] \preceq \design e[\design n/x]$
  \item If $\design d \preceq \design e$ then $\normalisation{\design d} \preceq \normalisation{\design e}$
  \end{itemize}
\end{theorem}

This means that the relation $\preceq$ compares the likelihood of convergence: if $\design d \perp \design e$ and $\design d \preceq \design d'$ then $\design d' \perp \design e$. In particular, if $\beh B$ is a behaviour, if $\design d \in \beh B$ and $\design d \preceq \design d'$ then $\design d' \in \beh B$.

Remark the following important fact: given a path $\pathLL s$ of some design $\design d$, there is a unique design maximal for $\preceq$ such that $\pathLL s$ is a path of it. Indeed, this design $\completed{\pathLL s}$ is obtained from $\design d$ by replacing all positive subdesigns (possibly $\Omega$) whose first positive action is not in $\pathLL s$ by $\daimon$. Note that, actually, the design $\completed{\pathLL s}$ does not depend on $\design d$ but only on the path $\pathLL s$.

\begin{example}
  Consider design $\design d$ and the path $\pathLL s$ below:
  \begin{align*}
    \design d \hspace{.5cm} & = \hspace{.5cm} \posdes{x}{a}{b(y).(\posdes{y}{e}{}),c().\daimon+d(z).(\posdes{z}{e}{})} \\
    \pathLL s \hspace{.5cm} & = \hspace{.5cm} \posdes{x}{a}{x_1, x_2} \hspace{1cm} b_{x_1}(y) \hspace{1cm} \posdes{y}{e}{} \hspace{1cm} c_{x_2}() \hspace{1cm} \daimon
  \end{align*}
\begin{center}
    \begin{tikzpicture}[grow=up, level distance=3cm, sibling distance = 4.5cm, scale=.4]
      \draw[-]
      
        node (a) {$\posdes{x_0}{a}{x_1,x_2}$}
        child {
          node (d) {$d_{x_2}(y)$}
          child {
            node (e2) {$\posdes{z}{e}{}$}
          }
        }
        child	{
          node (c) {$c_{x_2}()$}
          child	{
            node (dai) {$\daimon$}
          }
        }
        child	{
          node (b) {$b_{x_1}(y)$}
          child	{
            node (e1) {$\posdes{y}{e}{}$}
          }
        }
      ;

        \draw[-, red, thick, rounded corners] ($(a)+(-2.8,-.5)$) -- ($(a)+(-2.8,.5)$) -- ($(b)+(-1.7,-.5)$) -- ($(e1)+(-1.7,0)$) -- ($(e1)+(-1.4,-.3)$);
        \draw[-,red, dashed, thick] ($(e1)+(-1.4,-.3)$) -- ($(c)+(-1.7,-.2)$);
        \draw[->, red, thick, rounded corners] ($(c)+(-1.7,-.2)$) -- ($(c)+(-1.4,-.5)$) -- ($(dai)+(-1.4,.5)$) ;

        \node[red] at ($(b)+(-2.4,1.2)$) {$\pathLL s$};

    \end{tikzpicture}
\end{center}
We have $\completed{\pathLL s} = \posdes{x}{a}{b(y).(\posdes{y}{e}{}) + \sum_{f \neq b}f(\vect{x^f}).\daimon , \negdes{f}{\vect{x^f}}{\daimon}}$
\begin{center}
    \begin{tikzpicture}[grow=up, level distance=3cm, sibling distance = 3.5cm, scale=.4]
      \draw[-]
      
        node (a) {$\posdes{x_0}{a}{x_1,x_2}$}
        child {
          node (d) {$\dots$}
          edge from parent[draw=none]
        }
        child {
          node (d) {$f_{x_2}(\vect{x^f})$}
          child {
            node (e2) {$\daimon$}
          }
        }
        child {
          node (d) {$\dots$}
          edge from parent[draw=none]
        }
        child {
          node (d) {$d_{x_2}(y)$}
          child {
            node (e2) {$\daimon$}
          }
        }
        child	{
          node (c) {$c_{x_2}()$}
          child	{
            node (dai) {$\daimon$}
          }
        }
        child	{
          node (b) {$b_{x_1}(y)$}
          child	{
            node (e1) {$\posdes{y}{e}{}$}
          }
        }
        child {
          node (d) {$\dots$}
          edge from parent[draw=none]
        }
        child {
          node (d) {$f_{x_1}(\vect{x^f})$}
          child {
            node (e2) {$\daimon$}
          }
        }
        child {
          node (d) {$\dots$}
          edge from parent[draw=none]
        }
      ;

        \draw[-, red, thick, rounded corners] ($(a)+(-2.5,-.5)$) -- ($(a)+(-2.5,.5)$) -- ($(b)+(-1.4,-.5)$) -- ($(e1)+(-1.4,0)$) -- ($(e1)+(-1.1,-.3)$);
        \draw[-,red, dashed, thick] ($(e1)+(-1.1,-.3)$) -- ($(c)+(-1.5,-.2)$);
        \draw[->, red, thick, rounded corners] ($(c)+(-1.5,-.2)$) -- ($(c)+(-1.2,-.5)$) -- ($(dai)+(-1.2,.5)$) ;

        \node[red] at ($(b)+(-2.4,1.2)$) {$\pathLL s$};

    \end{tikzpicture}
\end{center}
\end{example}

\begin{proposition} \label{cor-mono}
  If $\pathLL s \in \visit{\beh B}$ then $\completed{\pathLL s} \in \beh B$.
\end{proposition}

\begin{proof}
  There exists $\design d \in \beh B$ such that $\pathLL s$ is a path of $\design d$, thus $\design d \preceq \completed{\pathLL s}$. The result then comes from monotonicity (Theorem~\ref{thm_mono}).
\end{proof}

\subsubsection{More on Paths}

Let $\beh B$ be a behaviour.

\begin{lemma} \label{lem_visit_path_inc}
  If $\design d \in \beh B$ and $\pathLL s \in \visit{B}$ is a path of $\design d$, then $\pathLL s$ is a path of $|\design d|$.
\end{lemma}

\begin{proof}
  Let $\design e \in \beh B^\perp$ such that $\pathLL s = \interseq{\design d}{\design e}$, and let $\pathLL t = \interseq{|\design d|}{\design e}$.
  \begin{itemize}
  \item Since $|\design d| \sqsubseteq \design d$, the path $\pathLL s$ cannot be a strict prefix of $\pathLL t$, and $\pathLL s$ and $\pathLL t$ cannot differ on a positive action.
  \item If $\pathLL t$ is a strict prefix of $\pathLL s$ then it is positive-ended. So $\dual{\pathLL s}$ and $\dual{\pathLL t}$ are paths of $\design e$ differing on a positive action, which is impossible.
  \item If $\pathLL s$ and $\pathLL t$ differ on a negative action, say $\pathLL u\kappa^-_1$ and $\pathLL u\kappa^-_2$ are respective prefixes of $\pathLL s$ and $\pathLL t$ with $\kappa^-_1 \neq \kappa^-_2$, then $\overline{\pathLL u\kappa^-_1}$ and $\overline{\pathLL u\kappa^-_2}$ are paths of $\design e$ differing on a positive action, which is impossible.
  \end{itemize}
  Thus we must have $\pathLL s = \pathLL t$, hence the result.
\end{proof}

\begin{lemma} \label{daimon_visit}
  Let $\pathLL s \in \visit B$. For every positive-ended (resp. negative-ended) prefix $\pathLL s'$ of $\pathLL s$, we have $\pathLL s' \in \visit B$ (resp. $\pathLL s'\daimon \in \visit B$).
\end{lemma}

\begin{proof}
  Let $\pathLL s = \interseq{\design d}{\design e}$ where $\design d \in B$ and $\design e \in B^\perp$, and let $\pathLL s'$ be a prefix of $\pathLL s$.
  \begin{itemize}
  \item If $\pathLL s'$ is negative-ended, let $\kappa^+$ be such that $\pathLL s'\kappa^+$ is a prefix of $\pathLL s$. The action $\kappa^+$ comes from $\design d$. Consider design $\design d'$ obtained from $\design d$ by replacing the positive subdesign of $\design d$ starting on action $\kappa^+$ with $\daimon$. Since $\design d \preceq \design d'$, by monotonicity $\design d' \in \beh B$. Moreover $\pathLL s' \daimon = \interseq{\design d'}{\design e}$, hence the result.
  \item If $\pathLL s'$ is positive-ended then either $\pathLL s' = \pathLL s$ and there is nothing to prove or $\pathLL s'$ is a strict prefix of $\pathLL s$, so assume we are in the second case. $\pathLL s'$ is $\daimon$-free, hence $\overline{\pathLL s'}$ is a negative-ended prefix of $\dual{\pathLL s} \in \visit{B^\perp}$. Using the argument above, it comes $\dual{\pathLL s'} = \overline{\pathLL s'} \daimon \in \visit{B^\perp}$, thus $\pathLL s' \in \visit B$.
  \end{itemize}
\end{proof}

\begin{lemma} \label{nec}
Let $\pathLL s \in \visit{B}$. For every prefix $\pathLL s' \kappa^-$ of $\pathLL s$ and every $\design d \in \beh B$ such that $\pathLL s'$ is a path of $\design d$, $\pathLL s'\kappa^-$ is a prefix of a path of $\design d$.
\end{lemma}

\begin{proof}
  There exist $\design d_0 \in \beh B$ and $\design e_0 \in \beh B^{\perp}$ such that $\pathLL s = \langle \design d_0 \leftarrow \design e_0 \rangle$. Let $\pathLL s' \kappa^-$ be a prefix of $\pathLL s$, and let $\design d \in \beh B$ such that $\pathLL s'$ is a path of $\design d$. Since $\overline{\pathLL s'}$ is a prefix of a path of $\design e_0$, $\pathLL s'$ is a prefix of $\langle \design d \leftarrow \design e_0 \rangle$. We cannot have $\pathLL s' = \langle \design d \leftarrow \design e_0 \rangle$, otherwise $\dual{\pathLL s'} = \overline{\pathLL s'}\daimon$ and $\overline{\pathLL s' \kappa^-}$ would be paths of $\design e_0$ differing on a positive action, which is impossible.
  Thus there exists $\kappa^{\prime-}$ such that $\pathLL s'\kappa^{\prime-}$ is a prefix of $\langle \design d \leftarrow \design e_0 \rangle$, which is a path of $\design d$, and necessarily $\kappa^- = \kappa^{\prime-}$. Finally $\pathLL s'\kappa^-$ is a prefix of a path of $\design d$.
\end{proof}

\subsubsection{An Alternative Definition of Regularity}

Define the \defined{anti-shuffle} ($\antishuffle$) as the dual operation of shuffle, that is:
\begin{itemize}
\item $\pathLL s \antishuffle \pathLL t = \dual{\dual{\pathLL s} \shuffle \dual{\pathLL t}}$ if $\pathLL s$ and $\pathLL t$ are paths of same polarity;
\item $S \antishuffle T = \dual{\dual{S} \shuffle \dual{T}}$ if $S$ and $T$ are sets of paths of same polarity.
\end{itemize}

\begin{definition}
  \begin{itemize}
  \item A \defined{trivial view} is an aj-sequence such that each proper action except the first one is justified by the immediate previous action. In other words, it is a view such that its dual is a view as well.
  \item The \defined{trivial view of} an aj-sequence is defined inductively by:
  \begin{align*}
    \triv{\epsilon} & =\epsilon && \mbox{ empty sequence} \\
    \triv{\pathLL s\daimon} & =\triv{\pathLL s}\daimon && \\
    \triv{\pathLL s\kappa} & =\kappa && \mbox{ if } \kappa \neq \daimon \mbox{ initial} \\
    \triv{\pathLL s\kappa} & =\triv{\pathLL s_0}\kappa && \mbox{ if } \kappa \neq \daimon \mbox{ justified, where } \pathLL s_0 \mbox{ prefix of } \pathLL s \mbox{ ending on } \mathrm{just}(\kappa)
  \end{align*}
  We also write $\trivv{\kappa}{\pathLL s}$ (or even $\triv{\kappa}$) instead of $\triv{\pathLL s' \kappa}$ when $\pathLL s' \kappa$ is a prefix of $\pathLL s$.
  \item \defined{Trivial views of a design} $\design d$ are the trivial views of its paths (or of its views). In particular, $\epsilon$ is a trivial view of negative designs only.
  \item Trivial views of designs in $|\beh B|$ are called \defined{trivial views of} $\beh B$.
  \end{itemize}
\end{definition}

\begin{lemma} \label{triv-view-path}
  \begin{enumerate}
  \item Every view is in the anti-shuffle of trivial views.
  \item Every path is in the shuffle of views.
  \end{enumerate}
\end{lemma}

\begin{proof} \
  \begin{enumerate}
  \item Let $\viewseq v$ be a view, the result is shown by induction on $\viewseq v$:
    \begin{itemize}
    \item If $\viewseq v = \epsilon$ or $\viewseq v = \kappa$, it is itself a trivial view, hence the result.
    \item Suppose $\viewseq v = \viewseq v'\kappa$ with $\viewseq v' \neq \epsilon$ and $\viewseq v' \in \viewseq t_1 \antishuffle \dots \antishuffle \viewseq t_n$ where the $\viewseq t_i$ are trivial views. 
      \begin{itemize}
      \item If $\kappa$ is negative, 
as $\viewseq v$ is a view, the action $\kappa$ is justified by the last action of $\viewseq v'$, say $\kappa^+$. Hence $\kappa^+$ is the last action of some trivial view $\viewseq t_{i_0}$.
Hence $\viewseq v \in \viewseq t_1 \antishuffle \dots \antishuffle \viewseq t_{i_0-1} \antishuffle (\viewseq t_{i_0}\kappa) \antishuffle \viewseq t_{i_0+1} \antishuffle \dots \antishuffle \viewseq t_n$.
      \item If $\kappa$ is positive, either it is initial and $\viewseq v \in \viewseq t_1 \antishuffle \dots \antishuffle \viewseq t_n \antishuffle \kappa$ with $\kappa$ a trivial view, or it is justified by $\kappa^-$ in $\viewseq v'$. In the last case, there exists a unique $i_0$ such that $\kappa^-$ appears in $\viewseq t_{i_0}$, so let $\viewseq t\kappa^-$ be the prefix of $\viewseq t_{i_0}$ ending with $\kappa^-$. We have that $\viewseq v \in \viewseq t_1 \antishuffle \dots \antishuffle \viewseq t_n \antishuffle (\viewseq t\kappa^-\kappa)$ where $\viewseq t\kappa^-\kappa$ is a trivial view.
      \end{itemize}
    \end{itemize}
  \item Similar reasoning as above, but replacing $\antishuffle$ by $\shuffle$, ``trivial view'' by ``view'', ``view'' by ``path'', and exchanging the role of the polarities of actions.
  \end{enumerate}
\end{proof}

\begin{remark}
  Following previous result, note that every view (resp. path) of a design $\design d$ is in the anti-shuffle of trivial views (resp. in the shuffle of views) of $\design d$.
\end{remark}

\begin{proposition} \label{reg2}
  $\beh B$ is regular if and only if the following conditions hold:
  \begin{itemize}
  \item the positive-ended trivial views of $\beh B$ are visitable in $\beh B$,
  \item $\visit{B}$ and $\visit{B^\perp}$ are stable under $\shuffle$ (i.e., $\visit{B}$ is stable under $\shuffle$ and $\antishuffle$).
  \end{itemize}
\end{proposition}

\begin{proof} Let $\beh B$ be a behaviour. \\
  \noindent ($\Rightarrow$) Suppose $\beh B$ is regular, and let $\viewseq t$ be a positive-ended trivial view of $\beh B$. There exists a view $\viewseq v$ of a design $\design d \in |\beh B|$ such that $\viewseq t$ is a subsequence of $\viewseq v$, and such that $\viewseq v$ ends with the same action as $\viewseq t$. Since $\viewseq v$ is a view of $\design d$, $\viewseq v$ is in particular a path of $\design d$, hence by regularity $\viewseq v \in \visit{B}$. There exists $\design e \in \beh B^\perp$ such that $\viewseq v = \interseq{\design d}{\design e}$, and by Lemma~\ref{lem_visit_path_inc} we can even take $\design e \in |\beh B^\perp|$. Since $\dual{\viewseq v}$ is a path of $\design e$, $\view{\dual{\viewseq v}}$ is a view of $\design e$. But notice that $\view{\dual{\viewseq v}} = \view{\dual{\viewseq t}} = \dual{\viewseq t}$ by definition of a view and of a trivial view. We deduce that $\dual{\viewseq t}$ is a view (and in particular a path) of $\design e$, hence $\dual{\viewseq t} \in \visit{B^\perp}$ by regularity. Finally, $\viewseq t \in \visit{B}$. \\
\noindent ($\Leftarrow$) Assume the two conditions of the statement. 
Let $\pathLL s$ be a path of some design of $|\beh B|$. By Lemma~\ref{triv-view-path}, we know that there exist views $\viewseq v_1$, \dots, $\viewseq v_n$ such that $\pathLL s \in \viewseq v_1 \shuffle \dots \shuffle \viewseq v_n$, and for each $\viewseq v_i$ there exist trivial views $\viewseq t_{i,1}$, \dots, $\viewseq t_{i,m_i}$ such that $\viewseq v_i \in \viewseq t_{i,1} \antishuffle \dots \antishuffle \viewseq t_{i,m_i}$.
By hypothesis each $\viewseq t_{i,j}$ is visitable in $\beh B$, hence as $\visit{B}$ is stable by anti-shuffle, $\viewseq v_i \in \visit{B}$. Thus as $\visit{B}$ is stable by shuffle, $\pathLL s \in \visit{B}$.
Similarly the paths of designs of $|\beh B^\perp|$ are visitable in $\beh B^\perp$. Hence the result.
\end{proof}

\subsection{Form of the Visitable Paths} \label{sec-inc-visit}

From internal completeness, we can make explicit the form of the visitable paths for behaviours constructed by logical connectives; such results are necessary for proving the stability of regularity and purity (\textsection~\ref{sub-reg} and  \ref{sub-pur} respectively).

We will use the notations given at the beginning of Subsection~\ref{sub_connectives}, and also the following. Given an action $\kappa$ and a set of sequences $V$, we write $\kappa V$ for $\setst{\kappa\pathLL s}{\pathLL s \in V}$. Let us note $\kappa_\symshpos = x_0|\symshpos\langle x\rangle$, $\kappa_\symshneg = \symshneg_{x_0}(x)$, $\kappa_\symtensor = x_0|\symtensor \langle x, y \rangle$ and $\kappa_{\symplus_i} = x_0|\symplus_i \langle x_i\rangle$ for $i \in \{1, 2\}$.

In this section are proved the following results:
  \begin{itemize}
  \item $\visit{\shpos N} = \kappa_\symshpos \visit{N}^x \cup \{\daimon\}$ and $\visit{\shneg P} = \kappa_\symshneg \visit{P}^x \cup \{\epsilon\}$ (Proposition~\ref{visit-sh}),
  \item $\visit{M \oplus N} = \kappa_{\symplus_1} \visit{M}^{x_1} \cup \kappa_{\symplus_2} \visit{N}^{x_2} \cup \{\daimon\}$ (Proposition~\ref{visit-pl}),
  \item $\visit{M \otimes N} = \kappa_\symtensor (\visit{M}^x \shuffle \visit{N}^y) \cup \{\daimon\}$ if $\beh M$ and $\beh N$ are regular (Proposition~\ref{visit-tensor-reg}),
  \item the general form of the visitable paths of $\beh M \otimes \beh N$, not as simple (Proposition~\ref{visit-tensor}),
  \item finally, the case of $\multimap$ easily deduced from $\otimes$ (Corollaries~\ref{visit-arrow} and \ref{visit-arrow-reg}).
  \end{itemize}

\subsubsection{Shifts}

\begin{lemma} \ \label{lem-inc-sh}
  \begin{enumerate}
    \item \label{lem1} $(\symshneg(x).(\beh N^\perp)^x)^\perp \subseteq \symshpos\langle\beh N\rangle \cup \{\daimon\}$.
    \item \label{lem2} $\symshneg(x).(\beh N^\perp)^x \subseteq \symshpos\langle\beh N\rangle^\perp$.
  \end{enumerate}
\end{lemma}

\begin{proof}
  Let $E = \symshpos\langle\beh N\rangle$, and let $F = \symshneg(x).(\beh N^\perp)^x$. To show the lemma, we must show $F^\perp \subseteq E \cup \{\daimon\}$ and $F \subseteq E^\perp$.
  \begin{enumerate}
  \item Let $\design q \in F^\perp$. If $\design q \neq \daimon$, $\design q$ is necessarily of the form $\symshpos\langle\design n\rangle$ where $\design n$ is a negative atomic design. For every design $\design p \in \beh N^\perp$, we have $\symshneg(x).\design p^x \in F$ and  $\design q[\symshneg(x).\design p^x/x_0] \leadsto \design p[\design n/x_0]$, thus $\normalisation{\design q[\symshneg(x).\design p^x/x_0]} = \normalisation{\design p[\design n/x_0]} = \daimon$ since $\design q \perp \symshneg(x).\design p^x$. We deduce $\design n \in \beh N$, hence $\design q \in E$.
  \item Let $\design m = \symshneg(x).\design p^x \in F$. For every design $\design n \in \beh N$, we have $\symshpos\langle\design n\rangle[\design m/x_0] \leadsto \design p[\design n/x_0]$, thus $\normalisation{\symshpos\langle\design n\rangle[\design m/x_0]} = \normalisation{\design p[\design n/x_0]} = \daimon$ since $\design p \in \beh N^\perp$ and $\design n \in \beh N$. Hence $\design m \in E^\perp$.
  \end{enumerate}
\end{proof}

\begin{lemma} \label{lem-inc-shneg}
  $\shneg \beh P = (\shpos \beh P^\perp)^\perp$.
\end{lemma}

\begin{proof}
  If we take $\beh N = \beh P^\perp$, Lemma~\ref{lem-inc-sh} gives us:  
  \begin{enumerate}
  \item \label{lem3} $(\symshneg(x).\beh P^x)^\perp \subseteq \symshpos\langle\beh P^\perp\rangle \cup \{\daimon\}$ and
  \item \label{lem4} $\symshneg(x).\beh P^x \subseteq \symshpos\langle\beh P^\perp\rangle^\perp$.
  \end{enumerate}
  Let $E = \symshpos\langle\beh P^\perp\rangle$, and let $F = \symshneg(x).\beh P^x$. By definition $\shneg \beh P = F^{\perp\perp}$. From (\ref{lem4}) we deduce $F^{\perp\perp} \subseteq E^\perp$, and from (\ref{lem3}) $E^\perp = (E \cup \{\daimon\})^\perp \subseteq  F^{\perp\perp}$. Hence $\shneg \beh P = F^{\perp\perp} = E^\perp = (\shpos \beh P^\perp)^\perp$.
\end{proof}

\begin{proposition}\  \label{visit-sh}
  \begin{enumerate}
  \item $\visit{\shpos N} = \kappa_\symshpos \visit{N}^x \cup \{\daimon\}$
  \item $\visit{\shneg P} = \kappa_\symshneg \visit{P}^x \cup \{\epsilon\}$
  \end{enumerate}
\end{proposition}

\begin{proof} \ 
  \begin{enumerate}
  \item $(\subseteq)$ Let $\design q \in \shpos \beh N$ and $\design m \in (\shpos \beh N)^\perp$, let us show that $\langle \design q \leftarrow \design m\rangle \in \kappa_\symshpos\visit{N}^x \cup \{\daimon\}$. By Lemma~\ref{lem-inc-shneg}, $\design m \in \shneg \beh N^\perp$. If $\design q = \daimon$ then $\langle \design q \leftarrow \design m \rangle = \daimon$. Otherwise, by Theorem~\ref{thm_intcomp_all}, $\design q = \symshpos\langle\design n\rangle$ with $\design n \in \beh N$. We have $\langle \design q \leftarrow \design m \rangle = \langle \design q \leftarrow |\design m| \rangle$ by Lemma~\ref{lem_visit_path_inc} , where $ |\design m| \in \symshneg(x).|(\beh N^\perp)^x|$ by Theorem~\ref{thm_intcomp_all}, hence $|\design m|$ is of the form $|\design m| = \symshneg(x).\design p^x$ with $\design p \in \beh N^\perp$. By definition $\langle \design q \leftarrow |\design m| \rangle = \kappa_\symshpos \langle \design n^x \leftarrow \design p^x \rangle$, where $\langle \design n^x \leftarrow \design p^x \rangle \in \visit{N}^x$.
    
    $(\supseteq)$ Indeed $\daimon \in \visit{\shpos N}$. Now let $\pathLL s \in \kappa_\symshpos \visit{N}^x$. There exist $\design n \in \beh N$ and $\design p \in \beh N^\perp$ such that $\pathLL s = \kappa_\symshpos \langle \design n^x \leftarrow \design p^x \rangle$. Note that $\symshpos\langle\design n\rangle \in \symshpos\langle\beh N\rangle$ and $\symshneg(x).\design p^x \in \symshneg(x).(\beh N^\perp)^x$. By Lemma~\ref{lem-inc-shneg}, $\shneg \beh N^\perp = (\shpos \beh N)^\perp$, hence $\symshpos\langle\design n\rangle \perp \symshneg(x).\design p^x$. Moreover $\interseq{\symshpos\langle\design n\rangle}{\symshneg(x).\design p^x} = \kappa_\symshpos \langle \design n^x \leftarrow \design p^x \rangle = \pathLL s$, therefore $\pathLL s \in \visit{\shpos N}$.

  \item By Lemma~\ref{lem-inc-shneg} and previous item, and remarking that $\visit B = \dual{\visit{B^\perp}}$ for every behaviour $\beh B$, we have: 
    $\visit{\shneg P} = \dual{\visit{(\shneg P)^\perp}} = \dual{\visit{\shpos P^\perp}} = \dual{(\kappa_\symshpos \visit{ P^\perp}^x \cup \{\daimon\})} = \kappa_\symshneg \dual{\visit{P^\perp}^x} \cup \{\epsilon\} = \kappa_\symshneg \visit{P}^x \cup \{\epsilon\}$. 
  \end{enumerate}
\end{proof}

\subsubsection{Plus}

\begin{proposition}\  \label{visit-pl}
  $\visit{M \oplus N} = \kappa_{\symplus_1} \visit{M}^{x_1} \cup \kappa_{\symplus_2} \visit{N}^{x_2} \cup \{\daimon\}$
\end{proposition}

\begin{proof}
  Remark that $\beh M \oplus \beh N = (\symplus_1 \langle \beh M \rangle \cup \{\daimon\}) \cup (\symplus_2 \langle \beh N \rangle \cup \{\daimon\})$ is the union of behaviours $\oplus_1 \beh M$ and $\oplus_2 \beh N$, which correspond respectively to $\shpos \beh M$ and $\shpos \beh N$ with a different name for the first action. Moreover, $(\beh M \oplus \beh N)^\perp = \setst{\design n}{\proj{\design n}{\pi_1} \in \pi_1(x).(\beh M^\perp)^x \mbox{ and } \proj{\design n}{\pi_2} \in \pi_2(x).(\beh N^\perp)^x} = (\with_1 \beh M^\perp) \cap (\with_2 \beh N^\perp)$, where the behaviours $\with_1 \beh M^\perp$ and $\with_2 \beh N^\perp$ correspond to $\shneg \beh M^\perp$ and $\shneg \beh N^\perp$ with different names; note also that for every $\design d \in |\with_1\beh M^\perp|$ (resp. $|\with_2\beh N^\perp|$) there exists $\design d' \in (\beh M \oplus \beh N)^\perp$ such that $\design d \sqsubseteq \design d'$, in other words such that $\design d = |\design d'|_{\with_1\beh M^\perp}$ (resp. $|\design d'|_{\with_2\beh N^\perp}$). Therefore the proof can be conducted similarly to the one of Proposition~\ref{visit-sh}(1).
\end{proof}

\subsubsection{Tensor and Linear Map}

The following proposition is a joint work with Fouquer\'e and Quatrini; in \cite{FQ2}, they prove a similar result in the framework of original Ludics.

\begin{proposition} \label{visit-tensor}
$\pathLL s \in \visit{M \otimes N}$ if and only if the two conditions below are satisfied:
\begin{enumerate}
\item $\pathLL s \in \kappa_\symtensor(\visit{M}^x \shuffle \visit{N}^y) \cup \{\daimon\}$,
\item for all $\pathLL t \in \visit{M}^x \shuffle \visit{N}^y$, for all $\kappa^-$ such that $\overline{\kappa_\symtensor \pathLL t \kappa^-}$ is a path of $\completed{\dual{\pathLL s}}$, $\pathLL t \kappa^- \daimon \in \visit{M}^x \shuffle \visit{N}^y$.
\end{enumerate}
\end{proposition}

The proof of this proposition uses some material on multi-designs introduced in Section~\ref{multi}. Note also that for all negative designs $\design m$ and $\design n$, we will write $\design m \otimes \design n$ instead of $x_0|\symtensor\langle\design m, \design n\rangle$.

\begin{proof} \
  $(\Rightarrow)$ Let $\pathLL s \in \visit{M \otimes N}$. If $\pathLL s = \daimon$ then both conditions are trivial, so suppose $\pathLL s \neq \daimon$. By internal completeness (Theorem~\ref{thm_intcomp_all}), there exist $\design m \in \beh M$, $\design n \in \beh N$ and $\design n_0 \in (\beh M \otimes \beh N)^{\perp}$ such that $\pathLL s = \langle \design m \otimes \design n \leftarrow \design n_0\rangle$. Thus $\design n_0$ must be of the form $\design n_0 = \negdes{a}{\vect{z^a}}{\design p_a}$ with $\design p_\sympar \neq \Omega$ (remember that $\symtensor = \overline \sympar$), and we have $\pathLL s = \kappa_\symtensor \pathLL s'$ where $\pathLL s' = \langle \{\design m^x, \design n^y\} \leftarrow \{\design p_\sympar, \design n_0/x_0\}\rangle = \langle \{\design m^x, \design n^y\} \leftarrow \design p_\sympar\rangle$. Let us prove both properties:
  \begin{enumerate}
  \item By Proposition~\ref{asso-path}, $\proj{\pathLL s'}{\design m^x} = \langle \design m^x \leftarrow \normalisation{\design p_\sympar[\design n/y]} \rangle$, where $\design m^x \in \beh M^x$. Moreover, $\normalisation{\design p_\sympar[\design n/y]} \in \beh M^{x\perp}$, indeed: for any $\design m' \in \beh M$, we have $\normalisation{\normalisation{\design p_\sympar[\design n/y]}[\design m'/x]} = \normalisation{\design p_\sympar[\design n/y, \design m'/x]} = \normalisation{(\design m' \otimes \design n)[\design n_0/x_0]} = \daimon$ using associativity and one reduction step backwards. Thus $\proj{\pathLL s'}{\design m^x} \in \visit{M}^x$. Likewise, $\proj{\pathLL s'}{\design n^y} = \langle \design n^y \leftarrow \normalisation{\design p_\sympar[\design m/x]} \rangle$, so $\proj{\pathLL s'}{\design n^y} \in \visit{N}^y$. Therefore $\pathLL s' \in (\visit{M}^x \shuffle \visit{N}^y)$.
    
  \item Now let $\pathLL t_1 \in \visit{M}^x, \pathLL t_2 \in \visit{N}^y$. Suppose $\pathLL t \in (\pathLL t_1 \shuffle \pathLL t_2)$ and $\kappa^-$ is a negative action such that $\overline {\kappa_\symtensor \pathLL t\kappa^-}$ is a path of $\completed{\dual{\pathLL s}}$. Without loss of generality, suppose moreover that the action $\kappa^-$ comes from $\design m^x$, and let us show that $\pathLL t_1 \kappa^- \daimon \in \visit{M}^x$.
    
    Let $\pathLL t' = \langle \{\completed{\pathLL t_1}/x, \completed{\pathLL t_2}/y\} \leftarrow \completed{\dual{\pathLL s'}} \rangle$. We will show that $\pathLL t_1 \kappa^-$ is a prefix of $\proj{\pathLL t'}{\completed{\pathLL t_1}}$ and that $\proj{\pathLL t'}{\completed{\pathLL t_1}} \in \visit M^x$, leading to the conclusion by Lemma~\ref{daimon_visit}. Note the following facts:
    \begin{enumerate}
    \item $\completed{\dual{\pathLL s}} = \sympar(x, y).\completed{\dual{\pathLL s'}} + \sum_{a \neq \sympar}a(\vect{z^a}).\daimon$, and thus $\completed{\dual{\pathLL s'}} \neq \daimon$ (otherwise a path of the form $\overline {\kappa_\symtensor \pathLL t\kappa^-}$ cannot be path of $\completed{\dual{\pathLL s}}$).
    \item $\pathLL t$ is a path of the multi-design $\{\completed{\pathLL t_1}/x, \completed{\pathLL t_2}/y\}$, and $\overline{\pathLL t}$ is a prefix of a path of $\completed{\dual{\pathLL s'}}$ since $\overline{\kappa_\symtensor \pathLL t\kappa^-}$ is a path of $\completed{\dual{\pathLL s}}$, thus $\pathLL t$ is a prefix of $\pathLL t'$ by Proposition~\ref{prefix-norm}.
    \item Since $\pathLL t$ is a $\daimon$-free positive-ended prefix of $\pathLL t'$, we have that $\overline{\kappa_\symtensor \pathLL t}$ is a strict prefix of $\dual{\kappa_\symtensor \pathLL t'}$. Thus there exists a positive action $\kappa_0^+$ such that $\overline{\kappa_\symtensor \pathLL t}\kappa_0^+$ is a prefix of $\dual{\kappa_\symtensor \pathLL t'}$. The paths $\overline{\kappa_\symtensor \pathLL t\kappa^-}$ and $\overline{\kappa_\symtensor \pathLL t}\kappa^+_0$ are both paths of $\completed{\dual{\pathLL s}}$, hence necessarily $\kappa^+_0 = \overline{\kappa^-}$. We deduce that $\pathLL t \kappa^-$ is a prefix of $\pathLL t'$.
    \item The sequence $\proj{\pathLL t'}{\completed{\pathLL t_1}}$ therefore starts with $\proj{(\pathLL t \kappa^-)}{\completed{\pathLL t_1}}$.
    \item We have $\proj{(\pathLL t \kappa^-)}{\completed{\pathLL t_1}} = (\proj{\pathLL t}{\completed{\pathLL t_1}})\kappa^-$ because, since $\kappa^-$ comes from $\design m^x$, it is hereditarily justified by an initial negative action of address $x$, and thus $\kappa^-$ appears in design $\completed{\pathLL t_1}$. We deduce $\proj{(\pathLL t \kappa^-)}{\completed{\pathLL t_1}} = (\proj{\pathLL t}{\completed{\pathLL t_1}})\kappa^- = \pathLL t_1\kappa^-$.
    \item Moreover, by Proposition~\ref{asso-path} $\proj{\pathLL t'}{\completed{\pathLL t_1}} = \langle \completed{\pathLL t_1} \leftarrow \normalisation{\completed{\dual{\pathLL s'}}[\completed{\pathLL t_2}/y]} \rangle$.
    \end{enumerate}
    Hence (by d, e, f) the sequence $\pathLL t_1\kappa^-$ is a prefix of $\proj{\pathLL t'}{\completed{\pathLL t_1}} = \langle \completed{\pathLL t_1} \leftarrow \normalisation{\completed{\dual{\pathLL s'}}[\completed{\pathLL t_2}/y]} \rangle$. Since $\completed{\pathLL t_1} \in \beh M^x$ (by Proposition~\ref{cor-mono}) and $\normalisation{\completed{\dual{\pathLL s'}}[\completed{\pathLL t_2}/y]} \in \beh M^{x\perp}$ (by associativity, similar reasoning as item $1$), we deduce $\proj{\pathLL t'}{\completed{\pathLL t_1}} \in \visit{M}^x$. Finally $\pathLL t_1\kappa^-\daimon \in \visit{M}^x$ by Lemma~\ref{daimon_visit}.
  \end{enumerate}
  
  \noindent $(\Leftarrow)$ Let $\pathLL s \in \kappa_\symtensor (\visit{M}^x \shuffle \visit{N}^y) \cup \{\daimon\}$ such that the second constraint is also satisfied. If $\pathLL s = \daimon$ then $\pathLL s \in \visit{M \otimes N}$ is immediate, so suppose $\pathLL s = \kappa_\symtensor \pathLL s'$ where $\pathLL s' \in (\visit{M}^x \shuffle \visit{N}^y)$. Consider the design $\completed{\dual{\pathLL s}}$, and note that $\completed{\dual{\pathLL s}} = \sympar(x, y).\completed{\dual{\pathLL s'}} + \sum_{a \neq \sympar}a(\vect{z^a}).\daimon$. We will show by contradiction that $\completed{\dual{\pathLL s}} \in (\beh M \otimes \beh N)^\perp$, leading to the conclusion.

Let $\design m \in \beh M$ and $\design n \in \beh N$ such that $\design m \otimes \design n \not\perp \completed{\dual{\pathLL s}}$. By Proposition~\ref{path-perp} and given the form of design $\completed{\dual{\pathLL s}}$, the interaction with $\design m \otimes \design n$ is finite and the cause of divergence is necessarily the existence of a path $\pathLL t$ and an action $\kappa^-$ such that:
    \begin{enumerate}
    \item $\pathLL t$ is a path of $\design m \otimes \design n$,
    \item $\overline{\pathLL t\kappa^-}$ is a path of $\completed{\dual{\pathLL s}}$
    \item $\pathLL t\kappa^-$ is not a path of $\design m \otimes \design n$.
    \end{enumerate}
    Hence $\pathLL t$ is of the form $\pathLL t = \kappa_\symtensor \pathLL t'$. Choose $\design m$ and $\design n$ such that $\pathLL t$ is of minimal length with respect to all such pairs of designs non orthogonal to $\completed{\dual{\pathLL s}}$. Let $\pathLL t_1 = \proj{\pathLL t'}{\design m^x}$ and $\pathLL t_2 = \proj{\pathLL t'}{\design n^y}$, we have $\pathLL t \in \kappa_\symtensor (\pathLL t_1 \shuffle \pathLL t_2)$. Consider the design $\completed{\dual{\pathLL t}}$, and note that $\completed{\dual{\pathLL t}} = \sympar(x, y).\completed{\dual{\pathLL t'}} + \sum_{a \neq \sympar}a(\vect{z^a}).\daimon$. We prove the following:
    \begin{itemize}
    \item \underline{$\completed{\dual{\pathLL t}} \in (\beh M \otimes \beh N)^\perp$}: By contradiction. Let $\design m' \in \beh M$ and $\design n' \in \beh N$ such that $\design m' \otimes \design n' \not\perp \completed{\dual{\pathLL t}}$. Again using Proposition~\ref{path-perp}, divergence occurs necessarily because there exists a path $\pathLL v$ and a negative action $\kappa^{\prime-}$ such that:
      \begin{enumerate}
      \item $\pathLL v$ is a path of $\design m' \otimes \design n'$,
      \item $\overline{\pathLL v\kappa^{\prime-}}$ is a path of $\completed{\dual{\pathLL t}}$,
      \item $\pathLL v\kappa^{\prime-}$ is not a path of $\design m' \otimes \design n'$.
      \end{enumerate}
      Since the views of $\overline{\pathLL v\kappa^{\prime-}}$ are views of $\overline{\pathLL t}$, $\overline{\pathLL v\kappa^{\prime-}}$ is a path of $\completed{\dual{\pathLL s}}$. Thus $\design m' \otimes \design n' \not\perp \completed{\dual{\pathLL s}}$. Moreover $\pathLL v$ is strictly shorter than $\pathLL t$, indeed: $\pathLL v$ and $\pathLL t$ are $\daimon$-free, and since $\overline{\pathLL v \kappa^{\prime-}}$ is a path of $\completed{\dual{\pathLL t}}$ any action of $\pathLL v \kappa^{\prime-}$ is an action of $\pathLL t$. This contradicts the fact that $\pathLL t$ is of minimum length. We deduce $\completed{\dual{\pathLL t}} \in (\beh M \otimes \beh N)^\perp$.

    \item \underline{$\pathLL t \in \kappa_\symtensor (\visit{M}^x \shuffle \visit{N}^y)$}: We show $\pathLL t_1 \in \visit{M}^x$, the proof of $\pathLL t_2 \in \visit{N}^y$ being similar. Since $\pathLL t$ is a path of $\design m \otimes \design n$ and $\dual{\pathLL t}$ a path of $\completed{\dual{\pathLL t}}$, we have $\pathLL t = \interseq{\design m \otimes \design n}{\completed{\dual{\pathLL t}}} = \kappa_\symtensor \interseq{\{\design m^x, \design n^y\}}{\completed{\dual{\pathLL t'}}}$, hence $\pathLL t' = \interseq{\{\design m^x, \design n^y\}}{\completed{\dual{\pathLL t'}}}$. Thus by Proposition~\ref{asso-path} $\pathLL t_1 = \proj{\pathLL t'}{\design m^x} = \interseq{\design m^x}{\normalisation{\completed{\dual{\pathLL t'}}[\design n/y]}}$. Moreover $\normalisation{\completed{\dual{\pathLL t'}}[\design n/y]} \in \beh M^{x\perp}$: for any design $\design m' \in \beh M$ we have $\normalisation{\completed{\dual{\pathLL t'}}[\design n/y]} \perp \design m'^x$ because of the equality $\normalisation{\normalisation{\completed{\dual{\pathLL t'}}[\design n/y]}[\design m'/x]} = \normalisation{\completed{\dual{\pathLL t'}}[\design n/y, \design m'/x]} = \normalisation{(\design m' \otimes \design n)[\completed{\dual{\pathLL t}}/x_0]} = \daimon$, using associativity, one reduction step backwards, and the fact that $\completed{\dual{\pathLL t}} \in (\beh M \otimes \beh N)^\perp$. It follows that $\pathLL t_1 \in \visit{M}^x$.

    \item \underline{$\pathLL t\kappa^-$ is a path of $\design m \otimes \design n$}: Remember that $\overline{\pathLL t\kappa^-}$ is a path of $\completed{\dual{\pathLL s}}$, and we have just seen that $\pathLL t \in \kappa_\symtensor (\visit{M}^x \shuffle \visit{N}^y)$. Using the second constraint of the proposition, we should have $\pathLL t_1 \kappa^-\daimon \in \visit{M}^x$ or $\pathLL t_2 \kappa^-\daimon \in \visit{N}^y$. Without loss of generality suppose $\pathLL t_1 \kappa^-\daimon \in \visit{M}^x$. Since $\design m^x \in \beh M^x$ and $\pathLL t_1$ is a path of $\design m^x$, we should also have that $\pathLL t_1\kappa^-$ is a prefix of a path of $\design m^x$ by Lemma~\ref{nec}, hence $\view{\pathLL t' \kappa^-} = \view{\pathLL t_1 \kappa^-}$ is a view of $\design m^x$. But in this case, knowing that $\pathLL t$ is a path of $\design m \otimes \design n$ and that $\view{\pathLL t \kappa^-} = \kappa_\symtensor \view{\pathLL t' \kappa^-}$ is a view of $\design m \otimes \design n$, we deduce that $\pathLL t\kappa^-$ is a path of $\design m \otimes \design n$.
    \end{itemize}
    Last point contradicts the cause of divergence between $\design m \otimes \design n$ and $\completed{\dual{\pathLL s}}$. Hence $\completed{\dual{\pathLL s}} \in (\beh M \otimes \beh N)^\perp$. Moreover, $\dual{\pathLL s}$ is a path of $\completed{\dual{\pathLL s}} $, and since $\pathLL s \in \kappa_\symtensor (\visit{M}^x \shuffle \visit{N}^y)$ there exist $\design m_0 \in \beh M$ and $\design n_0 \in \beh N$ such that $\pathLL s$ is a path of $\design m_0 \otimes \design n_0$ (and $\design m_0 \otimes \design n_0 \in \beh M \otimes \beh N$). We deduce $\pathLL s = \interseq{\design m_0 \otimes \design n_0}{\completed{\dual{\pathLL s}}}$, hence $\pathLL s \in \visit{M \otimes N}$.
\end{proof}

\begin{corollary} \label{visit-arrow}
  $\pathLL s \in \visit{N \multimap P}$ if and only if the two conditions below are satisfied:
  \begin{enumerate}
  \item $\dual{\pathLL s} \in \kappa_\symtensor(\visit{N}^x \shuffle \dual{\visit{P}^y}) \cup \{\daimon\}$
  \item for all $\pathLL t \in \visit{N}^x \shuffle \dual{\visit{P}^y}$, for all $\kappa^-$ such that $\overline{\kappa_\symtensor \pathLL t \kappa^-}$ is a path of $\completed{\pathLL s}$, $\pathLL t \kappa^- \daimon \in \visit{N}^x \shuffle \dual{\visit{P}^y}$.
  \end{enumerate}
\end{corollary}

\subsubsection{Tensor and Linear Map, Regular Case}

\begin{proposition} \label{visit-tensor-reg}
  If $\beh M$ and $\beh N$ regular then $\visit{M \otimes N} = \kappa_\symtensor (\visit{M}^x \shuffle \visit{N}^y) \cup \{\daimon\}$.
\end{proposition}

\begin{proof}
Suppose $\beh M$ and $\beh N$ regular. Following Proposition~\ref{visit-tensor}, it suffices to show that any path $\pathLL s \in \kappa_\symtensor (\visit{M}^x \shuffle \visit{N}^y) \cup \{\daimon\}$ satisfies the following condition: for all $\pathLL t \in \visit{M}^x \shuffle \visit{N}^y$, for all negative action $\kappa^-$ such that $\overline{\kappa_\symtensor \pathLL t \kappa^-}$ is a path of $\completed{\dual{\pathLL s}}$, $\pathLL t \kappa^- \daimon \in \visit{M}^x \shuffle \visit{N}^y$.

If $\pathLL s = \daimon$, there is nothing to prove, so suppose $\pathLL s = \kappa_\symtensor \pathLL s'$ where $\pathLL s' \in \visit{M}^x \shuffle \visit{N}^y$. Let $\pathLL t \in \visit{M}^x \shuffle \visit{N}^y$ and $\kappa^-$ be such that $\overline{\kappa_\symtensor \pathLL t \kappa^-}$ is a path of $\completed{\dual{\pathLL s}}$, that is $\overline{\pathLL t \kappa^-}$ is a path of $\completed{\dual{\pathLL s'}}$. Let $\pathLL s_1, \pathLL t_1 \in \visit{M}^x$ and $\pathLL s_2, \pathLL t_2 \in \visit{N}^y$ such that $\pathLL s' \in \pathLL s_1 \shuffle \pathLL s_2$ and $\pathLL t \in \pathLL t_1 \shuffle \pathLL t_2$. Without loss of generality, suppose $\kappa^-$ is an action in $\pathLL s_1$, thus we must show $\pathLL t_1 \kappa^- \daimon \in \visit{M}^x$. Notice that $\triv{\pathLL t_1\kappa^-} = \triv{\pathLL t\kappa^-} = \trivv{\kappa^-}{\pathLL s'} = \trivv{\kappa^-}{\pathLL s_1}$ (the second equality follows from the fact that $\overline{\pathLL t \kappa^-}$ is a path of $\completed{\,\dual{\pathLL s'}\,}$). Since $\pathLL s_1 \in \visit{M}^x$, the sequence $\trivv{\kappa^-}{\pathLL s_1} = \triv{\pathLL t_1\kappa^-}$ is a trivial view of $\beh M^x$. Let $\pathLL s_1'\kappa^-$ be the prefix of $\pathLL s_1$ ending with $\kappa^-$. By Lemma~\ref{daimon_visit} $\pathLL s_1'\kappa^-\daimon \in \visit{M}^x$, so $\triv{\pathLL t_1\kappa^- \daimon} = \triv{\pathLL s'_1\kappa^-\daimon}$ is also a trivial view of $\beh M^x$; by regularity of $\beh M$, we deduce $\triv{\pathLL t_1\kappa^- \daimon} \in \visit{M}^x$. We have $\pathLL t_1 \kappa^- \daimon \in \pathLL t_1 \shuffle \triv{\pathLL t_1\kappa^-\daimon}$, where both $\pathLL t_1$ and $\triv{\pathLL t_1\kappa^- \daimon}$ are in $\visit{M}^x$, hence $\pathLL t_1 \kappa^- \daimon \in \visit{M}^x$ by regularity of $\beh M$.
\end{proof}

\begin{corollary} \label{visit-arrow-reg}
  If $\beh N$ and $\beh P$ are regular then $\visit{N \multimap P} = \dual{\kappa_\symtensor (\visit{N} \shuffle \dual{\visit{P}})} \cup \{\epsilon\}$.
\end{corollary}

\subsection{Proof of Proposition~\ref{prop_reg_stable}: Regularity and Connectives} \label{sub-reg}

\begin{proposition}\ \label{reg-sh}
  \begin{enumerate}
    \item If $\beh N$ is regular then $\shpos \beh N$ is regular.
    \item If $\beh P$ is regular then $\shneg \beh P$ is regular.
  \end{enumerate}
\end{proposition}

\begin{proof}\ 
  \begin{enumerate}
  \item Following Proposition~\ref{reg2}:
    \begin{itemize}
    \item By internal completeness, the trivial views of $\shpos \beh N$ are of the form $\kappa_\symshpos\viewseq t$ where $\viewseq t$ is a trivial view of $\beh N$. Since $\beh N$ is regular $\viewseq t \in \visit{N}$. Hence by Proposition~\ref{visit-sh}, $\kappa_\symshpos\viewseq t \in \visit{\shpos N}$.

    \item Since $\visit{N}$ is stable by shuffle, so is $\visit{\shpos N} = \kappa_\symshpos \visit{N}^x$ where $\kappa_\symshpos$ is a positive action.
    \item For all paths $\kappa_\symshneg \pathLL s$, $\kappa_\symshneg \pathLL t \in \visit{(\shpos N)^\perp} = \kappa_\symshneg \visit{N^\perp}^x$ such that $\kappa_\symshneg \pathLL s \shuffle \kappa_\symshneg \pathLL t$ is defined, $\pathLL s$ and $\pathLL t$ start necessarily by the same positive action and $\pathLL s \shuffle \pathLL t \subseteq \visit{N^\perp}^x$ because $\visit{N^\perp}$ (thus also $\visit{N^\perp}^x$) is stable by $\shuffle$, hence $\kappa_\symshneg \pathLL s \shuffle \kappa_\symshneg \pathLL t = \kappa_\symshneg (\pathLL s \shuffle \pathLL t) \subseteq \visit{(\shpos N)^\perp}$.
    \end{itemize}
    \item If $\beh P$ is regular then $\beh P^\perp$ is too. Then by previous point $\shpos \beh P^\perp$ is regular, therefore so is $(\shpos \beh P^\perp)^\perp$. By Lemma~\ref{lem-inc-shneg}, this means that $\shneg \beh P$ is regular. 
  \end{enumerate}
\end{proof}

\begin{proposition}\ \label{reg-plus}
If $\beh M$ and $\beh N$ are regular then $\beh M \oplus \beh N$ is regular.
\end{proposition}

\begin{proof}
  Similar to Proposition~\ref{reg-sh} (1), by the same remark as in proof of Proposition~\ref{visit-pl}.
\end{proof}

In order to show that $\otimes$ preserves regularity, consider first the following definitions and lemma. We call \defined{quasi-path} a positive-ended P-visible aj-sequence. The \defined{shuffle} $\pathLL s \shuffle \pathLL t$ of two negative quasi-paths $\pathLL s$ and $\pathLL t$ is the set of paths $\pathLL u$ formed with actions from $\pathLL s$ and $\pathLL t$ such that $\proj{\pathLL u}{\pathLL s} = \pathLL s$ and $\proj{\pathLL u}{\pathLL t} = \pathLL t$.

\begin{lemma}\label{lem:dual_inverse_shuffle}
  Let $\pathLL s$ and $\pathLL t$ be negative quasi-paths. If $\pathLL s \shuffle \pathLL t \neq \emptyset$ then $\pathLL s$ and $\pathLL t$ are paths.
\end{lemma}

\begin{proof}
  We prove the result by contradiction. Let us suppose that there exists a triple $(\pathLL s, \pathLL t, \pathLL u)$ such that $\pathLL s$ and $\pathLL t$ are two negative quasi-paths, $\pathLL u \in \pathLL s \shuffle \pathLL t$ is a path, and at least one of $\pathLL s$ or $\pathLL t$ does not satisfy O-visibility, say $\pathLL s$: there exists a negative action $\kappa^-$ and a prefix $\pathLL s_0\kappa^-$ of $\pathLL s$ such that the action $\kappa^-$ is justified in $\pathLL s_0$ but $\mathrm{just}(\kappa^-)$ does not appear in $\antiview{\pathLL s_0}$.
  
  We choose the triple $(\pathLL s, \pathLL t, \pathLL u)$ such that the length of $\pathLL u$ is minimal with respect to all such triples. Without loss of generality, we can assume that $\pathLL u$ and $\pathLL s$ are of the form $\pathLL u = \pathLL u_0 \kappa^- \daimon$ and $\pathLL s = \pathLL s_0 \kappa^- \daimon$ respectively. Indeed, if this is not true, $\pathLL u$ has a strict prefix of the form $\pathLL u_0 \kappa^-$; in this case we can replace $(\pathLL s, \pathLL t, \pathLL u)$ by the triple $(\pathLL s_0\kappa^-\daimon, \proj{\pathLL u_0}{\pathLL t}, \pathLL u_0 \kappa^-\daimon)$ which satisfies all the constraints, and where the length of $ \pathLL u_0 \kappa^-\daimon$ is less or equal to the length of $\pathLL u$.
  
  Let $\kappa^+ = \mathrm{just}(\kappa^-)$. $\pathLL u$ is necessarily of the form $\pathLL u = \pathLL u_1 \alpha^- \pathLL u_2 \alpha^+ \kappa^- \daimon$ where $\alpha^-$ justifies $\alpha^+$ and $\kappa^+$ appears in $\pathLL u_1$, indeed:
  \begin{itemize}
  \item $\kappa^+$ does not appear immediately before $\kappa^-$ in $\pathLL u$, otherwise it would also be the case in $\pathLL s$, contradicting the fact that $\kappa^-$ is not O-visible in $\pathLL s$.
  \item The action $\alpha^+$ which is immediately before $\kappa^-$ in $\pathLL u$ is justified by an action $\alpha^-$, and $\kappa^+$ appears before $\alpha^-$ in $\pathLL u$, otherwise $\kappa^+$ would not appear in $\antiview{\pathLL u_0}$ and that would contradict O-visibility of $\pathLL u$.
  \end{itemize}
  Let us show by contradiction something that will be useful for the rest of this proof: in the path $\pathLL u$, all the actions of $\pathLL u_2$ (which cannot be initial) are justified in $\alpha^-\pathLL u_2$. If it is not the case, let $\pathLL u_1\alpha^-\pathLL u'_2\beta$ be longest prefix of $\pathLL u$ such that $\beta$ is an action of $\pathLL u_2$ justified in $\pathLL u_1$, and let $\beta'$ be the following action (necessarily in $\pathLL u_2\alpha^+$), thus $\beta'$ is justified in $\alpha^-\pathLL u_2$. If $\beta'$ is positive (resp. negative) then $\beta$ is negative (resp. positive), thus $\view{\pathLL u_1\alpha^-\pathLL u'_2\beta} = \view{\pathLL u'_1}$ (resp. $\antiview{\pathLL u_1\alpha^-\pathLL u'_2\beta} = \antiview{\pathLL u'_1}$) where $\pathLL u'_1$ is the prefix of $\pathLL u_1$ ending on $\mathrm{just}(\beta)$. But then $\view{\pathLL u_1\alpha^-\pathLL u'_2\beta}$ (resp. $\antiview{\pathLL u_1\alpha^-\pathLL u'_2\beta}$) does not contain $\mathrm{just}(\beta')$: this contradicts the fact that $\pathLL u$ is a path, since P-visibility (resp. O-visibility) is not satisfied.
  
  Now define $\pathLL u' = \pathLL u_1\kappa^-\daimon$, $\pathLL s' = \proj{\pathLL u'}{\pathLL s}$ and $\pathLL t' = \proj{\pathLL u'}{\pathLL t}$, and remark that:
  \begin{itemize}
  \item \underline{$\pathLL u'$ is a path}, indeed, O-visibility for $\kappa^-$ is still satisfied since $\antiview{\pathLL u_1 \alpha^- \pathLL u_2 \alpha^+ \kappa^-} = \antiview{\pathLL u_1} \alpha^- \alpha^+ \kappa^-$ and $\antiview{\pathLL u_1\kappa^-} = \antiview{\pathLL u_1}\kappa^-$ both contain $\kappa^+$ in $\antiview{\pathLL u_1}$.
  \item \underline{$\pathLL s'$ and $\pathLL t'$ are quasi-paths}, since $\pathLL s'$ is of the form $\pathLL s' = \pathLL s_1\kappa^-\daimon$ where $\pathLL s_1 = \proj{\pathLL u_1}{\pathLL s}$ is a prefix of $\pathLL s$ containing $\kappa^+ = \mathrm{just}(\kappa^-)$, and $\pathLL t' = \proj{\pathLL u'}{\pathLL t} = \proj{\pathLL u_1}{\pathLL t}$ is a prefix of $\pathLL t$.
  \item \underline{$\pathLL u' \in \pathLL s' \shuffle \pathLL t'$}.
  \item \underline{$\pathLL s'$ is not a path}: Note that $\pathLL s$ is of the form $\pathLL s_1 \pathLL s_2 \kappa^- \daimon$ where $\pathLL s_1 = \proj{\pathLL u_1}{\pathLL s}$ and $\pathLL s_2 = \proj{\alpha^- \pathLL u_2 \alpha^+}{\pathLL s}$. By hypothesis, $\pathLL s$ is not a path because $\kappa^+$ does not appear in $\antiview{\pathLL s_1\pathLL s_2}$. But $\antiview{\pathLL s_1\pathLL s_2}$ is of the form $\antiview{\pathLL s_1}\pathLL s_2'$, since all the actions of $\pathLL s_2$ are hereditarily justified by the first (necessarily negative) action of $\pathLL s_2$, indeed: we have proved that, in $\pathLL u$, all the actions of $\pathLL u_2$ (in particular those of $\pathLL s_2$) were justified in $\alpha^-\pathLL u_2$. Thus $\kappa^+$ does not appear in $\antiview{\pathLL s_1}$, which means that O-visibility is not satisfied for $\kappa^-$ in $\pathLL s' = \pathLL s_1\kappa^-\daimon$.
  \end{itemize}
  Hence the triple $(\pathLL s', \pathLL t', \pathLL u')$ satisfies all the conditions. This contradicts the minimality of $\pathLL u$.
\end{proof}

\begin{proposition} \label{reg-tensor}
  If $\beh M$ and $\beh N$ are regular, then $\beh M \otimes \beh N$ is regular.
\end{proposition}

\begin{proof}
  Following Proposition~\ref{reg2}, we will prove that the positive-ended trivial views of $\beh M \otimes \beh N$ are visitable in $\beh M \otimes \beh N$, and that $\visit{M \otimes N}$ and $\visit{(M \otimes N)^\perp}$ are stable by shuffle.

Every trivial view of $\beh M \otimes \beh N$ is of the form $\kappa_\symtensor\viewseq t$. It follows from internal completeness (incarnated form) that $\kappa_\symtensor\viewseq t$ is a trivial view of $\beh M \otimes \beh N$ iff $\viewseq t$ is a trivial view either of $\beh M^x$ or of $\beh N^y$. As $\beh M$ (resp. $\beh N$) is regular, positive-ended trivial views of $\beh M^x$ (resp. $\beh N^y$) are in $\visit{M}^x$ (resp. $\visit{N}^y$). Thus by Proposition~\ref{visit-tensor-reg}, positive-ended trivial views of $\beh M \otimes \beh N$ are in $\visit{M \otimes N}$.
  
From Proposition~\ref{visit-tensor-reg}, and from the fact that $\shuffle$ is associative and commutative, we also have that $\visit{M \otimes N}$ is stable by shuffle.
  
Let us prove that $\visit{M \otimes N}$ is stable by anti-shuffle. Let $\pathLL t, \pathLL u \in \visit{M \otimes N}$ and let $\pathLL s \in \pathLL t \antishuffle \pathLL u$, we show that $\pathLL s \in \visit{M \otimes N}$ by induction on the length of $\pathLL s$. Notice first that, from Proposition~\ref{visit-tensor-reg}, there exist paths $\pathLL t_1, \pathLL u_1 \in \visit{M}^x$ and $\pathLL t_2, \pathLL u_2 \in \visit{N}^y$ such that $\pathLL t \in \kappa_\symtensor (\pathLL t_1 \shuffle \pathLL t_2)$ and $\pathLL u \in \kappa_\symtensor (\pathLL u_1 \shuffle \pathLL u_2)$. In the case $\pathLL s$ of length $1$, either $\pathLL s = \daimon$ or $\pathLL s = \kappa_\symtensor$, thus the result is immediate. So suppose $\pathLL s = \pathLL s'\kappa^-\kappa^+$ and by induction hypothesis $\pathLL s' \in \visit{M \otimes N}$. Hence, it follows from Proposition~\ref{visit-tensor-reg} that there exist paths $\pathLL s_1 \in \visit{M}^x$ and $\pathLL s_2 \in \visit{N}^y$ such that $\pathLL s' \in \kappa_\symtensor (\pathLL s_1 \shuffle \pathLL s_2)$. Without loss of generality, we can suppose that $\kappa^-$ is an action of $\pathLL t_1$, hence of $\pathLL t$. We study the different cases, proving each time either that $\pathLL s \in \visit{M \otimes N}$ or that the case is impossible.
    
      \begin{itemize}
      
      \item Either $\kappa^+ = \daimon$. In that case, $\pathLL s_1\kappa^-\daimon$ is a negative quasi-path. As $\pathLL s$ is a path and $\pathLL s \in \kappa_\symtensor(\pathLL s_1\kappa^-\daimon \shuffle \pathLL s_2)$, by Lemma~\ref{lem:dual_inverse_shuffle}, we have moreover that $\pathLL s_1\kappa^-\daimon$ is a path. Notice that $\kappa_\symtensor\triv{\pathLL s_1\kappa^-} = \trivv{\kappa^-}{\pathLL s} = \trivv{\kappa^-}{\pathLL t} = \kappa_\symtensor\trivv{\kappa^-}{\pathLL t_1}$. Hence $\triv{\pathLL s_1\kappa^-} = \trivv{\kappa^-}{\pathLL t_1}$ is a trivial view of $\beh M^x$. Let $\viewseq t\kappa^- = \triv{\pathLL s_1\kappa^-}$. By Lemma~\ref{triv-view-path}, $\pathLL s_1$ is a shuffle of anti-shuffles of trivial views of $\beh M^x$, one of which is the trivial view $\viewseq t$. Then remark that $\pathLL s_1\kappa^-\daimon$ is also a shuffle of anti-shuffles of trivial views of $\beh M^x$, replacing $\viewseq t$ by $\viewseq t\kappa^-\daimon$ (note that $\viewseq t\kappa^-\daimon$ is indeed a trivial view of $\beh M^x$ since $\viewseq t\kappa^-\daimon = \triv{\pathLL t_0\kappa^-\daimon}$ where $\pathLL t_0\kappa^-$ is the prefix of $\pathLL t_1$ ending with $\kappa^-$, and $\pathLL t_0\kappa^-\daimon \in \visit{M}^x$ by Lemma~\ref{daimon_visit}). It follows from Proposition~\ref{reg2} that $\pathLL s_1\kappa^-\daimon \in \visit{M}^x$. Finally, as $\pathLL s \in \kappa_\symtensor (\pathLL s_1\kappa^-\daimon \shuffle \pathLL s_2)$ and by Proposition~\ref{visit-tensor-reg}, we have $\pathLL s \in \visit{M \otimes N}$.
        
      \item Or $\kappa^+$ is a proper action of $\pathLL t_1$, hence of $\pathLL t$. Remark that $\view{\pathLL s'\kappa^-} = \view{\kappa_\symtensor\pathLL s_1\kappa^-} = \kappa_\symtensor\view{\pathLL s_1\kappa^-}$, thus $\mathrm{just}(\kappa^+)$ appears in $\view{\pathLL s_1\kappa^-}$ hence $\pathLL s_1\kappa^-\kappa^+$ is a (negative) quasi-path. As $\pathLL s$ is a path and as $\pathLL s \in \kappa_\symtensor(\pathLL s_1\kappa^-\kappa^+ \shuffle \pathLL s_2)$, by Lemma~\ref{lem:dual_inverse_shuffle} $\pathLL s_1\kappa^-\kappa^+$ is a path. We already know from previous item that $\pathLL s_1\kappa^-\daimon \in \visit{M}^x$. Notice that $\kappa_\symtensor\triv{\pathLL s_1\kappa^-\kappa^+} = \trivv{\kappa^+}{\pathLL s} = \trivv{\kappa^+}{\pathLL t} = \kappa_\symtensor\trivv{\kappa^+}{\pathLL t_1}$. Hence $\triv{\pathLL s_1\kappa^-\kappa^+} = \trivv{\kappa^+}{\pathLL t_1}$ is a trivial view of $\beh M^x$. Let $\viewseq u\kappa^+ = \triv{\pathLL s_1\kappa^-\kappa^+}$. By Lemma~\ref{triv-view-path}, $\pathLL s_1\kappa^-\daimon$ is a shuffle of anti-shuffles of trivial views of $\beh M^x$, one of which is the trivial view $\viewseq u\daimon$. Remark that $\pathLL s_1\kappa^-\kappa^+$ is also a shuffle of anti-shuffles of trivial views of $\beh M^x$, replacing $\viewseq u\daimon$ by $\viewseq u\kappa^+$. By Proposition~\ref{reg2}, $\pathLL s_1\kappa^-\kappa^+ \in \visit{M}^x$. Finally, as $\pathLL s \in \kappa_\symtensor (\pathLL s_1\kappa^-\kappa^+ \shuffle \pathLL s_2)$ and by Proposition~\ref{visit-tensor-reg}, we have $\pathLL s \in \visit{M \otimes N}$.
        
      \item Or $\kappa^+$ is a proper action of $\pathLL u_1$, hence of $\pathLL u$. The reasoning is similar to previous item, using $\pathLL u$ and $\pathLL u_1$ instead of $\pathLL t$ and $\pathLL t_1$ respectively.
        
      \item Or $\kappa^+$ is a proper action of $\pathLL t_2$, hence of $\pathLL t$. This is impossible, being given the structure of $\pathLL s$: the action $\kappa_0^+$ following the negative action $\kappa^-$ in $\pathLL t$ is necessarily in $\pathLL t_1$ (due to the structure of a shuffle), hence the action following $\kappa^-$ in $\pathLL s$ is necessarily either $\kappa_0^+$ (hence in $\pathLL t_1$) or in $\pathLL u$.

      \item Or $\kappa^+$ is a proper action of $\pathLL u_2$, hence of $\pathLL u$: this case also leads to a contradiction. We know from previous item that a positive action of $\pathLL t_2$ cannot immediately follow a negative action of $\pathLL t_1$ in $\pathLL s$. Similarly, a positive action of $\pathLL u_2$ (resp.\ $\pathLL t_1$, $\pathLL u_1$) cannot immediately follow a negative action of $\pathLL u_1$ (resp.\ $\pathLL t_2$, $\pathLL u_2$) in $\pathLL s$. Suppose that there exists a positive action $\kappa_0^+$ of $\pathLL u_2$ (or resp.\ $\pathLL t_2$, $\pathLL u_1$, $\pathLL t_1$) which follows immediately a negative action $\kappa_0^-$ of $\pathLL t_1$ (or resp.\ $\pathLL u_1$, $\pathLL t_2$, $\pathLL u_2$). Let $\pathLL s_0\kappa_0^-\kappa_0^+$ be the shortest prefix of $\pathLL s$ satisfying such a property, say $\kappa_0^+$ is an action of $\pathLL u_2$ and $\kappa_0^-$ is an action of $\pathLL t_1$. Then the view $\view{\pathLL s_0\kappa_0^-}$ is necessarily only made of $\kappa_\symtensor$ and of actions from $\pathLL t_1$ or $\pathLL u_1$, thus it does not contain $\mathrm{just}(\kappa_0^+)$ (where $\kappa_0^+$ cannot be initial because $\beh N$ is negative), i.e., $\pathLL s$ does not satisfy P-visibility: contradiction.
      \end{itemize}
\end{proof}

\begin{corollary} \label{reg-arrow}
  If $\beh N$ and $\beh P$ are regular, then $\beh N \multimap \beh P$ is regular.
\end{corollary}

\subsection{Proofs of Propositions~\ref{prop_pure_stable} and \ref{prop_arrow_princ}: Purity and Connectives} \label{sub-pur}

\begin{proof}[Proof (Proposition~\ref{prop_pure_stable})]
We must prove:
  \begin{itemize}
    \item If $\beh N$ is pure then $\shpos \beh N$ is pure.
    \item If $\beh P$ is pure then $\shneg \beh P$ is pure.
    \item If $\beh M$ and $\beh N$ are pure then $\beh M \oplus \beh N$ is pure.
    \item If $\beh M$ and $\beh N$ are pure then $\beh M \otimes \beh N$ is pure.
  \end{itemize}
  For the shifts and plus, the result is immediate given the form of visitable paths of $\shpos \beh N$, $\shneg \beh P$ and $\beh M \oplus \beh N$ (Propositions~\ref{visit-sh} and \ref{visit-pl}). Let us prove the result for the tensor. 

Let $\pathLL s = \pathLL s' \daimon \in V_{\beh M \otimes \beh N}$. According to Proposition~\ref{visit-tensor}, either $\pathLL s = \daimon$ or there exist $\pathLL s_1 \in \visit{M}^x$ and $\pathLL s_2 \in \visit{N}^y$ such that $\pathLL s \in \kappa_\symtensor(\pathLL s_1 \shuffle \pathLL s_2)$. If $\pathLL s = \daimon$ then it is extensible with $\kappa_\symtensor$, so suppose $\pathLL s \in \kappa_\symtensor(\pathLL s_1 \shuffle \pathLL s_2)$. Without loss of generality, suppose $\pathLL s_1 = \pathLL s_1'\daimon$. Since $\beh M$ is pure, $\pathLL s_1$ is extensible: there exists a proper positive action $\kappa^+$ such that $\pathLL s_1'\kappa^+ \in \visit{M}^x$. Then, note that $\pathLL s'\kappa^+$ is a path: indeed, since $\pathLL s_1'\kappa^+$ is a path, the justification of $\kappa^+$ appears in $\view{\pathLL s_1'} = \view{\pathLL s'}$. Moreover $\pathLL s'\kappa^+ \in \kappa_\symtensor(\visit{M}^x \shuffle \visit{N}^y)$, let us show that $\pathLL s'\kappa^+ \in V_{\beh M\otimes \beh N}$. Let $\pathLL t \in \visit{M}^x \shuffle \visit{N}^y$ and $\kappa^-$ a negative action such that $\overline{\kappa_\symtensor\pathLL t \kappa^-}$ is a path of $\completed{\dual{\pathLL s'\kappa^+}}$, and by Proposition~\ref{visit-tensor} it suffices to show that $\pathLL t \kappa^- \daimon \in \visit{M}^x \shuffle \visit{N}^y$. But $\completed{\dual{\pathLL s'\kappa^+}} = \completed{\overline{\pathLL s'\kappa^+}\daimon} = \completed{\overline{\pathLL s'}} = \completed{\dual{\pathLL s}}$, therefore  $\overline{\kappa_\symtensor\pathLL t \kappa^-}$ is a path of $\completed{\dual{\pathLL s}}$. Since $\pathLL s \in V_{\beh M \otimes \beh N}$, by Proposition~\ref{visit-tensor} we get $\pathLL t \kappa^- \daimon \in \visit{M}^x \shuffle \visit{N}^y$. Finally $\pathLL s'\kappa^+ \in V_{\beh M\otimes \beh N}$, hence $\pathLL s$ is extensible.
\end{proof}

\begin{proof}[Proof (Proposition~\ref{prop_arrow_princ})]
Since $\beh N$ and $\beh P$ are regular, $\visit{(N \multimap P)^\perp} = \kappa_\symtensor (\visit{N}^x \shuffle \dual{\visit{P}^y}) \cup \{\daimon\}$ by Corollary~\ref{visit-arrow-reg}. Let $\pathLL s \in V_{({\beh N \multimap \beh P})^\perp}$ and suppose $\dual{\pathLL s}$ is $\daimon$-ended, i.e., $\pathLL s$ is $\daimon$-free. We must show that either $\dual{\pathLL s}$ is extensible or $\dual{\pathLL s}$ is not well-bracketed. The path $\pathLL s$ is of the form $\pathLL s = \kappa_\symtensor\pathLL s'$ and there exist $\daimon$-free paths $\pathLL t \in \visit{N}^x$ and $\pathLL u \in \dual{\visit{P}^y}$ such that $\pathLL s' \in \pathLL t \shuffle \pathLL u$.
We are in one of the following situations:
\begin{itemize}
\item Either $\dual{\pathLL u} \in \visit{P}^y$ is not well-bracketed, hence neither is $\dual{\pathLL s}$.
\item Otherwise, since $\beh P$ is quasi-pure, $\dual{\pathLL u} = \overline{\pathLL u}\daimon$ is extensible, i.e., there exists a proper positive action $\kappa_{\pathLL u}^+$ such that $\overline{\pathLL u}\kappa_{\pathLL u}^+ \in \visit{P}^y$. If $\overline{\pathLL s}\kappa_{\pathLL u}^+$ is a path, then $\overline{\pathLL s}\kappa_{\pathLL u}^+ \in \visit{N \multimap P}$, hence $\dual{\pathLL s}$ is extensible: indeed, $\dual{\overline{\pathLL s}\kappa_{\pathLL u}^+} = \pathLL s\overline{\kappa_{\pathLL u}^+}\daimon \in \kappa_\symtensor(\pathLL t \shuffle \pathLL u \overline{\kappa_{\pathLL u}^+}\daimon)$, thus $\pathLL s\overline{\kappa_{\pathLL u}^+}\daimon \in \kappa_\symtensor(\visit{N}^x \shuffle \dual{\visit{P}^y})$. In the case $\overline{\pathLL s}\kappa_{\pathLL u}^+$ is not a path, this means that $\kappa_{\pathLL u}^+$ is justified by an action $\kappa_{\pathLL u}^-$ that does not appear in $\view{\overline{\pathLL s}}$, thus we have something of the form:

\begin{center}
 \begin{tikzpicture}
   
   \node (a) {$\overline{\pathLL s}\kappa_{\pathLL u}^+ = \hspace{.5cm} \dots \hspace{.5cm} \kappa^+ \hspace{.5cm} \dots \hspace{.5cm} \kappa_{\pathLL u}^- \hspace{.5cm} \dots \hspace{.5cm} \kappa^- \hspace{.5cm} \dots \hspace{.5cm} \kappa_{\pathLL u}^+$} ; 
  
 \draw[->,blue, thick] ($(a)+(4.3,.3)$) to [ out = 130, in = 50] ($(a)+(.3,.3)$) ;
 \draw[->,blue, thick] ($(a)+(2.2,.3)$) to [ out = 130, in = 50] ($(a)+(-1.8,.3)$) ;

 \draw[dotted, orange, thick] ($(a)+(4.3,-.3)$) to [ out = -90, in = -90] ($(a)+(3.4,-.3)$) ;
 \draw[dotted, orange, thick] ($(a)+(3.4,-.3)$) to [ out = -90, in = -90] ($(a)+(2.9,-.3)$) ;
 \draw[dotted, orange, thick] ($(a)+(2.9,-.3)$) to [ out = -90, in = -90] ($(a)+(2.2,-.3)$) ;
 \draw[-, orange, thick] ($(a)+(2.2,-.3)$) to [ out = -130, in = -50] ($(a)+(-1.8,-.3)$) ;
 \draw[->, dotted, orange, thick] ($(a)+(-1.8,-.3)$) to [ out = -90, in = -90] ($(a)+(-2.8,-.3)$) ;

\node[blue] (just1) at ($(a)+(2.9,.8)$) {just.} ;
\node[blue] (just2) at ($(a)+(-.4,.8)$) {just.} ;

\node[orange] (view) at ($(a)+(2.6,-1)$) {view $\view{\overline{\pathLL s}\kappa_{\pathLL u}^+}$} ;

\end{tikzpicture}
\end{center}

If $\kappa^-$ comes from $\overline{\pathLL t}$, and thus also $\kappa^+$, then $\overline s$ is not well-bracketed, indeed: since $\kappa_{\pathLL u}^-$ is hereditarily justified by $\overline{\kappa_\symtensor}$ and by no action from $\overline{\pathLL t}$, we have:

\begin{center}
 \begin{tikzpicture}
   
   \node (a) {$\overline{\pathLL s} = \hspace{.5cm} \overline{\kappa_\symtensor} \hspace{.5cm} \dots \hspace{.5cm} \kappa^+ \hspace{.5cm} \dots \hspace{.5cm} \kappa_{\pathLL u}^- \hspace{.5cm} \dots \hspace{.5cm} \kappa^- \hspace{.5cm} \dots$} ;

 \draw[->,blue, thick] ($(a)+(2.8,.3)$) to [out = 130, in = 50] ($(a)+(-.9,.3)$) ;

 \draw[dotted, blue, thick] ($(a)+(.9,-.3)$) to [ out = -90, in = -90] ($(a)+(.2,-.3)$) ;
 \draw[dotted, blue, thick] ($(a)+(.2,-.3)$) to [ out = -90, in = -90] ($(a)+(-.3,-.3)$) ;
 \draw[blue, thick] ($(a)+(-.3,-.3)$) to [ out = -90, in = -90] ($(a)+(-1.7,-.3)$) ;
 \draw[dotted, blue, thick] ($(a)+(-1.7,-.3)$) to [ out = -90, in = -90] ($(a)+(-2.4,-.3)$) ;
 \draw[->, dotted, blue, thick] ($(a)+(-2.4,-.3)$) to [out = -90, in = -90] ($(a)+(-2.9,-.3)$) ;

\node[blue] (just1) at ($(a)+(.9,.8)$) {just.} ;
\node[blue] (just2) at ($(a)+(-.9,-1)$) {just.} ;

\end{tikzpicture}
\end{center}

So suppose now that $\kappa^-$ comes from $\overline{\pathLL u}$, thus also $\kappa^+$. We know that $\view{\overline{\pathLL u}}$ contains $\kappa_{\pathLL u}^- = \mathrm{just}(\kappa_{\pathLL u}^+)$, thus in particular $\view{\overline{\pathLL u}}$ does not contain $\kappa^-$; on the contrary, we have seen that $\view{\overline{\pathLL s}}$ contains $\kappa^-$. By definition of the view of a sequence, this necessarily means that, in $\overline{\pathLL s}$, between the action $\kappa^-$ and the end of the sequence, the following happens: $\view{\overline{\pathLL s}}$ comes across an action $\alpha_{\pathLL t}^-$ from $\overline{\pathLL t}$, justified by an action $\alpha_{\pathLL t}^+$ also from $\overline{\pathLL t}$, making the view miss at least one action $\alpha_{\pathLL u}$ from $\overline{\pathLL u}$ appearing in $\view{\overline{\pathLL u}}$, as depicted below.

\begin{center}
 \begin{tikzpicture}
   
   \node (a) {$\overline{\pathLL s} = \hspace{.5cm} \overline{\kappa_\symtensor} \hspace{.5cm} \dots \hspace{.5cm} \kappa^- \hspace{.5cm} \dots \hspace{.5cm} \alpha_{\pathLL t}^+ \hspace{.5cm} \dots \hspace{.5cm} \alpha_{\pathLL u} \hspace{.5cm} \dots \hspace{.5cm} \alpha_{\pathLL t}^- \hspace{.5cm} \dots$} ;
   
   \draw[dotted, blue, thick] ($(a)+(1.8,.3)$) to [out = 130, in = 50] ($(a)+(.8,.3)$) ;
   \draw[blue, thick] ($(a)+(.8,.3)$) to [ out = 90, in = 90] ($(a)+(-1.1,.3)$) ;
   \draw[dotted, blue, thick] ($(a)+(-1.1,.3)$) to [out = 90, in = 90] ($(a)+(-2.9,.3)$) ;
   \draw[->, dotted, blue, thick] ($(a)+(-2.9,.3)$) to [out = 90, in = 90] ($(a)+(-3.9,.3)$) ;
   
   \node[blue] (just1) at ($(a)+(1.3,.8)$) {just.} ;
   
   \draw[dotted, orange, thick] ($(a)+(4.8,-.3)$) to [ out = -90, in = -90] ($(a)+(3.8,-.3)$) ;
   \draw[-, orange, thick] ($(a)+(3.8,-.3)$) to [ out = -130, in = -50] ($(a)+(0,-.3)$) ;
   \draw[dotted, orange, thick] ($(a)+(0,-.3)$) to [ out = -90, in = -90] ($(a)+(-1.4,-.3)$) ;
   \draw[dotted, orange, thick] ($(a)+(-1.4,-.3)$) to [ out = -90, in = -90] ($(a)+(-2.1,-.3)$) ;
   \draw[dotted, orange, thick] ($(a)+(-2.1,-.3)$) to [ out = -90, in = -90] ($(a)+(-2.8,-.3)$) ;
   \draw[->, dotted, orange, thick] ($(a)+(-2.8,-.3)$) to [ out = -90, in = -90] ($(a)+(-3.9,-.3)$) ;
   
   \node[orange] (view) at ($(a)+(1.9,-.8)$) {view $\view{\overline{\pathLL s}}$} ;
   
 \end{tikzpicture}
\end{center}

Since $\alpha_{\pathLL u}$ is hereditarily justified by $\overline{\kappa_\symtensor}$ and by no action from $\overline{\pathLL t}$, the path $\overline{\pathLL s}$ is not well-bracketed: the justifications of $\alpha_{\pathLL u}$ and of $\alpha_{\pathLL t}^-$ intersect.

To sum up, we have proved that in the case when $\dual{\pathLL u} = \overline{\pathLL u}\daimon$ is extensible, either $\dual{\pathLL s}$ is extensible too or it is not well-bracketed.
\end{itemize}
Hence $\beh N \multimap \beh P$ is quasi-pure.
\end{proof}

\section{Proofs of Section~\ref{sec-induct}}

In this section we prove:
\begin{itemize}
\item that the functions $\phi^{A}_{\sigma}$ are Scott-continuous (Proposition~\ref{prop_scott_conti}),
\item internal completeness for particular infinite unions of behaviours (Theorem~\ref{thm_union_beh}),
\item two lemmas of Subsection~\ref{reg_pure_data} (Lemmas~\ref{lem_incarn_hier} and \ref{lem_cb_basis}).
\end{itemize}

\subsection{Proof of Proposition~\ref{prop_scott_conti}}

\begin{lemma} \label{lem_intcomp_set}
Let $E, F$ be sets of atomic negative designs and $G$ be a set of atomic positive designs.
  \begin{enumerate}
    \item $\shpos (E^{\perp\perp}) = \symshpos\langle E\rangle^{\perp\perp}$
    \item $\shneg (G^{\perp\perp}) = \setst{\design n}{\proj{\design n}{\symshneg} \in \symshneg(x).G^x}^{\perp\perp}$
    \item $(E^{\perp\perp}) \oplus (F^{\perp\perp})  = (\symplus_1\langle E\rangle \cup \symplus_2\langle F\rangle)^{\perp\perp}$
    \item $(E^{\perp\perp}) \otimes (F^{\perp\perp}) = \symtensor\langle E, F\rangle^{\perp\perp}$
  \end{enumerate}
\end{lemma}

\begin{proof}
  We prove (1) and (2), the other cases are very similar to (1).
  \begin{enumerate}
  \item $\symshpos\langle E\rangle^{\perp\perp} = \setst{\design n}{\proj{\design n}{\symshneg} \in \symshneg(x).(E^\perp)^x}^\perp = (\shneg (E^\perp))^\perp = (\shpos (E^{\perp\perp}))^{\perp\perp} = \shpos (E^{\perp\perp})$,
  \item $\setst{\design n}{\proj{\design n}{\symshneg} \in \symshneg(x).G^x}^{\perp\perp} = \setst{\symshpos\langle\design m\rangle}{\design m \in G^\perp}^\perp = (\shpos (G^\perp))^\perp = \shneg (G^{\perp\perp})$,
  \end{enumerate}
using the definition of the orthogonal, internal completeness, and Lemma~\ref{lem-inc-shneg}.
\end{proof}

\begin{proof}[Proof (Proposition~\ref{prop_scott_conti})]
  By induction on $A$, we prove that for every $X$ and every $\sigma$ the function $\phi^{A}_{\sigma}$ is continuous. Note that $\phi^{A}_{\sigma}$ is continuous if and only if for every directed subset $\mathbb P \subseteq \mathcal B^+$ we have $\bigvee_{\beh P \in \mathbb P} (\interpret{A}^{\sigma, X \mapsto \beh P}) = \interpret{A}^{\sigma, X \mapsto \bigvee \mathbb P}$. The cases $A = Y \in \mathcal V$ and $A = a \in \mathcal S$ are trivial, and the case $A = A_1 \oplus^+ A_2$ is very similar to the tensor, hence we only treat the two remaining cases. Let $\mathbb P \subseteq \mathcal B^+$ be directed.
  \begin{itemize}
    
  \item Suppose $A = A_1 \otimes^+ A_2$, thus $\interpret{A}^{\sigma, X \mapsto \beh P} = \interpret{A_1}^{\sigma, X \mapsto \beh P} \otimes^+ \interpret{A_2}^{\sigma, X \mapsto \beh P}$, with both functions $\phi^{A_i}_{\sigma}: \beh P \mapsto \interpret{A_i}^{\sigma, X \mapsto \beh P}$ continuous by induction hypothesis. For any positive behaviour $\beh P$, let us write $\sigma_{\beh P}$ instead of $\sigma, X \mapsto \beh P$. We have
    \[ \bigvee_{\beh P \in \mathbb P} \interpret{A}^{\sigma_{\beh P}} = (\bigcup_{\beh P \in \mathbb P} \interpret{A}^{\sigma_{\beh P}})^{\perp\perp} = (\bigcup_{\beh P \in \mathbb P} (\interpret{A_1}^{\sigma_{\beh P}} \otimes^+ \interpret{A_2}^{\sigma_{\beh P}}))^{\perp\perp} \]
    Let us show that 
    \[\bigcup_{\beh P \in \mathbb P} (\interpret{A_1}^{\sigma_{\beh P}} \otimes^+ \interpret{A_2}^{\sigma_{\beh P}}) = \symtensor \langle \bigcup_{\beh P' \in \mathbb P} \shneg \interpret{A_1}^{\sigma_{\beh P'}}, \bigcup_{\beh P'' \in \mathbb P} \shneg \interpret{A_2}^{\sigma_{\beh P''}} \rangle \cup \{\daimon\} \tag{$*$} \label{star}\]
    By internal completeness we have $\interpret{A_1}^{\sigma_{\beh P}} \otimes^+ \interpret{A_2}^{\sigma_{\beh P}} = \symtensor \langle \shneg \interpret{A_1}^{\sigma_{\beh P}}, \shneg \interpret{A_2}^{\sigma_{\beh P}} \rangle \cup \{\daimon\}$ for every $\beh P \in \mathbb P$. The inclusion $(\subseteq)$ of (\ref{star}) is then immediate, so let us prove $(\supseteq)$. First, indeed, $\daimon$ belongs to the left side. Let $\beh P', \beh P'' \in \mathbb P$, let $\design m \in \shneg \interpret{A_1}^{\sigma_{\beh P'}}$, $\design n \in \shneg \interpret{A_2}^{\sigma_{\beh P''}}$, and let us show that $\symtensor \langle \design m, \design n \rangle \in \interpret{A_1}^{\sigma_{\beh P}} \otimes^+ \interpret{A_2}^{\sigma_{\beh P}}$ where $\beh P = \beh P' \vee \beh P''$ (note that $\beh P \in \mathbb P$ since $\mathbb P$ is directed). By induction hypothesis, $\phi^{A_1}_{\sigma}$ is continuous, thus in particular increasing; since $\beh P' \subseteq \beh P$, it follows that $\interpret{A_1}^{\sigma_{\beh P'}} = \phi^{A_1}_{\sigma}(\beh P') \subseteq \phi^{A_1}_{\sigma}(\beh P) = \interpret{A_1}^{\sigma_{\beh P}}$. Similarly, $\interpret{A_2}^{\sigma_{\beh P''}} \subseteq \interpret{A_2}^{\sigma_{\beh P}}$. We get $\symtensor \langle \design m, \design n \rangle \in \symtensor \langle \shneg \interpret{A_1}^{\sigma_{\beh P}}, \shneg \interpret{A_2}^{\sigma_{\beh P}} \rangle \subseteq \interpret{A_1}^{\sigma_{\beh P}} \otimes^+ \interpret{A_2}^{\sigma_{\beh P}}$, using internal completeness for $\shneg$, which proves (\ref{star}). Using internal completeness, Lemma~\ref{lem_intcomp_set} and induction hypothesis, we deduce

    \begin{align*}
      (\bigcup_{\beh P \in \mathbb P} (\interpret{A_1}^{\sigma_{\beh P}} \otimes^+ \interpret{A_2}^{\sigma_{\beh P}}))^{\perp\perp} & = \symtensor \langle \bigcup_{\beh P' \in \mathbb P} \shneg \interpret{A_1}^{\sigma_{\beh P'}}, \bigcup_{\beh P'' \in \mathbb P} \shneg \interpret{A_2}^{\sigma_{\beh P''}} \rangle^{\perp\perp} \\
      & = (\bigcup_{\beh P' \in \mathbb P} \shneg \interpret{A_1}^{\sigma_{\beh P'}})^{\perp\perp} \otimes (\bigcup_{\beh P'' \in \mathbb P} \shneg \interpret{A_2}^{\sigma_{\beh P''}})^{\perp\perp} \\ 
      & = (\bigcup_{\beh P' \in \mathbb P} \interpret{A_1}^{\sigma_{\beh P'}})^{\perp\perp} \otimes^+ (\bigcup_{\beh P'' \in \mathbb P} \interpret{A_2}^{\sigma_{\beh P''}})^{\perp\perp} \\ 
      & = \interpret{{A_1}}^{\sigma, X \mapsto \bigvee \mathbb P} \otimes^+ \interpret{{A_2}}^{\sigma, X \mapsto \bigvee \mathbb P} \\ 
      & = \interpret{A}^{\sigma, X \mapsto \bigvee \mathbb P}
    \end{align*}
    Consequently $\phi^{A}_{\sigma}$ is continuous.
  \item If $A = \mu Y.A_0$, define $f_0 : \beh Q \mapsto \interpret{A_0}^{\sigma, X \mapsto \bigvee \mathbb P, Y \mapsto \beh Q}$ and, for every $\beh P \in \mathcal B^+$, $f_{\beh P} : \beh Q \mapsto \interpret{A_0}^{\sigma, X \mapsto \beh P, Y \mapsto \beh Q}$. Those functions are continuous by induction hypothesis, thus using Kleene fixed point theorem we have $\mathrm{lfp}(f_0) = \bigvee_{n \in \mathbb N}{f_0}^n(\daimon) \mbox{ and } \mathrm{lfp}(f_{\beh P}) = \bigvee_{n \in \mathbb N}{f_{\beh P}}^n(\daimon)$.
    Therefore $\bigvee_{\beh P \in \mathbb P} (\interpret{A}^{\sigma, X \mapsto \beh P}) = \bigvee_{\beh P \in \mathbb P} (\mathrm{lfp}(f_{\beh P})) = \bigvee_{\beh P \in \mathbb P} (\bigvee_{n \in \mathbb N}{f_{\beh P}}^n(\daimon)) = \bigvee_{n \in \mathbb N} (\bigvee_{\beh P \in \mathbb P} {f_{\beh P}}^n(\daimon))$.

    For every $\beh Q \in \mathcal B^+$ the function $g_{\beh Q}: \beh P \mapsto f_{\beh P}(\beh Q)$ is continuous by induction hypothesis, hence $f_0(\beh Q) = \bigvee_{\beh P \in \mathbb P} f_{\beh P}(\beh Q)$.
    From this, we prove easily by induction on $m$ that for every $\beh Q \in \mathcal B^+$ we have ${f_0}^m(\beh Q) = \bigvee_{\beh P \in \mathbb P} {f_{\beh P}}^m(\beh Q)$.
    Thus $\bigvee_{\beh P \in \mathbb P} (\interpret{A}^{\sigma, X \mapsto \beh P}) = \bigvee_{n \in \mathbb N}{f_0}^n(\daimon) = \mathrm{lfp}(f_0) = \interpret{A}^{\sigma, X \mapsto \bigvee \mathbb P}$. We conclude that the function $\phi^{A}_{\sigma}$ is continuous.
  \end{itemize}
\end{proof}

\subsection{Proof of Theorem~\ref{thm_union_beh}}

Before proving Theorem~\ref{thm_union_beh} we need some lemmas. Suppose $(\beh A_n)_{n \in \mathbb N}$ is an infinite sequence of regular behaviours such that for all $n \in \mathbb N$, $|\beh A_n| \subseteq |\beh A_{n+1}|$; the simplicity hypothesis is not needed for now. Let us note $\beh A = \bigcup_{n \in \mathbb N} \beh A_n$. Notice that the definition of visitable paths can harmlessly be extended to any set $E$ of designs of same polarity, even if it is not a behaviour; the same applies to the definition of incarnation, provided that $E$ satisfies the following: if $\design d, \design e_1, \design e_2 \in E$ are cut-free designs such that $\design e_1 \sqsubseteq \design d$ and $\design e_2 \sqsubseteq \design d$ then there exists $\design e \in E$ cut-free such that $\design e \sqsubseteq \design e_1$ and $\design e \sqsubseteq \design e_2$. In particular, as a union of behaviours, $\beh A$ satisfies this condition.

\begin{lemma} \label{0lem_visit_hier} \label{0lem_visit_union}  \label{0lem_incarn_union}
  \begin{enumerate}
  \item $\forall n \in \mathbb N$, $V_{\beh A_n} \subseteq V_{\beh A_{n+1}}$.
  \item $V_{\bigcup_{n \in \mathbb N} \beh A_n} = \bigcup_{n \in \mathbb N} V_{\beh A_n}$.
  \item $|\bigcup_{n \in \mathbb N} \beh A_n| = \bigcup_{n \in \mathbb N} |\beh A_n|$.
  \end{enumerate}
\end{lemma}

\begin{proof}
  \begin{enumerate}
  \item Fix $n$ and let $\pathLL s \in V_{A_n}$. There exist $\design d \in |\beh A_n|$ such that $\pathLL s$ is a path of $\design d$. Since $|\beh A_n| \subseteq |\beh A_{n+1}|$ we have $\design d \in |\beh A_{n+1}|$, thus by regularity of $\beh A_{n+1}$, $\pathLL s \in V_{A_{n+1}}$.
  \item $(\subseteq)$ Let $\pathLL s \in \visit A$. There exist $n \in \mathbb N$ and $\design d \in |\beh A_n|$ such that $\pathLL s$ is a path of $\design d$. By regularity of $\beh A_n$ we have $\pathLL s \in V_{\beh A_n}$. \\
    $(\supseteq)$ Let $m \in \mathbb N$ and $\pathLL s \in V_{\beh A_m}$. For all $n \ge m$, $V_{\beh A_m} \subseteq V_{\beh A_n}$ by previous item, thus $\pathLL s \in V_{\beh A_n}$. Hence if we take $\design e = \completed{\dual{\pathLL s}}$, we have $\design e \in {\beh A_n}^\perp$ for all $n \ge m$ by monotonicity. We deduce $\design e \in \bigcap_{n \ge m}{\beh A_n}^\perp = (\bigcup_{n \ge m} \beh A_n)^\perp = (\bigcup_{n \in \mathbb N} \beh A_n)^\perp = \beh A^\perp$. Let $\design d \in \beh A_m$ such that $\pathLL s$ is a path of $\design d$; we have $\design d \in \beh A$ and $\design e \in \beh A^\perp$, thus $\interseq{\design d}{\design e} = \pathLL s \in \visit A$.
  \item $(\subseteq)$ Let $\design d$ be cut-free and minimal for $\sqsubseteq$ in $\beh A$. There exists $m \in \mathbb N$ such that $\design d \in \beh A_m$. Thus $\design d$ is minimal for $\sqsubseteq$ in $\beh A_m$ otherwise it would not be minimal in $\beh A$, hence the result. \\
    $(\supseteq)$ Let $m \in \mathbb N$, and let $\design d \in |\beh A_m|$. By hypothesis, $\design d \in |\beh A_n|$ for all $n \ge m$. Suppose $\design d$ is not in $|\beh A|$, so there exists $\design d' \in \beh A$ such that $\design d' \sqsubseteq \design d$ and $\design d' \neq \design d$. In this case, there exists $n \ge m$ such that $\design d' \in \beh A_n$, but this contradicts the fact that $\design d \in |\beh A_n|$.
  \end{enumerate}
\end{proof}

\begin{lemma} \label{0lem_visit_double}
$V_{\bigcup_{n \in \mathbb N} \beh A_n} = \dual{V_{(\bigcup_{n \in \mathbb N} \beh A_n)^\perp}} = V_{(\bigcup_{n \in \mathbb N} \beh A_n)^{\perp\perp}}$.
\end{lemma}

\begin{proof}
  In this proof we use the alternative definition of regularity (Proposition~\ref{reg2}). We prove $\visit A = \dual{\visit{A^\perp}}$, and the result will follow from the fact that for any behaviour $\beh B$ (in particular if $\beh B = \beh A^{\perp\perp}$) we have $\dual{\visit{B^\perp}} = \visit B$. First note that the inclusion $\visit A \subseteq \dual{\visit{A^\perp}}$ is immediate.

  Let $\pathLL s \in \visit{A^\perp}$ and let us show that $\dual{\pathLL s} \in \visit A$. Let $\design e \in |\beh A^\perp|$ such that $\pathLL s$ is a path of $\design e$. By Lemma~\ref{triv-view-path} and the remark following it, $\pathLL s$ is in the shuffle of anti-shuffles of trivial views $\viewseq t_1, \dots, \viewseq t_k$ of $\beh A^\perp$. For every $i \le k$, suppose $\viewseq t_i = \triv{\kappa_i}$; necessarily, there exists a design $\design d_i \in \beh A$ such that $\kappa_i$ occurs in $\interseq{\design e}{\design d_i}$, i.e., such that $\viewseq t_i$ is a subsequence of $\interseq{\design e}{\design d_i}$, otherwise $\design e$ would not be in the incarnation of $\beh A^\perp$ (it would not be minimal). Let $n$ be big enough such that $\design d_1, \dots, \design d_k \in \beh A_n$, and note that in particular $\design e \in {\beh A_n}^\perp$. For all $i$, $\dual{\viewseq t_i}$ is a trivial view of $|\design d_i|_{\beh A_n}$, thus it is a trivial view of $\beh A_n$. By regularity of $\beh A_n$ we have $\dual{\viewseq t_i} \in V_{\beh A_n}$. Since $\dual{\pathLL s}$ is in the anti-shuffle of shuffles of $\dual{\viewseq t_1}, \dots, \dual{\viewseq t_k}$, we have $\dual{\pathLL s} \in V_{\beh A_n}$ using regularity again. Therefore $\dual{\pathLL s} \in \visit A$ by Lemma~\ref{0lem_visit_union}.
\end{proof}

\begin{lemma} \label{0lem_union_reg}
  $(\bigcup_{n \in \mathbb N} \beh A_n)^\perp$ and $(\bigcup_{n \in \mathbb N} \beh A_n)^{\perp\perp}$ are regular.
\end{lemma}

\begin{proof}
  Let us show $\beh A^\perp$ is regular using the equivalent definition (Proposition~\ref{reg2}).
  \begin{itemize}
  \item Let $\viewseq t$ be a trivial view of $\beh A^\perp$. By a similar argument as in the proof above, there exists $n \in \mathbb N$ such that $\dual{\viewseq t}$ is a trivial view of $\beh A_n$, thus $\dual{\viewseq t} \in V_{\beh A_n} \subseteq \visit A$. By Lemma~\ref{0lem_visit_double} $\viewseq t \in \visit{A^\perp}$.
  \item Let $\pathLL s, \pathLL t \in \visit{A^\perp}$. By Lemma~\ref{0lem_visit_double}, $\dual{\pathLL s}, \dual{\pathLL t} \in \visit A$. By Lemma~\ref{0lem_visit_union}(2), there exists $n \in \mathbb N$ such that $\dual{\pathLL s}, \dual{\pathLL t} \in V_{\beh A_n}$, thus by regularity of $\beh A_n$ we have $\dual{\pathLL s} \antishuffle \dual{\pathLL t}$, $\dual{\pathLL s} \shuffle \dual{\pathLL t} \subseteq V_{\beh A_n} \subseteq \visit A$, in other words $\dual{\pathLL s \shuffle \pathLL t}$, $\dual{\pathLL s \antishuffle \pathLL t} \subseteq \visit A$. By Lemma~\ref{0lem_visit_double} we deduce $\pathLL s \shuffle \pathLL t$, $\pathLL s \antishuffle \pathLL t \subseteq \visit{A^\perp}$, hence $\visit{A^\perp}$ is stable under shuffle and anti-shuffle.
  \end{itemize}
  Finally $\beh A^\perp$ is regular. We deduce that $\beh A^{\perp\perp}$ is regular since regularity is stable under orthogonality.
\end{proof}

Let us introduce some more notions for next proof. An \defined{\boldmath$\infty$-path} (resp. \defined{\boldmath$\infty$-view}) is a finite or infinite sequence of actions satisfying all the conditions of the definition of path (resp. view) but the requirement of finiteness. In particular, a finite $\infty$-path (resp. $\infty$-view) is a path (resp. a view). An \defined{\boldmath$\infty$-path} (resp. \defined{\boldmath$\infty$-view}) \defined{of} a design $\design d$ is such that any of its positive-ended prefix is a path (resp. a view) of $\design d$. We call \defined{infinite chattering} a closed interaction which diverges because the computation never ends; note that infinite chattering occurs in the interaction between two atomic designs $\design p$ and $\design n$ if and only if there exists an infinite $\infty$-path $\pathLL s$ of $\design p$ such that $\dual{\pathLL s}$ is an $\infty$-path of $\design n$ (where, when $\pathLL s$ is infinite, $\dual{\pathLL s}$ is obtained from $\pathLL s$ by simply reversing the polarities of all the actions). Given an infinite $\infty$-path $\pathLL s$, the design $\completed{\pathLL s}$ is constructed similarly to the case when $\pathLL s$ is finite (see \textsection~\ref{ord-mono}).

For the proof of the theorem, suppose now that the behaviours $(\beh A_n,)_{n \in \mathbb N}$ are simple. Remark that the second condition of simplicity implies in particular that the dual of a path in a design of a simple behaviour is a view.

\begin{proof}[Proof (Theorem~\ref{thm_union_beh})]

We must show that $\beh A^{\perp\perp} \subseteq \beh A$ since the other inclusion is trivial. Remark the following: given designs $\design d$ and $\design d'$, if $\design d \in \beh A$ and $\design d \sqsubseteq \design d'$ then $\design d' \in \beh A$. Indeed, if $\design d \in \beh A$ then there exists $n \in \mathbb N$ such that $\design d \in \beh A_n$; if moreover $\design d \sqsubseteq \design d'$ then in particular $\design d \preceq \design d'$, and by monotonicity $\design d' \in \beh A_n$, hence $\design d' \in \beh A$. Thus it is sufficient to show $|\beh A^{\perp\perp}| \subseteq \beh A$ since for every $\design d' \in \beh A^{\perp\perp}$ we have $|\design d'| \in |\beh A^{\perp\perp}|$ and $|\design d'| \sqsubseteq \design d'$.

So let $\design d \in |\beh A^{\perp\perp}|$ and suppose $\design d \notin \beh A$. First note the following: by Lemmas~\ref{0lem_visit_double} and \ref{0lem_union_reg}, every path $\pathLL s$ of $\design d$ is in $\visit{A^{\perp\perp}} = \visit A$, thus there exists $\design d' \in |\beh A|$ containing $\pathLL s$. We explore separately the possible cases, and show how they all lead to a contradiction. \\
\textbf{If {\boldmath$\design d$} has an infinite number of maximal slices} then:
\begin{itemize}
\item Either there exists a negative subdesign $\design n = \negdes{a}{\vect{x^a}}{\design p_a}$ of $\design d$ for which there is an infinity of names $a \in \mathcal A$ such that $\design p_a \neq \Omega$. In this case, let $\viewseq v$ be the view of $\design d$ such that for every action $\kappa^-$ among the first ones of $\design n$, $\viewseq v\kappa^-$ is the prefix of a view of $\design d$. All such sequences $\viewseq v\kappa^-$ being prefixes of paths of $\design d$, we deduce by regularity of $\beh A^{\perp\perp}$ and using Lemma~\ref{daimon_visit} that $\viewseq v\kappa^-\daimon \in \visit{A^{\perp\perp}}$. Let $\design d' \in |\beh A|$ be such that $\viewseq v$ is a view of $\design d'$. Since $\design d'$ is also in $\beh A^{\perp\perp}$, we deduce by Lemma~\ref{nec} that for every action $\kappa^-$ among the first ones of $\design n$, $\viewseq v\kappa^-$ is the prefix of a view of $\design d'$. Thus $\design d'$ has an infinite number of slices: contradiction.

\item Or we can find an infinite $\infty$-view $\viewseq v = (\kappa^-_0)\kappa^+_1\kappa^-_1\kappa^+_2\kappa^-_1\kappa^+_3\kappa^-_3 \dots$ of $\design d$ (the first action $\kappa^-_0$ being optional depending on the polarity of $\design d$) satisfying the following: there is an infinity of $i \in \mathbb N$ such than $\kappa^-_i$ is one of the first actions of a negative subdesign $\negdes{a}{\vect{x^a}}{\design p_a}$ of $\design d$ with at least two names $a \in \mathcal A$ such that $\design p_a \neq \Omega$. Let $\viewseq v_i$ be the prefix of $\viewseq v$ ending on $\kappa^+_i$. There is no design $\design d' \in |\beh A|$ containing $\viewseq v$, indeed: in this case, for all $i$ and all negative action $\kappa^-$ such that $\viewseq v_i\kappa^-$ is a prefix of a view of $\design d$, $\viewseq v_i\kappa^-$ would be a prefix of a view of $\design d'$ by Lemma~\ref{nec}, thus $\design d'$ would have an infinite number of slices, which is impossible since the $\beh A_n$ are simple. Thus consider $\design e = \completed{\dual{\viewseq v}}$: since all the $\viewseq v_i$ are views of designs in $|\beh A| = \bigcup_{n \in \mathbb N} |\beh A_n|$ and since the $\beh A_n$ are simple, the sequences $\dual{\viewseq v_i}$ are views, thus $\dual{\viewseq v}$ is an $\infty$-view. Therefore an interaction between a design $\design d' \in \beh A$ and $\design e$ necessarily eventually converges by reaching a daimon of $\design e$, indeed: infinite chattering is impossible since we cannot follow $\viewseq v$ forever, and interaction cannot fail after following a finite portion of $\viewseq v$ since those finite portions $\viewseq v_i$ are in $\visit A$. Hence $\design e \in \beh A^\perp$. But $\design d \not \perp \design e$, because of infinite chattering following $\viewseq v$. Contradiction.
\end{itemize}
\textbf{If {\boldmath$\design d$} has a finite number of maximal slices} $\design c_1, \dots, \design c_k$ then for every $i \le k$ there exist an $\infty$-path $\pathLL s_i$ that visits all the positive proper actions of $\design c_i$. Indeed, any (either infinite or positive-ended) sequence $\pathLL s$ of proper actions in a slice $\design c \sqsubseteq \design d$, without repetition, such that polarities alternate and the views of prefixes of $\pathLL s$ are views of $\design c$, is an $\infty$-path:
\begin{itemize}
\item (Linearity) is ensured by the fact that we are in only one slice,
\item (O-visibility) is satisfied since positive actions of $\design d$, thus also of $\design c$, are justified by the immediate previous negative action (a condition true in $|\beh A|$, thus also satisfied in $\design d$ because all its views are views of designs in $|\beh A|$)
\item (P-visibility) is natively satisfied by the fact that $\pathLL s$ is a promenade in the tree representing a design.
\end{itemize}
For example, $\pathLL s$ can travel in the slice $\design c$ as a breadth-first search on couples of nodes $(\kappa^-,\kappa^+)$ such that $\kappa^+$ is just above $\kappa^-$ in the tree, and $\kappa^+$ is proper. Then 2 cases:

\begin{itemize}
\item Either for all $i$, there exists $n_i \in \mathbb N$ and $\design d_i \in \beh A_{n_i}$ such that $\pathLL s_i$ is an $\infty$-path of $\design d_i$. Without loss of generality we can even suppose that $\design c_i \sqsubseteq \design d_i$: if it is not the case, replace some positive subdesigns (possibly $\Omega$) of $\design d_i$ by $\daimon$ until you obtain $\design d'_i$ such that $\design c_i \sqsubseteq \design d'_i$, and note that indeed $\design d'_i \in \beh A_{n_i}$ since $\design d_i \preceq \design d'_i$. Let $N = \mathrm{max}_{1 \le i \le k} (n_i)$. Since $\design d \not \in \beh A$, thus in particular $\design d \not \in \beh A_N$, there exists $\design e \in \beh A_N^\perp$ such that $\design d \not \perp \design e$. The reason of divergence cannot be infinite chattering, otherwise there would exist an infinite $\infty$-path $\pathLL t$ in $\design d$ such that $\dual{\pathLL t}$ is in $\design e$, and $\pathLL t$ is necessarily in a single slice of $\design d$ (say $\design c_i$) to ensure its linearity; but in this case we would also have $\design d_i \not \perp \design e$ where $\design d_i \in \beh A_N$, impossible. Similarly, for all (finite) path $\pathLL s$ of $\design d$, there exists $i$ such that $\pathLL s$ is a path of $\design c_i$ thus of $\design d_i \in \beh A_N$; this ensures that interaction between $\design d$ and $\design e$ cannot diverge after a finite number of steps either, leading to a contradiction.

\item Or there is an $i$ such that the (necessarily infinite) $\infty$-path $\pathLL s_i$ is in no design of $\beh A$. In this case, let $\design e = \completed{\dual{\pathLL s_i}}$ (where $\dual{\pathLL s_i}$ is a view since the $\beh A_n$ are simple), and with a similar argument as previously we have $\design e \in \beh A^\perp$ but $\design d \not \perp \design e$ by infinite chattering, contradiction.
\end{itemize}

\end{proof}

\subsection{Proofs of Subsection~\ref{reg_pure_data}}

\begin{proof}[Proof (Lemma~\ref{lem_incarn_hier})]
  By induction on $A$, we prove that for all $X \in \mathcal V$ and $\sigma: \mathrm{FV}(A) \setminus \{X\} \to \mathcal B^+$ simple and regular, the induction hypothesis consisting in all the following statements holds:
  \begin{enumerate}
  \item for all simple regular behaviours $\beh P, \beh P' \in \mathcal B^+$, if $|\beh P| \subseteq |\beh P'|$ then $|{\phi^{A}_{\sigma}}(\beh P)| \subseteq |{\phi^{A}_{\sigma}}(\beh P')|$;
  \item for all $n \in \mathbb N$, $|(\phi^{A}_{\sigma})^n(\daimon)| \subseteq |(\phi^{A}_{\sigma})^{n+1}(\daimon)|$;
  \item for all simple regular behaviour $\beh P \in \mathcal B^+$, ${\phi^{A}_{\sigma}}(\beh P)$ is simple and regular;
  \item $\interpret{\mu X.A}^{\sigma} = \bigcup_{n \in \mathbb N} (\phi^{A}_{\sigma})^n(\daimon)$.
  \item $|\interpret{\mu X.A}^{\sigma}| = \bigcup_{n \in \mathbb N} |(\phi^{A}_{\sigma})^n(\daimon)|$.
  \end{enumerate}
  Let us write $\sigma_{\beh P}$ for $\sigma, X\mapsto \beh P$. Note that the base cases are immediate. If $A = A_1 \oplus^+ A_2$ or $A = A_1 \otimes^+ A_2$ then:
  \begin{enumerate}
  \item Follows from the incarnated form of internal completeness (in Theorem~\ref{thm_intcomp_all}).
  \item Easy by induction on $n$, using previous item.
  \item Regularity of ${\phi^{A}_{\sigma}}(\beh P)$ comes from Proposition~\ref{prop_reg_stable}, and simplicity is easy since the structure of the designs in $\interpret{A}^{\sigma_{\beh P}}$ is given by internal completeness.
  \item By Corollary~\ref{coro_kleene_sup} we have $\interpret{\mu X.A}^{\sigma} = (\bigcup_{n \in \mathbb N} (\phi^{A}_{\sigma})^n(\daimon))^{\perp\perp}$, and  by Theorem~\ref{thm_union_beh} we have $(\bigcup_{n \in \mathbb N} (\phi^{A}_{\sigma})^n(\daimon))^{\perp\perp} = \bigcup_{n \in \mathbb N} (\phi^{A}_{\sigma})^n(\daimon)$ since items (2) and (3) guarantee that the hypotheses of the theorem are satisfied.
  \item By previous item and Lemma~\ref{0lem_incarn_union}(3).
  \end{enumerate}
If $A = \mu Y.A_0$ then:
  \begin{enumerate}
    \item Suppose $|\beh P| \subseteq |\beh P'|$, where $\beh P$ and $\beh P'$ are simple regular.
      We have $|{\phi^{A}_{\sigma}}(\beh P)| = |\interpret{\mu Y.A_0}^{\sigma_{\beh P}}| = \bigcup_{n \in \mathbb N} |(\phi^{A_0}_{\sigma_{\beh P}})^n(\daimon)|$ by induction hypothesis (5), and similarly for $\beh P'$. By induction on $n$, we prove that
      \[|(\phi^{A_0}_{\sigma_{\beh P}})^n(\daimon)| \subseteq |(\phi^{A_0}_{\sigma_{\beh P'}})^n(\daimon)| \tag{$\delta$}\]
      It is immediate for $n = 0$, and the inductive case is:
      \begin{align*}
        |(\phi^{A_0}_{\sigma_{\beh P}})^{n+1}(\daimon)| & = |\phi^{A_0}_{\sigma_{\beh P}}((\phi^{A_0}_{\sigma_{\beh P}})^n(\daimon))| & \\
        & \subseteq |\phi^{A_0}_{\sigma_{\beh P}}((\phi^{A_0}_{\sigma_{\beh P'}})^n(\daimon))| & \mbox{ by induction hypotheses (1), (3) and ($\delta$)} \\
        & = |\phi^{A_0}_{\sigma, Y \mapsto (\phi^{A_0}_{\sigma_{\beh P'}})^n(\daimon)}(\beh P)|\\
        & \subseteq |\phi^{A_0}_{\sigma, Y \mapsto (\phi^{A_0}_{\sigma_{\beh P'}})^n(\daimon)}(\beh P')| & \mbox{ by induction hypotheses (1) and (3)} \\
        & = |(\phi^{A_0}_{\sigma_{\beh P'}})^{n+1}(\daimon)|
      \end{align*}
      \setcounter{enumi}{2}
    \item By induction hypotheses (2), (3) and (4) respectively, we have
      \begin{itemize}
      \item for all $n \in \mathbb N$, $|(\phi^{A_0}_{\sigma})^n(\daimon)| \subseteq |(\phi^{A_0}_{\sigma})^{n+1}(\daimon)|$,
      \item for all $n \in \mathbb N$, $(\phi^{A_0}_{\sigma})^n(\daimon)$ is simple regular,
      \item $\interpret{\mu Y.A_0}^{\sigma} = \bigcup_{n \in \mathbb N}(\phi^{A_0}_{\sigma})^n(\daimon)$. 
      \end{itemize}
      Consequently, by Corollary~\ref{coro_reg_pur}, $\interpret{\mu Y.A_0}^{\sigma}$ is simple regular.
      \setcounter{enumi}{1}
    \item 4. 5. Similar to the cases $A = A_1 \oplus^+ A_2$ and $A = A_1 \otimes^+ A_2$.
  \end{enumerate}
\end{proof}

\begin{proof}[Proof (Lemma~\ref{lem_cb_basis})]
  By induction on $A$:
  \begin{itemize}
  \item If $A = a$ then it has basis $\interpret{a} = \beh C_a$.
  \item If $A = A_1 \oplus^+ A_2$, without loss of generality suppose $A_1$ is steady, with basis $\beh B_1$. Take $\otimes_1 \shneg \beh B_1$, as a basis for $A$, where the connective $\otimes_1$ is defined like $\shpos$ with a different name of action: $\otimes_1 \beh N :=\symplus_1\langle\beh N\rangle^{\perp\perp}$ and by internal completeness $\otimes_1 \beh N :=\symplus_1\langle\beh N\rangle$.
  \item If $A = A_1 \otimes^+ A_2$ then both $A_1$ and $A_2$ are steady, of respective base $\beh B_1$ and $\beh B_2$. The behaviour $\beh B = \beh B_1 \otimes^+ \beh B_2$ is a basis for $A$, indeed: since $\beh B_1$ and $\beh B_2$ are regular, Proposition~\ref{visit-tensor-reg} gives $V_{\beh B_1 \otimes^+ \beh B_2} = \kappa_\symtensor(V_{\shneg \beh B_1}^x \shuffle V_{\shneg \beh B_2}^y) \cup \{\daimon\}$ where, by Proposition~\ref{visit-sh}, $V_{\shneg \beh B_i} = \kappa_\symshneg V_{\beh B_i}^x \cup \{\epsilon\}$ for $i \in \{1, 2\}$; from this, and using internal completeness, we deduce that $\beh B$ satisfies all the conditions.
  \item Suppose $A = \mu X.A_0$, where $A_0$ is steady and has a basis $\beh B_0$, let us show that $\beh B_0$ is also a basis for $A$.
    \begin{itemize}
    \item By Proposition~\ref{prop_reg_union}, $\interpret{A}^{\sigma} = \bigcup_{n \in \mathbb N}(\phi^{A_0}_{\sigma})^n(\daimon)$, and since $\beh B_0$ is a basis for $A_0$ we have $\beh B_0 \subseteq \interpret{A_0}^{\sigma, X \to \daimon} = (\phi^{A_0}_{\sigma})(\daimon)$, so indeed $\beh B_0 \subseteq \interpret{A}^{\sigma}$.
      \item By induction hypothesis, we immediately have that for every path $\pathLL s \in V_{\beh B_0}$, there exists $\pathLL t \in V_{\beh B_0}^{max}$ $\daimon$-free extending $\pathLL s$.
    \item By Lemma~\ref{0lem_visit_union}(2) $V_{\interpret{A}^{\sigma}} = \{\daimon\} \cup \bigcup_{n \in \mathbb N} V_{(\phi^{A_0}_{\sigma})^{n+1}(\daimon)} = \{\daimon\} \cup \bigcup_{n \in \mathbb N} V_{\interpret{A_0}^{\sigma_n}}$ where $\sigma_n = \sigma, X \mapsto (\phi^{A_0}_{\sigma})^n(\daimon)$ has a simple regular image. By induction hypothesis, for all $n \in \mathbb N$, $V_{\beh B}^{max} \subseteq V_{\interpret{A_0}^{\sigma_n}}^{max}$, therefore $V_{\beh B}^{max} \subseteq V_{\interpret{A}^{\sigma}}^{max}$.
    \end{itemize}
  \end{itemize}
\end{proof}

\section{Proof of Proposition~\ref{prop_main}} \label{wpf}

In this section, we prove Proposition~\ref{prop_main}, which requires first several lemmas. Let us denote the set of functional behaviours by $\mathcal F$, and recall that $\mathcal D$ stands for the set of data behaviours.

\begin{lemma} \label{lem_datafunc_pure}
  Let $\beh P \in \mathcal D$, and let $\beh Q$ be a pure regular behaviour. The behaviour $\beh P \multimap^+ \beh Q$ is pure.
\end{lemma}

\begin{proof}
  By Proposition~\ref{prop_pure_stable} it suffices to show that $(\shneg \beh P) \multimap \beh Q$ is pure. Remark first that, by construction of data behaviours, the following assertion is satisfied in views (thus also in paths) of $\shneg \beh P$: every proper positive action is justified by the negative action preceding it.
  
By regularity and Corollary~\ref{visit-arrow-reg}, we have $\visit{(\shneg P) \multimap Q} = \dual{\kappa_\symtensor (\visit{\shneg P} \shuffle \dual{\visit{Q}})} \cup \{\epsilon\}$. Let $\pathLL s\daimon \in \visit{(\shneg P) \multimap Q}$, and we will prove that it is extensible. There exist $\pathLL t_1 \in \visit{\shneg P}$ and $\pathLL t_2 \in \visit{Q}$ such that $\dual{\pathLL s\daimon} = \overline{\pathLL s} \in \kappa_\symtensor(\pathLL t_1 \shuffle \dual{\pathLL t_2})$. In particular $\pathLL t_1$ is $\daimon$-free and $\pathLL t_2$ is $\daimon$-ended, say $\pathLL t_2 = \pathLL t_2'\daimon$. Since $\beh Q$ is pure, there exists $\kappa^+$ such that $\pathLL t_2'\kappa^+ \in V_{\beh Q}$. Let us show that $\pathLL s\kappa^+$ is a path, i.e., that if $\kappa^+$ is justified then $\mathrm{just}(\kappa^+)$ appears in $\view{\pathLL s}$, by induction on the length of $\pathLL t_1$:
  \begin{itemize}
  \item If $\pathLL t_1 = \epsilon$ then  $\pathLL s\kappa^+ = \pathLL t_2'\kappa^+$ hence it is a path.
  \item Suppose $\pathLL t_1 = \pathLL t_1'\kappa_p^-\kappa_p^+$. Since $\pathLL t_1$ is $\daimon$-free, $\kappa_p^+$ is proper. Thus $\pathLL s$ is of the form $\pathLL s = \pathLL s_1 \overline{\kappa_p^-\kappa_p^+}\pathLL s_2$, and by induction hypothesis $\pathLL s_1 \pathLL s_2 \kappa^+$ is a path, i.e., $\mathrm{just}(\kappa^+)$ appears in $\view{\pathLL s_1 \pathLL s_2}$.
    \begin{itemize}
    \item Either $\view{\pathLL s} = \view{\pathLL s_1 \pathLL s_2}$ and indeed $\mathrm{just}(\kappa^+)$ also appears in $\view{\pathLL s}$.
    \item Or $\view{\pathLL s}$ is of the form $\view{\pathLL s} = \view{\pathLL s_1}\overline{\kappa_p^-}\overline{\kappa_p^+}\pathLL s'_2$ since, by the remark at the beginning of this proof, $\kappa_p^+$ is justified by $\kappa_p^-$. This means in particular that $\pathLL s'_2$ start with the same positive action as $\pathLL s_2$, thus we have $\view{\pathLL s_1 \pathLL s_2} = \view{\pathLL s_1} \pathLL s'_2$. Since $\mathrm{just}(\kappa^+)$ appears in $\view{\pathLL s_1 \pathLL s_2}$ and it is an action of $\pathLL s_1$, it appears in $\view{\pathLL s_1}$ thus also in $\view{\pathLL s}$.
    \end{itemize}
  \end{itemize}
  Therefore $\pathLL s\kappa^+$ is a path. Since $\pathLL s \kappa^+ \in \dual{\kappa_\symtensor (\visit{\shneg P} \shuffle \dual{\visit{Q}})}$ and the behaviours are regular, $\pathLL s\kappa^+ \in \visit{P \multimap^+ Q}$, thus $\pathLL s\daimon$ is extensible. As this is true for every $\daimon$-ended path in $\visit{(\shneg P) \multimap Q}$, the behaviour $(\shneg \beh P) \multimap \beh Q$ is pure, and so is $\beh P \multimap^+ \beh Q$.
\end{proof}

\begin{lemma}\label{cons-arrow-pure}
  If $\beh P \in \mathcal F$ and $\beh Q \in \mathrm{Const}$ then $\beh P \multimap^+ \beh Q$ is pure.
\end{lemma}

\begin{proof}
We prove that $(\shneg \beh P) \multimap \beh Q$ is pure, and the conclusion will follow from Proposition~\ref{prop_pure_stable}. Let $\kappa^+ = \posdes{x_0}{a}{\vect{y}}$ where $\beh Q = \beh C_a$, and let $\pathLL s \daimon \in \visit{(\shneg P) \multimap Q}$. As in the proof of Lemma~\ref{lem_datafunc_pure}, there exist $\pathLL t_1 \in \visit{\shneg P}$ and $\pathLL t_2 \in \visit{Q}$ such that $\dual{\pathLL s\daimon} = \overline{\pathLL s} \in \kappa_\symtensor(\pathLL t_1 \shuffle \dual{\pathLL t_2})$ with $\pathLL t_2$ $\daimon$-ended. But $\visit{Q} = \{\daimon, \kappa^+\}$, thus $\pathLL t_2 = \daimon$ and $\dual{\pathLL t_2} = \epsilon$. Hence $\pathLL s \daimon = \dual{\kappa_\symtensor\pathLL t_1}$, and this path is extensible with action $\kappa^+$, indeed: $\pathLL s\kappa^+$ is a path because $\kappa^+$ is justified by $\kappa_\symtensor$, which is the only initial action of $\pathLL s\kappa^+$ thus appearing in $\view{\pathLL s}$; moreover $\dual{\pathLL s\kappa^+} \in \kappa_\symtensor(\pathLL t_1 \shuffle \dual{\kappa^+})$ where $\kappa^+ \in \visit Q$, therefore $\pathLL s \kappa^+ \in \visit{(\shneg P) \multimap Q}$.
\end{proof}

\begin{lemma} \label{arrow-max}
Let $\beh P, \beh Q \in \mathcal F$. If there is $\pathLL s \in \visit{Q}$ $\daimon$-free (resp. $\daimon$-ended) and maximal, then there is $\pathLL t \in \visit{P \multimap^+ Q}$ $\daimon$-free (resp. $\daimon$-ended) and maximal.
\end{lemma}

\begin{proof}
  Suppose there exists $\pathLL s \in \visit{Q}$ $\daimon$-free (resp. $\daimon$-ended) and maximal.
  Since $\beh P$ is positive and different from $\daimon$, there exists $\pathLL s' \in \visit{\shneg P}$ $\daimon$-free and non-empty. Let $\pathLL t' = \dual{\kappa_\symtensor\pathLL s'\dual{\pathLL s}}$, and remark that $\pathLL t' = \overline{\kappa_\symtensor\pathLL s'}\pathLL s$. This is a path (O- and P-visibility are satisfied), it belongs to $\visit{(\shneg P) \multimap Q}$, it is $\daimon$-free (resp. $\daimon$-ended). Suppose it is extensible, and consider both the ``$\daimon$-free'' and the ``$\daimon$-ended'' cases:
  \begin{itemize}
  \item if $\pathLL s$ and $\pathLL t'$ are $\daimon$-free, then there exists a negative action $\kappa^-$ such that $\pathLL t'\kappa^-\daimon \in \visit{(\shneg P) \multimap Q}$. But $\pathLL t'\kappa^-\daimon = \overline{\kappa_\symtensor\pathLL s'}\pathLL s\kappa^-\daimon$, and since it belongs to $\visit{(\shneg P) \multimap Q} = \dual{\kappa_\symtensor (\visit{\shneg P} \shuffle \visit{Q^\perp})} \cup \{\epsilon\}$, we necessarily have $\pathLL s\kappa^-\daimon \in \visit{Q}$ -- indeed, the sequence $\overline{\pathLL s'}\kappa^-$ has two adjacent negative actions. This contradicts the maximality of $\pathLL s$ in $\visit{Q}$.
  \item if $\pathLL s$ and $\pathLL t'$ are $\daimon$-ended, there exists a positive action $\kappa^+$ that extends $\pathLL t'$ and a contradiction arises with a similar reasoning.
  \end{itemize}
Hence $\pathLL t'$ is maximal in $\visit{(\shneg P) \multimap Q}$. Finally, $\pathLL t = \kappa_\symshpos \pathLL t'$ fulfills the requirements.\end{proof}

\begin{lemma}\label{max-df}
  For every behaviour $\beh P \in \mathcal F$, there exists $\pathLL s \in \visit{P}$ maximal and $\daimon$-free.
\end{lemma}

\begin{proof}
By induction on $\beh P$. If $\beh P \in \mathcal D$ then take $\pathLL s \in \visit B$ maximal, where $\beh B$ is a base of $\beh P$. Use Lemma~\ref{arrow-max} in the case of $\multimap^+$, and the result is easy for $\otimes^+$ and $\oplus^+$.
\end{proof}

\begin{lemma} \label{contra-pur}
  Let $\beh P \in \mathcal F$ and let $\mathcal C$ be a context. If $\mathcal C[\beh P]$ pure then $\beh P$ pure.
\end{lemma}

\begin{proof}
  We prove the contrapositive by induction on $\mathcal C$. Suppose $\beh P$ is impure.
  \begin{itemize}
  \item If $\mathcal C = [~]$ then $\mathcal C[\beh P] = \beh P$, thus $\mathcal C[\beh P]$ is impure.
  \item If $\mathcal C = \mathcal C' \oplus^+ \beh Q$ or $\beh Q \oplus^+ \mathcal C'$ and by induction hypothesis $\mathcal C'[\beh P]$ is impure, i.e., there exists a maximal path $\pathLL s\daimon \in \visit{\mathcal C'[P]}$, then one of $\kappa_{\symplus_1}\kappa_\symshneg \pathLL s\daimon$ or $\kappa_{\symplus_2}\kappa_\symshneg \pathLL s\daimon$ is maximal in $\visit{\mathcal C[P]}$, hence the result.
  \item If $\mathcal C = \mathcal C' \otimes^+ \beh Q$ or $\beh Q \otimes^+ \mathcal C'$ and by induction hypothesis there exists a maximal path $\pathLL s\daimon \in \visit{\mathcal C'[P]}$, then by Lemma~\ref{max-df}, there exists a $\daimon$-free maximal path $\pathLL t \in \visit{Q}$. Consider the path $\pathLL u = \kappa_\symtensor \kappa_\symshneg^{\pathLL t} \pathLL t \kappa_\symshneg^{\pathLL s} \pathLL s\daimon$, where:
    \begin{itemize}
    \item $\kappa_\symshneg^{\pathLL t}$ justifies the first action of $\pathLL t$,
    \item $\kappa_\symshneg^{\pathLL s}$ justifies the first one of $\pathLL s$, and
    \item $\kappa_\symtensor$ justifies $\kappa_\symshneg^{\pathLL t}$ and $\kappa_\symshneg^{\pathLL s}$, one on each (1\textsuperscript{st} or 2\textsuperscript{nd}) position, depending on the form of $\mathcal C$.
    \end{itemize}
    We have $\pathLL u \in \visit{\mathcal C[\beh P]}$, and $\pathLL u$ is $\daimon$-ended and maximal, hence the result.
  \item If $\mathcal C = \beh Q \multimap^+ \mathcal C'$ and by induction hypothesis $\mathcal C'[\beh P]$ is impure, then Lemma~\ref{arrow-max} (in its ``$\daimon$-ended'' version) concludes the proof. 
  \end{itemize}
\end{proof}

\begin{proof}[Proof (Proposition~\ref{prop_main})]
\noindent $(\Rightarrow)$ Suppose $\beh P$ impure. By induction on behaviour $\beh P$:
\begin{itemize}
\item $\beh P \in \mathcal D$ impossible by Corollary~\ref{coro_data_pur}.
\item If $\beh P = \beh P_1 \oplus^+ \beh P_2$ (resp. $\beh P = \beh P_1 \otimes^+ \beh P_2$) then one of $\beh P_1$ or $\beh P_2$ is impure by Proposition~\ref{prop_pure_stable}, say $\beh P_1$. By induction hypothesis, $\beh P_1$ is of the form $\beh P_1 = \mathcal C'_1[~\mathcal C'_2[\beh Q_1 \multimap^+ \beh Q_2] \multimap^+ \beh R~]$. Let $\mathcal C_1 = \mathcal C'_1 \oplus^+ \beh P_2$ (resp. $\mathcal C_1 = \mathcal C'_1 \otimes^+ \beh P_2$) and $\mathcal C_2 = \mathcal C'_2$, in order to get the result for $\beh P$.
\item If $\beh P = \beh P_1 \multimap^+ \beh P_2$, then $\beh P_2 \not \in \mathrm{Const}$ by Lemma~\ref{cons-arrow-pure}, and:
  \begin{itemize}
  \item If $\beh P_2$ impure, then by induction hypothesis $\beh P_2$ is of the form $\beh P_2 = \mathcal C'_1[~\mathcal C'_2[\beh Q_1 \multimap^+ \beh Q_2] \multimap^+ \beh R~]$, and it suffices to take $\mathcal C_1 = \beh P_1 \to \mathcal C'_1$ and $\mathcal C_2 = \mathcal C'_2$ to get the result for $\beh P$.
  \item If $\beh P_2$ is pure, since it is also regular the conclusion follows from Lemma~\ref{lem_datafunc_pure}.
  \end{itemize}
\end{itemize}

\noindent $(\Leftarrow)$ Let $\mathcal C_1, \mathcal C_2$ be contexts, $\beh Q_1, \beh Q_2, \beh R \in \mathcal P$ with $\beh R \not \in \mathrm{Const}$. Let $\beh P = \mathcal C_1[~\mathcal C_2[\beh Q_1 \multimap^+ \beh Q_2] \multimap^+ \beh R~]$ and $\beh Q = \mathcal C_2[\beh Q_1 \multimap^+ \beh Q_2]$. We prove that $\beh P$ is impure.

First suppose that $\beh P = \mathcal C_2[\beh Q_1 \multimap^+ \beh Q_2] \multimap^+ \beh R$, and in this case we show the result by induction on the depth of context $\mathcal C_2$. The exact induction hypothesis will be: there exists a maximal $\daimon$-ended path in $\visit P$ of the form $\kappa_\symshpos\pathLL s\daimon$ where $\overline{\pathLL s} \in \kappa_\symtensor((\kappa_\symshneg\visit Q) \shuffle \dual{\visit R})$.
\begin{itemize}

\item If $\mathcal C_2 = [~]$, then $\beh Q = \beh Q_1 \multimap^+ \beh Q_2 = \shpos(\shneg \beh Q_1 \multimap \beh Q_2)$ and $\beh P = \beh Q \multimap^+ \beh R = \shpos (\shneg \beh Q \multimap \beh R)$. In order to differentiate actions $\kappa_\symshpos, \kappa_\symshneg, \kappa_\symtensor$ used to construct $\beh Q$ from those to construct $\beh P$, we will use corresponding superscripts. Let $\kappa^Q_\symshneg\pathLL t_1 \in \visit{\shneg Q_1}$ be $\daimon$-free (and non-empty). Let $\pathLL t_2 \in V_{\beh Q_2}$ be a maximal $\daimon$-free path: its existence is ensured by Lemma~\ref{max-df}, and it has one proper positive initial action $\kappa_2^+$. Now let $\pathLL t = \dual{\kappa^Q_\symtensor\kappa^Q_\symshneg\pathLL t_1\dual{\pathLL t_2}} = \overline{\kappa^Q_\symtensor\kappa^Q_\symshneg\pathLL t_1}\pathLL t_2$. Similarly to the path constructed in proof of Lemma~\ref{arrow-max}, we have that $\pathLL t$ is $\daimon$-free, it is in $\visit{(\shneg Q_1) \multimap Q_2}$, and it is maximal. Thus $\kappa^Q_\symshpos\pathLL t \in \visit{Q}$. Since $\beh R \notin \mathrm{Const}$, there exists a path of the form $\kappa^+\kappa^-\daimon \in \visit{R}$, and thus necessarily $\kappa^+$ justifies $\kappa^-$. Define the sequence:
  \[\pathLL s\daimon = \overline{\kappa^P_\symtensor\kappa^P_\symshneg\kappa^Q_\symshpos}\kappa^Q_\symtensor\kappa^Q_\symshneg \kappa^+ \kappa^-\pathLL t_1\overline{\pathLL t_2}\daimon\]
  and notice the following facts:
  \begin{enumerate}
  \item \underline{$\pathLL s\daimon$ is a path}: it is a linear aj-sequence. Since $\kappa^-$ is justified by $\kappa^+$, O- and P-visibility are easy to check.
  \item \underline{$\pathLL s\daimon \in V_{\shneg \beh Q \multimap \beh R}$}: indeed, we have $\dual{\pathLL s\daimon} \in \kappa^P_\symtensor(\kappa^P_\symshneg\kappa^Q_\symshpos\pathLL t \shuffle \dual{\kappa^+\kappa^-\daimon})$ where $\kappa^P_\symshneg\kappa^Q_\symshpos\pathLL t \in \visit{\shneg Q}$ and $\kappa^+\kappa^-\daimon \in \visit R$.
  \item \underline{$\pathLL s\daimon$ is maximal}: Let us show that $\pathLL s\daimon$ is not extensible. First, it is not possible to extend it with an action from $\beh Q^\perp$, because this would contradict the maximality of $\pathLL t$ in $\visit{Q}$. Suppose it is extensible with an action $\kappa^{+\prime}$ from $\beh R$, that is $\pathLL s \kappa^{+\prime} \in \visit{\shneg Q \multimap R}$ and $\dual{\pathLL s\kappa^{+\prime}} \in \kappa^P_\symtensor(\kappa^P_\symshneg\kappa^Q_\symshpos\pathLL t \shuffle \dual{\kappa^+\kappa^-\kappa^{+\prime}})$ where $\kappa^+\kappa^-\kappa^{+\prime} \in \visit R$. The action $\kappa^{+\prime}$ (that cannot be initial) is necessarily justified by $\kappa^-$. But $\view{\pathLL s}$ contains necessarily the first negative action of $\overline{\pathLL t_2}$, which is the only initial action in $\overline{\pathLL t_2}$, and this action is justified by $\kappa^Q_\symtensor$ in $\pathLL s$. Therefore $\view{\pathLL s}$ does not contain any action from $\pathLL s$ between $\kappa^Q_\symtensor$ and $\overline{\pathLL t_2}$, in particular it does not contain $\kappa^- = \mathrm{just}(\kappa^{+\prime})$. Thus $\pathLL s\kappa^{+\prime}$ is not P-visible: contradiction. Hence $\pathLL s\daimon$ maximal.
  \end{enumerate}
  Finally $\kappa^P_\symshpos\pathLL s \daimon \in \visit{P}$ is not extensible, and of the required form.

\item If $\mathcal C_2 = \beh Q_0 \multimap^+ \mathcal C$, then $\beh Q$ is of the form $\beh Q = \beh Q_0 \multimap^+ \beh Q'$, thus previous reasoning applies.
  
\item If $\mathcal C_2 = \mathcal C \otimes^+ \beh Q_0$ or $\beh Q_0 \otimes^+ \mathcal C$, the induction hypothesis gives us the existence of a maximal path in $\visit{\mathcal C[Q_1 \multimap^+ Q_2]\multimap^+ R}$ of the form $\kappa^P_\symshpos\overline{\kappa^P_\symtensor\kappa^P_\symshneg}\pathLL s'\daimon$ where $\kappa^P_\symshneg\overline{\pathLL s'} \in (\kappa^P_\symshneg\pathLL t') \shuffle \dual{\pathLL u}$ with $\pathLL t' \in \visit{\mathcal C[Q_1 \multimap^+ Q_2]}$ and $\pathLL u \in \visit R$.
  Let $\pathLL t_0 \in \visit{Q_0}$ be $\daimon$-free and maximal, using Lemma~\ref{max-df}. Consider the following sequence:
  \[\pathLL s \daimon = \overline{\kappa^P_\symtensor\kappa^P_\symshneg\kappa^Q_\symtensor\kappa^0_\symshneg\pathLL t_0\kappa^1_\symshneg}\pathLL s'\daimon\]
  where:
  \begin{itemize}
  \item $\kappa_\symshneg^0$ justifies the first action of $\pathLL t_0$, 
  \item $\kappa_\symshneg^1$ justifies the first action of $\overline{\pathLL s'}$ thus the first action of $\pathLL t'$,
  \item $\kappa^Q_\symtensor$ justifies $\kappa_\symshneg^0$ and $\kappa_\symshneg^1$,
  \item $\kappa_\symshneg^P$ now justifies $\kappa^Q_\symtensor$,
  \item $\kappa_\symtensor^P$ justifies the same actions as before.
  \end{itemize}
  Notice that:
  \begin{enumerate}
  \item \underline{$\pathLL s\daimon$ is a path}: O- and P-visibility are satisfied.
  \item \underline{$\pathLL s\daimon \in V_{\shneg \beh Q \multimap \beh R}$}:
    We have $\kappa^Q_\symtensor\kappa^0_\symshneg\pathLL t_0\kappa^1_\symshneg\overline{\pathLL t'} \in \kappa^Q_\symtensor(\kappa^0_\symshneg\visit{\beh Q_0} \shuffle \kappa^1_\symshneg\visit{\mathcal C[Q_1 \multimap^+ Q_2]}) = \visit{Q}$, hence $\dual{\pathLL s\daimon} \in \kappa^P_\symtensor(\visit{\shneg Q} \shuffle \dual{\visit R})$.
  \item \underline{$\pathLL s\daimon$ is maximal}: Indeed, it cannot be extended neither by an action of $\beh Q_0^\perp$ (contradicts the maximality of $\pathLL t_0$) nor by an action of $\mathcal C[\beh Q_1 \multimap^+ \beh Q_2]^\perp$ or $\beh R$ (contradicts the maximality of $\pathLL s'$).
  \end{enumerate}
  Finally $\kappa^P_\symshpos\pathLL s\daimon \in \visit{P}$ is a path satisfying the constraints.

\item If $\mathcal C_2 = \mathcal C \oplus^+ \beh Q_0$ or $\beh Q_0 \oplus^+ \mathcal C$, by induction hypothesis, there exists a path of the form $\kappa^P_\symshpos\overline{\kappa^P_\symtensor\kappa^P_\symshneg}\pathLL s'\daimon$ maximal in $\visit{\mathcal C[Q_1 \multimap^+ Q_2]\multimap^+ R}$, where $\kappa^P_\symshneg\overline{\pathLL s'} \in (\kappa^P_\symshneg\pathLL t') \shuffle \dual{\pathLL u}$ with $\pathLL t' \in \visit{\mathcal C[Q_1 \multimap^+ Q_2]}$ and $\pathLL u \in \visit R$. Reasoning as previous item, we see that for one of $i \in \{1, 2\}$ (depending on the form of context $\mathcal C_2$) the path $\kappa^P_\symshpos\overline{\kappa^P_\symtensor\kappa^P_\symshneg\kappa^Q_{\symplus_i}\kappa_\symshneg}\pathLL s'\daimon$ (where $\kappa_\symshneg^P$ now justifies $\kappa^Q_{\symplus_i}$) is in $\visit{P}$, maximal, and of the required form.
\end{itemize}
The result for the general case, where $\beh P = \mathcal C_1[~\mathcal C_2[\beh Q_1 \multimap^+ \beh Q_2] \multimap^+ \beh R~]$, finally comes from Lemma~\ref{contra-pur}.
\end{proof}

\end{document}